\documentclass[10pt]{article}
\usepackage{amssymb,amsmath,amsfonts,amsthm,enumerate,color}
\usepackage[toc,page,title,titletoc,header]{appendix}
\usepackage{graphicx}
\usepackage{indentfirst}
\usepackage{multicol}
\usepackage{booktabs}
\usepackage{url}

\numberwithin{equation}{section}

\newtheorem{Theorem}{Theorem}[section]
\newtheorem{Lemma}[Theorem]{Lemma}

\newtheorem{Definition}[Theorem]{Definition}
\newtheorem{Corollary}[Theorem]{Corollary}

\numberwithin{equation}{section}

 \def\p{\partial} 
\def \Vh0{\stackrel{\circ}{V}_h}

\newcommand{\q}{\quad}    
  \def\f{\frac}  

\def\m{\mbox}

\def\p{\partial}

\newcommand{\lc}
{\mathrel{\raise2pt\hbox{${\mathop<\limits_{\raise1pt\hbox
{\mbox{$\sim$}}}}$}}}

\newcommand{\gc}
{\mathrel{\raise2pt\hbox{${\mathop>\limits_{\raise1pt\hbox{\mbox{$\sim$}}}}$}}}

\newcommand{\ec}
{\mathrel{\raise2pt\hbox{${\mathop=\limits_{\raise1pt\hbox{\mbox{$\sim$}}}}$}}}

\def\bb{\begin{equation}} \def\ee{\end{equation}}

\def\beqn{\begin{eqnarray}}  \def\eqn{\end{eqnarray}}

\def\beqnx{\begin{eqnarray*}} \def\eqnx{\end{eqnarray*}}

\def\bn{\begin{enumerate}} \def\en{\end{enumerate}}

\def\bd{\begin{description}} \def\ed{\end{description}}

\makeatletter

\newenvironment{figurehere}
  {\def\@captype{figure}}
  {}
\makeatother

\title{Super-resolution in Imaging High Contrast Targets from the Perspective
of Scattering Coefficients\thanks{\footnotesize This work was supported  by the
ERC Advanced
Grant Project MULTIMOD--267184 and the Hong Kong RGC grants
(projects 405513 and 404611).}}

\author{
Habib Ammari\footnote{Department of Mathematics and Applications, Ecole Normale Sup$\acute{\text{e}}$rieure, 45 Rue d'Ulm, 75005 Paris, France.
(habib.ammari@ens.fr).}
\and Yat Tin Chow\footnote{Department of Mathematics, Chinese University of Hong Kong, Shatin, N.T., Hong Kong (ytchow@math.cuhk.edu.hk, zou@math.cuhk.edu.hk).}
\and Jun Zou\footnotemark[3]
}

\begin{document}

\date{}
\maketitle

\begin{abstract}
In this paper we consider the inverse scattering problem for
high-contrast targets. We mathematically analyze the experimentally-observed phenomenon of super-resolution in imaging the target shape.
This is the first time that a mathematical theory of super-resolution has been established in the context of
imaging high contrast inclusions.
We illustrate our main findings with a variety of numerical examples. Our analysis is based on the novel concept of scattering coefficients.
These findings may help in developing resonant structures for resolution enhancement.
\end{abstract}

\bigskip

\noindent {\bf Mathematics Subject Classification}
(MSC\,2000): 35B30, 35R30

\noindent {\bf Keywords}: inverse scattering, super-resolution,  scattering coefficients

\section{Introduction}
The aim of this work is to mathematically investigate the mechanism underlying the experimentally-observed phenomenon of super-resolution in reconstructing targets of high contrast from far-field measurements.  Our main focus is to explore the possibility of breaking the diffraction barrier from the far-field measurements using the novel concept of scattering coefficients \cite{lnm, homoscattering, han}.  This diffraction barrier, referred to as the Abbe-Rayleigh or the resolution limit, places a fundamental limit on the minimal distance at which we can resolve the shape of a target \cite{yves}. It applies only to waves that have propagated for a distance substantially larger than its wavelength \cite{gang1,gang2}.

Since the mid-20th century, several approaches have aimed at pushing this diffraction limits.
Resolution enhancement in imaging the target shape from far-field measurements can be achieved using sub-wavelength-scaled resonant media \cite{hai, LFL2011, lemoult11, lemoult10, lerosey07} and single molecule imaging \cite{science}.
Another innovative method to overcome the diffraction barrier has been proposed after some experimental observations in \cite{physicssuperresolution}.
In their work, resolution enhancement in shape reconstruction of the inclusion was experimentally shown when the contrast value is very high.
In the reconstructed images from far-field measurements, the observed resolution is smaller than half of the operating wavelength, $2\pi/k$, where $k$ is the wave number.
This encouraging observation suggests a possibility of breaking the resolution limit with high permittivity of the target.
It is therefore the purpose of this work to prove that the higher the permittivity of the target is, the higher the resolving power %or equivalently the resolution
is in imaging its shape.

For the transmission problem of a strictly convex domain, it was proved in \cite{convex} that
there exists an infinite sequence of complex resonant frequencies located at the upper half plane.
These resonances converge to the real axis exponentially fast, and the real part of these resonances correspond to the quasi-resonant modes introduced as in \cite{convex}.
Quasi-resonance occurs when the wavelength inside the inclusion is larger than the wavelength in the background media
and is such that it reaches the real part of one of these true resonant frequencies.
In this paper, we have shown, via the analysis of the shape derivative of the scattering coefficients,
that these resonant state of the inclusion actually has a signature in the far-field
and can be used for super-resolved imaging from far-field data.
To be more exact, we have proved that, in the shape derivative of the scattering coefficients for a circular domain,
there are simple poles at the complex resonant states,
and therefore peaks corresponding to the the real parts of these resonances.
Henceforth, as the material contrast increases to infinity
and is such that it is equal to the real part of a resonance,
the sensitivity in the scattering coefficients becomes large and super-resolution for imaging becomes possible.

Throughout this paper, we consider the following  scattering problem in $\mathbb{R}^2$,
\beqn
\bigg(\Delta+k^2\Big(1+q(x)\Big)\bigg)u = 0,
\label{scattering1}
\eqn
where $u$ is the total field, $q(x) > 0 $ is the contrast of the medium and $k$ is the wave number.

We consider an inclusion $D$ contained inside a homogeneous background medium, and assume that $D$ is an open bounded connected domain with a $\mathcal{C}^{1,\alpha}$-boundary for some $0< \alpha<1$.
Suppose that the function $q$ is of the form
\begin{equation}
q (x) = \varepsilon^* \chi_{D} (x),
\end{equation}
where $ \chi_{D} $ denotes the characteristic function of $D$ and $\varepsilon^*>0$ is a constant.
We shall always complement the system (\ref{scattering1}) by the physical outgoing Sommerfeld radiation condition:
\beqn
   \big| \f{\partial}{\partial |x|} u^s- i k u^s \big|= O(|x|^{-\f{3}{2}})\quad \text{ as }\quad |x| \rightarrow \infty \, ,
    \label{sommerfield}
\eqn
where $u^s:=u-u^i$ is the scattered field and $u^i$ is the incident field.
The solution $u$ to the system \eqref{scattering1}-\eqref{sommerfield} represents the total field due to the scattering from the inclusion $D$ corresponding to the incident field $u^i$.

Following the work of \cite{heteroscattering, homoscattering, han},  the scattering coefficients provide a powerful and efficient tool for shape classification of the target $D$. Therefore, we aim at exhibiting the mechanism underlying the super-resolution phenomenon experimentally-observed in \cite{physicssuperresolution} in terms of the scattering coefficients corresponding to high-contrast inclusions.

In \cite{homoscattering}, it is proved that the scattering coefficient of order $(n,m)$ decays very quickly as the orders $|n|$, $|m|$ increase. Nonetheless, it is shown in \cite{han} that the scattering coefficients can be stably reconstructed from
the far-field measurements by a least-squares method. The stability of the reconstruction in the presence of a measurement noise is analyzed and the resolving power of the reconstruction in terms of the signal-to-noise-ratio is estimated.
It is the purpose of this paper
to use the scattering coefficients to estimate the resolution limit for imaging high contrast targets from far-field measurements as function of the material contrast, and to prove that the higher the permittivity is inside the target, the better the resolution is for imaging its shape from far-field measurements.

In order to achieve this goal, in this work,
we first give a decay estimate of the scattering coefficients in arbitrary shaped domains, and then in the particular case of a circular domain.
Our estimate shows different behaviors of the scattering coefficients of different orders as the material contrast increases.
Then we provide a sensitivity analysis of the scattering coefficients, which clearly shows that, in the linearized case, the scattering coefficient of order $(n,m)$ of a circular domain contains information about the $(n-m)$-th Fourier mode of the shape perturbation.
Afterwards, we establish the asymptotic behavior of eigenvalues of an important family of integral operators closely related to the scattering coefficients.
Series representations of the scattering coefficients and their shape derivatives in the case of a circular domain are given based on this asymptotic behavior.
From these series representations, we prove that as the material contrast increases and moves close to the reciprocal of the eigenvalues, the shape derivatives of the scattering coefficients behave like simple poles.
This explains the better conditioning of the inversion process of higher Fourier modes of inclusions with large material contrast, and hence an enhanced resolution of reconstructing the perturbation using the scattering coefficients.
Numerical examples illustrate that the relative magnitudes of higher order scattering coefficients grow as the medium coefficients grow and move close to the reciprocals of the eigenvalues, therefore providing more information about the shape of the domain with a fixed signal-to-noise ratio.
Our approach provides a good and promising direction of understanding towards the super-resolution phenomenon for high-contrast targets.

This paper is organized as follows. In section \ref{sec2} we give a brief review of the concept of scattering coefficients. We also prove a fundamental expression of the scattering coefficients in terms of a family of important integral operators.  Sensitivity analysis of the scattering coefficients with a fixed contrast is then presented in section \ref{sec3}, which shows that the shape derivative can also be represented by the family of integral operators introduced in section \ref{sec2}. Section \ref{sec4_1} briefly recalls Riesz decomposition of compact operators.
Asymptotic behavior of eigenvalues and eigenfunctions of the introduced integral operators will be studied in section \ref{sec4_2}.  Section \ref{sec4_3} provides a series representation of the scattering coefficients and their shape derivative. A mathematical explanation of
 the super-resolution phenomenon is given. Numerical results are reported in section \ref{numerical} to
 illustrate the phenomenon of super-resolution as the material contrast increases.

\section{The concept of scattering coefficients and a fundamental expression} \label{sec2}

In this section, we estimate the behavior of the scattering coefficients.
Without loss of generality, from now on, we normalize the wave number $k$ in \eqref{scattering1} to be $k=1$ by a change of variables.

To begin with, we first recall the definition of
the scattering coefficients $W_{nm} (D, \varepsilon^*)$ from \cite{heteroscattering, homoscattering}.
For this purpose, we introduce the following several notions.
The fundamental solution $\Phi$  to the Helmholtz operator $\Delta + 1$ in two dimensions
satisfying
\beqn
(\Delta + 1)\Phi(x) = \delta_0(x),
\label{solutionfun}
\eqn
where $\delta_0$ is the Dirac mass at $0$, with the outgoing Sommerfeld radiation condition:
$$
    \big| \f{\partial}{\partial |x|}\Phi- i  \Phi \big| = O(|x|^{-\f{3}{2}}) \quad \text{ as }\quad |x| \rightarrow \infty \, ,
$$
is  given by
\beqn
    \Phi (x) =
    -\f{i}{4} H^{(1)}_0(|x|) \,,
    \label{fundamental}
\eqn
where $H^{(1)}_0$ is the Hankel function of the first kind of order zero.

Now, given an incident field $u^i$ satisfying the homogeneous Helmholtz equation, i.e.,
\beqn
\Delta u^i + u^i = 0\, ,
\eqn
the solution $u$ to (\ref{scattering1}) and (\ref{sommerfield})
can be readily represented  by the Lippmann-Schwinger equation as
\beqn
u(x)= u^i(x) - \varepsilon^* \int_{D}\Phi(x-y)u(y)dy \, , \quad x \in \mathbb{R}^2\,,
\label{lipsch}
\eqn
and hence, the scattered field reads
\beqn
u^s(x)= -  \varepsilon^* \int_{D}\Phi(x-y)u(y)dy \, , \quad x \in \mathbb{R}^2\,.
\label{lipsch2}
\eqn

Let $S_{\p D}$ be the single-layer potential defined by the kernel $\Phi( \, \cdot\, )$, i.e.,
\beqn
 S_{\p D}[\phi] (x) = \int_{\p D} \Phi( x-y ) \phi (y) \, ds(y) \,
\eqn
for $\phi \in L^2(\p D)$.
Let $S^{\sqrt{\varepsilon^* +1}}_{\p D}$ be the single-layer potential associated with the kernel $\Phi\left( \sqrt{1 +\varepsilon^*} \, ( \, \cdot\, ) \right)$.

\begin{Definition}
\label{defdefuse}
The scattering coefficient $W_{nm} (D, \varepsilon^*)$ for $n,m \in \mathbb{Z}$ is defined as follows:
\begin{equation}
W_{nm} (D, \varepsilon^*) =
\int_{\partial \Omega} J_{n}( r_x ) \, e^{-i n \theta_x} \phi_m (x) \, ds(x) \,,
\label{defint2}
\end{equation}
where $x= r_x (\cos \theta_x, \sin \theta_x)$ in polar coordinates and the weight function $\phi_m \in L^2 (\partial D)$ is such that the pair $( \phi_m, \psi_m  )\in L^2 (\partial D) \times L^2 (\partial D)$ satisfies the following system of integral equations:
\beqn
\begin{cases}
S^{\sqrt{\varepsilon^* +1}}_{\p D} [\phi_m] (x)- S_{\p D} [\psi_m] (x)= J_{m}(r_x) e^{i m \theta_x} \,, \\
 \f{\p}{\p \nu}S^{\sqrt{\varepsilon^* +1}}_{\p D}  [\phi_m]  \mid_- (x)  - \f{\p}{\p \nu} S_{\p D} [\psi_m]\mid_+(x) = \f{\p}{\p \nu} (J_{m}(r_x) e^{i m \theta_x} ).
\label{defint}
\end{cases}
\eqn

\end{Definition}
%For the sake of simplicity, we shall often write $W_{n,m}(D, \varepsilon^*) := W_{nm} [D, 1, \varepsilon^*+1, 1,1,1]$ for $n,m\in \mathbb{Z}$.
\noindent Here $+$ and $-$ in the subscripts respectively indicate the limit from outside $D$ and inside $D$ to $\partial D$ along the normal direction, and $\p / \p \nu$ denotes the normal derivative.

According to \cite{heteroscattering, homoscattering}, the scattering coefficients $W_{nm}(D, \varepsilon^*) $ are basically the Fourier coefficients of  the far-field pattern (scattering amplitude) which is $2\pi$-periodic function in two dimensions. The far-field pattern
$A_{\infty} (\widehat{d}, \widehat{x} )$, when the incident field is given by $e^{i
\widehat{d}\cdot x}$ for a unit vector $\widehat{d}$, is defined to be
$$
(u - u^i)(x) =  i e^{-\pi i/4} \frac{e^{i |x|}}{\sqrt{8\pi |x|}}  A_{\infty} ( \widehat{d}, \widehat{x}) + O(|x|^{-\frac{3}{2}}) \quad \mbox{as } |x|\rightarrow \infty,
$$
with $\widehat{x} := x/|x|$.  We have, recalling from \cite{heteroscattering, homoscattering}, that
\beqn
W_{nm} (D, \varepsilon^*) = i^{n-m} \mathfrak{F}_{{\theta}_d, \theta_x} [A_{\infty} (\widehat{d}, \widehat{x} )] (-m,n),
\label{fourierfourier}
\eqn
where $\widehat{x}=(\cos \theta_x, \sin \theta_x)$ and $\widehat{d}= (\cos \theta_d, \sin \theta_d)$ in polar coordinates  and $ \mathfrak{F}_{{\theta}_d, \theta_x} [A_{\infty} (\widehat{d}, \widehat{x} )] (m,n)$ denotes the $(m,n)$-th Fourier coefficient of the far-field pattern $A_{\infty} (\widehat{d}, \widehat{x})$.

Our first objective is then to work out an explicit relation between the far-field pattern and the contrast $\varepsilon^*$ so as to obtain the behavior of the scattering coefficients when $\varepsilon^*$ is large.

In view of \eqref{lipsch}, we introduce the following operator for the subsequent analysis.
\begin{Definition} \label{def1}
The operator $\widetilde{K}_{D} : L^2(D) \rightarrow L^2(D)$ is defined by
\begin{equation}
\widetilde{K}_{D}[\phi] (x)=
\int_{D}\Phi(x-y)\phi(y)\,dy\,, \quad \text{ for }  x \in D \, \text{ and } \phi\in L^2(D) \,;
\label{defK}
\end{equation}
whereas, the operator $\widetilde{\widetilde{K}}_{D} : L^2(D) \rightarrow L^\infty(\mathbb{R}^2)$ is given by
\begin{equation}
\widetilde{\widetilde{K}}_{D} [\phi] (x)=
\int_{D}\Phi(x-y)\phi(y)\,dy\,,
\quad \text{ for }  x \in  \mathbb{R}^2  \, \text{ and } \phi\in L^2(D) \,.
\label{defKext}
\end{equation}
\end{Definition}
\noindent It is easy to see from the definition of $\widetilde{K}_{D}$ and the Rellich lemma that $\widetilde{K}_{D}$ is a compact operator.
However, it is worth emphasizing that $\widetilde{K}_{D}$ is not a normal operator in $L^2 (D)$. Therefore, it is not unitary equivalent to a multiplicative operator.
With Definition \ref{def1}, we can rewrite \eqref{lipsch} as
\beqn
(I + \varepsilon^* \widetilde{K}_D )[u](x) = u^i (x) \, , \quad  \forall x \in D \,,
\eqn
hence in $L^2(D)$,
\beqn
u  = ( I + \varepsilon^* \widetilde{K}_D  )^{-1} [u^i] \,.
\label{happy}
\eqn
From the well-known fact that
\beqn
    \Phi(x-y) = -\f{i}{4} H^{(1)}_0(|x-y|) =  - i e^{- \pi i/{4}} {\f{e^{i |x| - i \widehat{x} \cdot y  }}{ \sqrt{8\pi |x|}}}
      + O(|x|^{-\f{3}{2}})
    \q \m{as} \q |x| \rightarrow \infty \,,
    \label{fundamentalfarfield1}
\eqn
%where $\widehat{x} = x / |x|$,
we have
\beqn
   u^s(x) = - \varepsilon^* \int_D \Phi(x-y) u(y) \,dy = i \varepsilon^*  e^{- {\pi} i/{4}} {\f{ e^{i |x|}}{\sqrt{8 \pi |x|}}}  \int_D e^{ - i \widehat{x} \cdot y } u(y) \,dy  + O(|x|^{-\f{3}{2}})
    \q \m{as} \q |x| \rightarrow \infty.
    \label{fundamentalfarfield2}
\eqn
Therefore, the far-field of the scattered field can be written as
\beqn
A_{\infty} (\theta_d, \theta_x):= A_{\infty} (\widehat{d}, \widehat{x} ) = \varepsilon^*   \int_{D}  e^{ - i \widehat{x} \cdot y } u (y) dy  \, .
\label{farfield2}
\eqn
Recall the following well-known Jacobi-Anger identity \cite{Watson} for any unit vector $\widehat{d}$,
\beqn
e^{ - i  \widehat{d} \cdot x }  =  \sum_{n=-\infty}^{\infty} (-i)^n J_n ( r ) e^{i n ( \theta_d - \theta)} \,
\label{haha}
\eqn
for $x=(r,\theta)$ in polar coordinates.
Using (\ref{haha}) and taking the Fourier transform with respect to $\theta_x$, we get
\beqn
\mathfrak{F}_{\theta_x} [A_{\infty}] (n) =  (-i)^n \varepsilon^*  \langle J_n ( r) e^{i n  \theta},  u \rangle_{L^2(D)} = i^{-n} \left \langle J_n ( r ) e^{i n  \theta},  ( {\varepsilon^*}^{-1} +  \widetilde{K}_D  )^{-1} [u^i]  \right \rangle_{L^2(D)}   \, .
\label{farfield23}
\eqn
%where $\langle \, , \, \rangle_{L^2(D)}$ denotes the $L^2$-scalar product on $D$.

Now using $u^i(x)= e^{i \widehat{d}\cdot x}$, it follows from \eqref{fourierfourier} and \eqref{haha}-\eqref{farfield23} that the following theorem holds:
%\beqn
%W_{nm} (D,\varepsilon^*)= i^{(n-m)} \mathfrak{F}_{\theta_k, \theta_x} [A(\theta_k, \theta_x)] (-m,n)
%=    \left\langle J_n ( r) e^{i n  \theta},  ( {\varepsilon^*}^{-1} + \tilde{K}_D )^{-1}[J_m ( r) e^{i m  \theta}] \right\rangle_{L^2(D)} \,.
%\eqn
\begin{Theorem}\label{rfundamental}
For a domain $D$ and a contrast  $\varepsilon^*$, the scattering coefficient $W_{nm}(D,\varepsilon^*)$ for $n,m \in \mathbb{Z}$ can be written in the following form
\beqn
W_{nm} (D,\varepsilon^*) = i^{(n-m)} \mathfrak{F}_{\theta_d, \theta_x} [A(\theta_d, \theta_x)] (-m,n) =  \left\langle J_n ( r) e^{i n  \theta},  ( {\varepsilon^*}^{-1} + \tilde{K}_D )^{-1}[J_m ( r) e^{i m  \theta}] \right\rangle_{L^2(D)},
\label{fundamentalexpression}
\eqn
where $\tilde{K}_D$ is defined by \eqref{defK}.
\end{Theorem}

The expression \eqref{fundamentalexpression} of the scattering coefficients $W_{nm}$ will be fundamental to the analysis of the behavior of $W_{nm} $ with respect to $\varepsilon^*$.

Using (\ref{fundamentalexpression}), we can readily obtain an \textit{a priori} estimate for the coefficients $W_{nm}$. Let us first recall the following facts on Schatten-von Neumann ideals; see, for example, \cite{ideal}. Given a Hilbert space $H$, we let $\mathfrak{B}(H)$ to be the set of bounded operators on $H$.
We denote by $S_{\infty}(H)$ the closed two-sided ideal of compact operators in $\mathfrak{B}(H)$. For $K \in S_{\infty}$ and $k \in \mathbb{N}$, let the $k$-th singular number $s_k(K)$ be defined as the $k$-th eigenvalue of $|K| = \sqrt{K^*K}$ ordered in descending order of magnitude and being repeated according to its multiplicity, written as $s_k(K) := \lambda_k(|K|)$. Now, for $0 < p \leq \infty$, we shall often write the following Schatten-von Neumann quasi-norms (which are norms if $1 \leq p \leq \infty$) as follows :
\beqn
|| K ||_{S_{p}(H)} := \left( \sum_{k=1}^{\infty} s_k(K)^p \right)^{1/p} \text{ for } p < \infty \,; \quad || K ||_{S_{\infty}(H)} := ||K||_{H} \,,
\label{vonneumannnorm}
\eqn
whenever they are finite. %, and $|| K ||_{S_{\infty}(H)} := ||K||_{H}$ by convention.
Now let the Schatten-von Neumann quasi-normed operator ideal $S_{p}(H)$ be defined by
\beqn
S_p(H) := \left\{ K \in S_{\infty} \, : \, || K ||_{S_{p}(H)}  < \infty \right\} \,.
\eqn
Note that with this convention, $S_{1}(H)$ is the well-known trace class, $S_{2}(H)$ is the usual Hilbert-Schmidt class, and $S_{\infty}(H)$ is the usual class of compact operators in $H$. Moreover, if $H = L^2(D)$ and $K \in S_{2}(H)$ is the integral operator defined by
\beqn
K[f] (x)= \int_D K(x,y) f(y) \,dy ,  \quad \text{ for } x \in {D} \text{ and } f \in L^2(D) \,,
\eqn
then it holds that%the square of Hilbert-Schmidt norm $|| K ||_{S_{2}(L^2(D))} $ is equivalent to the following integral
\beqn
|| K ||_{S_{2}(L^2(D))}^2 = \int_D  \int_D | K(x,y) |^2 \, dx \, dy,
\eqn
which is always well-defined for any $K \in S_{2}(L^2(D))$. We refer the reader to, for example, \cite{ideal} for more properties concerning the Schatten-von Neumann ideals.

For a compact operator $K$, let  $\sigma(K) := \{ \lambda \in \mathbb{C} | \,  \lambda  - K \text { is singular} \} $  denote its spectrum and $( z  - K)^{-1}$ its resolvent operator whenever $ z \in \mathbb{C} \backslash \sigma(K)$.
Now, we have the following resolvent estimate \cite{normofresolvent}.
\begin{Theorem} \label{thmnorm}
For $0 < p < \infty$ and $K \in S_{p}(H)$, we have the following estimate for the resolvent operator $\left(z - K\right)^{-1}$ that
\beqn
\bigg|\bigg| \left(z - K\right)^{-1} \bigg|\bigg|_{H}
\leq \f{1}{d( z , \sigma (K) ) } \exp \left(\,
a_p \, \f{|| K ||_{S_{p}(H)} ^p }{ d( z , \sigma (K) )^p } + b_p \,
\right) \,,
\label{resolventineq}
\eqn
where $a_p$, $b_p$ are two constants depending on $p$ %, $|| K ||_{S_{p}(H)} $ is the Schatten-von Neumann norm defined as in \eqref{vonneumannnorm},
and
$d( z , \sigma (K) )$ is defined by
\beqn
d( z , \sigma (K) ) := \inf_{\lambda\in\sigma(K)} | z - \lambda|  \,.
\eqn
\end{Theorem}
% however for general case, this estimate still holds true with $a_2 = 2, b_2 =1/2 $ %\cite{normofresolvent}.

Now we can apply Theorem \ref{thmnorm} to get an estimate for $W_{nm}( D, \varepsilon^*)$. In fact, with the logarithmic type singularity of the function $H^{(1)}_0$, we readily obtain that
\beqn
|| \tilde{K}_D ||_{S_{2}(L^2(D))}^2 = \int_D  \int_D | H^{(1)}_0(| x-y|) |^2 \, dx \, dy  < C \left(1+R\right)^4 \left( 1+ \log R \right)^2 < \infty
\eqn
whenever $D \subset B(0,R)$, and hence $\tilde{K}_D \in S_{2}$. Therefore, using the Cauchy-Schwartz inequality and applying \eqref{resolventineq} for $H = L^2(D)$ to \eqref{fundamentalexpression}, together with
the following well-known asymptotic expression of $J_m$ for large $m$ \cite[pp.~365-366 ]{handbook},
\beqn
J_m(z)  \bigg \slash \f{1}{\sqrt{2 \pi m}} \left(\frac{e z}{2 m}\right)^m \rightarrow 1 \quad \text{ as } m \rightarrow \infty \,,
\label{asymptotics}
\eqn
we readily obtain the following inequality (using that $a_2 = 1/2, b_2 =1/2 $  if $p=2$ \cite{carleman}):
\beqn
|W_{nm} (D,\varepsilon^*)| &=&  \left| \left\langle J_n ( r) e^{i n  \theta},  ( {\varepsilon^*}^{-1} + \tilde{K}_D )^{-1}[J_m ( r) e^{i m  \theta}] \right\rangle_{L^2(D)} \right| \notag\\
&\leq&
 \left|\left|  ( {\varepsilon^*}^{-1} + \tilde{K}_D )^{-1} \right|\right|_{L^2(D)}  \left|\left| J_n ( r) e^{i  n \theta} \right|\right|_{L^2(D)} \, \left|\left| J_m ( r) e^{i m   \theta} \right|\right|_{L^2(D)}
\notag \\
&\leq&
\f{1}{d( - {\varepsilon^*}^{-1} , \sigma (\tilde{K}_D) ) } \exp \left( \,
 \, \f{|| \tilde{K}_D ||_{S_{2}(L^2(D))} ^2 }{ d( - {\varepsilon^*}^{-1}  , \sigma (\tilde{K}_D) )^2 } + \f{1}{2}
\, \right)
\left|\left| J_n ( r) e^{i  n \theta} \right|\right|_{L^2(D)} \, \left|\left| J_m ( r) e^{i m   \theta} \right|\right|_{L^2(D)}
\notag \\
& \leq &
\f{1}{d( - {\varepsilon^*}^{-1} , \sigma (\tilde{K}_D) ) } \exp \left( \,
 \, \f{C_{1,R} }{ d( - {\varepsilon^*}^{-1}  , \sigma (\tilde{K}_D) )^2 } + \f{1}{2}
\,\right) \frac{C_{2,R}^{|m|+|n|}}{{|m|}^{|m|} {|n|}^{|n|}} ,\notag
\eqn
where $C_{i,R} \, ( i = 1,2 )$ are some constants, which depend only on the radius $R$ such that $D \subset B(0,R)$.
We summarize the above result in the following theorem.
\begin{Theorem}\label{boundgeneral}
For a given domain $D$ and a contrast  $\varepsilon^*$, we have the following estimate for the scattering coefficient $W_{nm}(D,\varepsilon^*)$, for $n,m \in \mathbb{Z}$,
\beqn
|W_{nm}(D,\varepsilon^*)|
\leq
\f{1}{d( - {\varepsilon^*}^{-1} , \sigma (\tilde{K}_D) ) } \exp \left( \,
 \, \f{ C_{1,R} }{ d( - {\varepsilon^*}^{-1}  , \sigma (\tilde{K}_D) )^2 } + \f{1}{2}
\,\right) \frac{C_{2,R}^{|m|+|n|}}{{|m|}^{|m|} {|n|}^{|n|}} \,.
\label{estimate}
\eqn
%where $C_{i,R}, i = 1,2$ are some constant  depending only on the radius $R$ such that $D \subset B(0,R)$, and $\widetilde{K}_D$ is defined by \eqref{defK}.
\end{Theorem}
From Theorem \ref{boundgeneral}, we foresee that the magnitude of $W_{nm}$ may grow as $\varepsilon^*$ increases, and becomes a very large value as ${\varepsilon^*}^{-1}$ is close to the spectrum of the operator $\tilde{K}_D$.

\subsection{The case of a circular domain}

Now, we consider the operator $\tilde{K}_D$ for a circular domain, i.e., when $D=B(0,R)$.
In this case, the operator $\tilde{K}_D$ becomes more explicit.
Actually, from the Graf's formula \cite{Watson}, we have for $|x| \neq |y|$ that
\beqn
H_0^{(1)}(|x-y|) = \sum_{m=-\infty}^{\infty} \chi_{\{|x| < |y|\} } J_m (|x|)e^{-i m \theta_{x}} H_m^{(1)} (|y|)e^{i m \theta_{y}} + \chi_{\{|x| > |y|\} } H_m^{(1)} ( |x|)e^{-i m \theta_{x}} J_m ( |y|)e^{i m \theta_{y}} \,. \nonumber
\eqn
Therefore, for all $f \in L^2(D)$, the operator $\widetilde{K}_D$ can be written as
\beqn
\widetilde{K} _D [f ](y)= -\f{i}{4} \sum_{m=-\infty}^{\infty} \bigg[ \langle  J_m (r)e^{i m \theta} , f\rangle_{D\bigcap B(0,|y|)} H^{(1)}_m (|y|)e^{i m \theta_{y}}  \nonumber\\
+ \langle \overline{H^{(1)}_m (r)} e^{i m \theta}, f\rangle_{D\backslash \overline{B(0,|y|)}  } J_m (|y|)e^{i m \theta_{y}} \bigg] \,. \notag \nonumber
\eqn
The above expression of $\tilde{K}_D$ will be helpful to investigate the behavior of $\tilde{K}_D$ and $W_{nm}$. Before we continue our discussion on the operator $\tilde{K}_D$, we shall first define some operators.
\begin{Definition}
Given an integer $m \in \mathbb{Z}$, the operators $\widetilde{K}^{(i)}_m : L^2((0,R), r \, dr ) \rightarrow L^2((0,R), r \, dr )$ for $i=1,2$ are defined as
\begin{equation}
\widetilde{K}^{(i)}_m[\phi] (h) =
- \f{i}{4} \bigg(\int_0^h r J_m (r) \phi(r) dr \bigg) H^{(i)}_m (h) - \f{i}{4} \bigg( \int_h^R r H^{(i)}_m (r) \phi (r) dr \bigg) J_m (h)
\label{defKm}
\end{equation}
for $h \in (0,R)$ and $\phi\in L^2((0,R), r \, dr )$,
and their extensions $\widetilde{\widetilde{K}}^{(i)}_m : L^2((0,R), r \, dr ) \rightarrow L^\infty((0,+\infty) )$ for $i=1,2$  as
\begin{equation}
\widetilde{\widetilde{K}}^{(i)}_m [\phi] (h) =
- \f{i}{4} \bigg(\int_0^h r J_m (r) \phi(r) dr \bigg) H^{(i)}_m (h) - \f{i}{4} \bigg(\int_h^R r H^{(i)}_m (r) \phi (r) dr \bigg) J_m (h)
\label{defKmext}
\end{equation}
for $h \in (0,+\infty)$ and $\phi\in L^2((0,R), r \, dr )$.
\end{Definition}

\noindent With this notion, we can readily see that if $f \in L^2(D)$ has the form $f = \phi(r) e^{i m \theta}$, then we have in polar coordinates by the orthogonality of $\{ e^{i m \theta} \}_{m \in \mathbb{Z}}$ on $L^2(\mathbb{S}^1)$ that
\beqn
\widetilde{K}_D [ f ](h,\theta) &=&
-\f{i}{4} \bigg(\int_0^h r J_m (r) \phi(r) dr \bigg) H^{(1)}_m (h) e^{i m \theta}
-\f{i}{4} \bigg(\int_h^R r H^{(1)}_m (r) \phi (r) dr\bigg)  J_m (h) e^{i m \theta} \notag\\
& = & \widetilde{K}^{(1)}_m [\phi] (h) e^{i m \theta} \,,
\label{sub}
\eqn
and
$
\widetilde{K}^*_D [f ](h,\theta) = \widetilde{K}^{(2)}_m [\phi] (h) e^{i m \theta} \,.
$
Furthermore, we can directly see that $\sigma( \tilde{K}^{(2)}_m )
= \overline{\sigma( \tilde{K}^{(1)}_m )} $.
Moreover, using the following relations for all $m \in \mathbb{Z}$,
\beqn
J_{-m}(z) = (-1)^m J_m (z) \quad \text{ and } \quad H^{(1)}_{-m}(z) = (-1)^m H^{(1)}_{m}(z),
\label{symmetry}
\eqn
we immediately infer the properties for the integral operators:
\begin{equation}
\widetilde{K}^{(i)}_{-m}  = \widetilde{K}^{(i)}_m   \quad \text{ and } \quad \widetilde{\widetilde{K}}^{(i)}_{-m} = \widetilde{\widetilde{K}}^{(i)}_m \, .
\label{symmetryeq}
\end{equation}

Substituting \eqref{sub} into Theorem \ref{rfundamental}, we obtain the following simplified expressions of the scattering coefficients when $D = B(0,R)$.
\begin{Theorem}
\label{fundamentalhaha}
For a domain $D = B(0,R)$ for some $R>0$ and a contrast value $\varepsilon^*$, the scattering coefficient $W_{nm}(D,\varepsilon^*), \, n,m \in \mathbb{Z}$, can be written in the following form
\beqn
W_{nm} (D,\varepsilon^*) =  \delta_{nm} \left\langle  J_n, \left({\varepsilon^*}^{-1} +  \tilde{K}^{(1)}_m\right)^{-1} [J_m] \right\rangle_{L^2((0,R), r \, dr ) },
\label{useful}
\eqn
where $\delta_{nm}$ is the Kronecker symbol.
\end{Theorem}
As a consequence of Theorem \ref{fundamentalhaha}, we easily see that $W_{nm}=0$ for $n \neq m$.
Moreover, we readily have the following \textit{a priori} estimate for the coefficients $W_{nm}$ by the same arguments as those in Theorem \ref{boundgeneral}.
In order to obtain the desired estimate,
we consider the asymptotic expression of $Y_m$ as $m \rightarrow \infty$ \cite[pp.~365-366 ]{handbook}:
\beqn
Y_m(z) \bigg \slash \sqrt{\f{2}{ \pi m} }  \left(\frac{e z}{2 m}\right)^{-m}  \rightarrow 1 \, .
\label{asymptotics2}
\eqn
Together with (\ref{asymptotics}) and the logarithmic type singularity of $Y_0$, we have from the definitions of $\widetilde{K}^{(i)}_m$ for $i=1,2$ in \eqref{defKm} that
\beqn
||\widetilde{K}^{(i)}_m)||_{S_{2}(L^2((0,R), r \, dr))}^2  \leq C_m \left(1+R \right)^4 \left( 1+ \log R \right)^2 < \infty \,.
\eqn
Consequently, following the same arguments as the ones for \eqref{estimate},
we arrive at the estimate:
\beqn
|W_{nm} (D,\varepsilon^*)|
&=&  \delta_{nm} \left| \left\langle  J_n, \left({\varepsilon^*}^{-1} +  \tilde{K}^{(1)}_m\right)^{-1}[J_m] \right\rangle_{L^2((0,R), r \, dr ) } \right|\notag \\
&\leq &
 \delta_{nm}  \f{1}{d\left( - {\varepsilon^*}^{-1} , \sigma (\tilde{K}^{(1)}_m) \right) } \exp \left( \,
 \, \f{C_{m} C_{1,R} }{ d\left( - {\varepsilon^*}^{-1} , \sigma (\tilde{K}^{(1)}_m) \right)^2 } + \f{1}{2}
\,\right)\frac{C_{2,R}^{|m|+|n|}}{{|m|}^{|m|} {|n|}^{|n|}} , \notag
\eqn
where $C_m$ is a constant depending only on $m$ and $C_{i,R}, i = 1,2,$ are constants only depending on the radius $R$ such that $D \subset B(0,R)$.
\begin{Theorem}\label{boundball}
For a circular domain $D = B(0,R)$ and a contrast  $\varepsilon^*$, we have the following estimate for the scattering coefficient $W_{nm}(D,\varepsilon^*)$, for $n,m \in \mathbb{Z}$,
\beqn
|W_{nm} (D,\varepsilon^*)| \leq   \delta_{nm}
\f{1}{d\left( - {\varepsilon^*}^{-1} , \sigma (\tilde{K}^{(1)}_m) \right) } \exp \left( \,
 \, \f{C_{m} C_{1,R} }{ d\left( - {\varepsilon^*}^{-1} , \sigma (\tilde{K}^{(1)}_m) \right)^2 } + \f{1}{2}
\,\right) \frac{C_{2,R}^{|m|+|n|}}{{|m|}^{|m|} {|n|}^{|n|}} \, .
\label{boundspecific}
\eqn
%where $C_m$ are some constant depending only on $m$, $C_{i,R}, i = 1,2,$ are some constants only depending on the radius $R$ such that $D \subset B(0,R)$ and $\widetilde{K}^{(1)}_m$ is given by \eqref{defKm}.
\end{Theorem}

In the next section we perform a  sensitivity analysis of the scattering coefficients in order to obtain a quantitative description of what piece of information is provided by the scattering coefficients of different orders.

\section{Sensitivity analysis of the scattering coefficients for a given contrast} \label{sec3}

In this section, for a given contrast $\varepsilon^*$, we calculate the shape derivative  $\mathcal{D} \,W_{nm}(D, \varepsilon^* )[h]$ of the scattering coefficient $W_{nm}(D, \varepsilon^* )$ along the variational direction $h \in \mathcal{C}^1(\p D)$ when $\p D$ is of class $\mathcal{C}^2$.  From the shape derivative, we will clearly understand what piece of information is provided by the scattering coefficients of different orders, and how the knowledge of scattering coefficient of order $(n,m)$ is related to the resolution of the reconstructed shapes.

Before going into the sensitivity analysis, we will consider the inclusion of the operators and spectra between operators for the subsequent analysis.
To do so, we define the following inclusion maps.
\begin{Definition}
For a given domain $D$, suppose that the bounded linear operator $\tilde{K}_D \in  \mathfrak{B} \left( L^2(D)  \right)$ is defined as in \eqref{defK}.  Consider any domain $\widehat{D}$ such that $D \subset \widehat{D}$, we shall often write $\iota(\tilde{K}_D) \in \mathfrak{B} \left( L^2(\widehat{D})  \right)$ as the following operator:
\beqn
  \iota(\tilde{K}_D) \left[ f \right]  (x)  =
 \tilde{\tilde{K}}_{D} \left[ \chi_{D} f \right] (x)
\quad \text{ for any } f \in L^2( \widehat{D}  ) \, ,
\label{defKiota}
\eqn
where %$\tilde{\tilde{K}}_{D}$ is defined as in \eqref{defKext} and
$\chi_{D}$ is the characteristic  function of $D$. Likewise, for a given radius $R>0$, assume the bounded linear operators $\tilde{K}^{(i)}_m \in  \mathfrak{B} \left( L^2((0,R), r \, dr )  \right)$  $ ( m \in \mathbb{Z}, \, i = 1,2 ) $ , which are defined in \eqref{defKm}. Then we write  %Then for any radius $\hat{R} > R$, the following operator $\iota(\tilde{K}^{(i)}_m) \in  \mathfrak{B} \left( L^2((0,\hat{R}), r \, dr )  \right)$ is defined as
\beqn
  \iota(\tilde{K}^{(i)}_m) \left[ f \right]  (x)  =
 \tilde{\tilde{K}}^{(i)}_m \left[ \chi_{(0,R)} f \right] (x)
\quad \text{ for any } f \in L^2((0,\widehat{R}), r \, dr )  \, .
\label{defKmiota}
\eqn
%where $\tilde{\tilde{K}}^{(i)}_m$ is defined as in \eqref{defKmext}.
\end{Definition}
Then the operators $\iota(\tilde{K}_D) $ and $\iota(\tilde{K}^{(i)}_m), \, i = 1,2,$ are compact on $L^2((0,\widehat{R}), r \, dr ) $. Moreover, we have the following relations between the spectra of $\tilde{K}_D$ and $\iota(\tilde{K}_D)$, as well as between $\tilde{K}^{(i)}_m$ and $\iota(\tilde{K}^{(i)}_m)$ for $m \in \mathbb{Z},\, i = 1,2$.
\begin{Lemma}
\label{lemmaeigen}
Let $\tilde{K}_D$ and $\iota(\tilde{K}_D)$ be defined as in \eqref{defK} and \eqref{defKiota}, respectively. Then,
the following simple relationship between the spectra of $\tilde{K}_D$ and $\iota(\tilde{K}_D)$ holds:
\beqn
\sigma(\iota(\tilde{K}_D)) = \sigma(\tilde{K}_D) \bigcup \{0\} \,.
\eqn
Likewise, for $m \in \mathbb{Z},\, i = 1,2$, we have
\beqn
\sigma(\iota(\tilde{K}^{(i)}_m)) = \sigma(\tilde{K}^{(i)}_m) \bigcup \{0\} \,.
\eqn
\end{Lemma}
\begin{proof}
For a given $\lambda$, suppose that the pair $(\lambda, e_\lambda)$ is an eigenpair of $\tilde{K}_D$ over $L^2(D)$. If $\lambda \neq 0$, we denote by $\widetilde{e_\lambda} \in L^2(\widehat{D}) $ the following function
\beqn
 \widetilde{e_\lambda}:= \f{1}{\lambda} \widetilde{\widetilde{K_D}} [ e_\lambda ] \,.
\eqn
If $\lambda$ = 0, we write $\widetilde{e_\lambda} \in L^2(\widehat{D})$ as the extension by zero of the function $e_\lambda$ outside the domain $D$, i.e.,
\beqn
 \widetilde{e_\lambda} (x):=
\begin{cases}
e_\lambda (x) &\text{ if } x \in $D$\,, \\
0 & \text{otherwise} \,.
\end{cases}
\eqn
Then we readily check from the definition of $\iota(\tilde{K}_D)$ that $ \iota(\tilde{K}_D) [\widetilde{e_\lambda}] = \lambda  \widetilde{e_\lambda} $
and hence the pair $(\lambda, \widetilde{e_\lambda})$ is an eigenpair of $\iota(\tilde{K}_D)$ over $L^2(\widehat{D})$.
As any function $f \in L^2(\widehat{D} \backslash D)$ is a zero eigenfunction of $\iota(\tilde{K}_D)$, hence we know
$ \sigma(\tilde{K}_D) \bigcup \{0\} \subset \sigma(\iota(\tilde{K}_D))$.

Conversely, if a pair $(\lambda, \widetilde{e_\lambda})$ is an eigenpair of $\iota(\tilde{K}_D)$ over $L^2(\widehat{D})$, then, by writing $e_\lambda := \widetilde{e_\lambda} \mid_{D} $, it is easy to see form the definition of $\tilde{K}_D$ that $(\lambda, e_\lambda)$ is an eigenpair of $\tilde{K}_D$. Hence, $\sigma(\iota(\tilde{K}_D)) \subset \sigma(\tilde{K}_D)$.
The proof of $\sigma(\iota(\tilde{K}^{(i)}_m)) = \sigma(\tilde{K}^{(i)}_m) \bigcup \{0\} $ is the same.
\end{proof}
Lemma \ref{lemmaeigen} and the Fredholm alternative yield that ${\varepsilon^*}^{-1} + \iota(\tilde{K}_D)$ is invertible over $L^2(\widehat{D})$ if and only if ${\varepsilon^*}^{-1} +\tilde{K}_D$ is invertible over $L^2(D)$.
%Moreover, we can again use the spectral theorem \label{thm1} to obtain a Jordan block decomposition of $\iota(\tilde{K}_D)$ similar to that of \eqref{decomposition} for $\tilde{K}_D$.
Moreover, from the definition, we can show as in section \ref{sec2} that $\iota(\tilde{K}_D) \in S_2(L^2(\widehat{D}))$ and then apply \eqref{resolventineq} to obtain the following resolvent estimate for ${\varepsilon^*}^{-1} + \iota(\tilde{K}_D)$ that
\beqn
\bigg|\bigg| \left( {\varepsilon^*}^{-1} + \iota(\tilde{K}_D) \right)^{-1} \bigg|\bigg|_{L^2(\widehat{D})} &\leq&
\f{1}{d\left( - {\varepsilon^*}^{-1} , \sigma(\iota(\tilde{K}_D)) \right) } \exp \left( \,
 \, \f{  C_{1,R} }{ d\left( - {\varepsilon^*}^{-1} ,  \sigma(\iota(\tilde{K}_D)) \right)^2 } + \f{1}{2}
\,\right) \notag\\
& = &\f{1}{d\left( - {\varepsilon^*}^{-1} , \sigma (\tilde{K}_D) \right) } \exp \left( \,
 \, \f{ C_{1,R} }{ d\left( - {\varepsilon^*}^{-1} , \sigma (\tilde{K}_D) \right)^2 } + \f{1}{2}
\,\right).
\label{ineqKiota}
\eqn
Here the last equality comes from Lemma \ref{lemmaeigen} and the fact that $\sigma(\widetilde{K}_{D})$ must have zero as its accumulation point, since $L^2(D)$ is infinite dimensional.  The above argument also applies to the operators $\iota(\tilde{K}^{(i)}_m)$ for $m\in \mathbb{Z}, i=1,2$, where the resolvent estimate reads
\beqn
\bigg|\bigg| \left( {\varepsilon^*}^{-1} +\iota(\tilde{K}^{(i)}_m)) \right)^{-1} \bigg|\bigg|_{L^2((0,\widehat{R}), r \, dr)} \leq
\f{1}{d\left( - {\varepsilon^*}^{-1} , \sigma (\tilde{K}^{(1)}_m) \right) } \exp \left( \,
 \, \f{C_{m} C_{1,R} }{ d\left( - {\varepsilon^*}^{-1} , \sigma (\tilde{K}^{(1)}_m) \right)^2 } + \f{1}{2}
\,\right) \,.
\label{ineqKmiota}
\eqn
Furthermore, we can easily recover the relationship between $\iota(\widetilde{K}_{B(0,R)})$ and $\iota(\tilde{K}^{(i)}_m)$ for any $D$ such that $B(0,R) \subset D$ from their definitions.
In fact, for any $f \in L^2(D)$ in the form $f = \phi(r) e^{i m \theta}$, where $(r,\theta) \in D$, we have in polar coordinates that
\beqn
\iota(\widetilde{K}_{B(0,R)}) [f ](h,\theta) = \iota(\widetilde{K}^{(1)}_m) [\phi] (h) e^{i m \theta} \, , \quad
\iota(\widetilde{K}^*_{B(0,R)})[ f ](h,\theta) = \iota(\widetilde{K}^{(2)}_m) [\phi] (h) e^{i m \theta} \,,
\eqn
where the operators $\iota(\widetilde{K}^{(i)}_m)$ for $m \in \mathbb{Z}, i = 1,2,$ are the extensions to $L^2((0,\widehat{R}_{\theta}), r \, dr)$ with the radii $\widehat{R}_\theta$ being defined as $\widehat{R}_\theta := \sup \{r : (r , \theta) \in D \}$ for different $\theta \in [0,2 \pi]$. Although the extensions $\iota(\widetilde{K}^{(i)}_m)$ are now different for different angles $\theta$, no difficulty will arise in understanding the properties of $\iota(\widetilde{K}_{B(0,R)})$ via estimating $\iota(\widetilde{K}^{(1)}_m)$, since the conclusions of Lemma \ref{lemmaeigen} and \eqref{ineqKmiota} do not depend on the choice of $\widehat{R}$ and thus can be applied to different choices of radii.

From now on, %for simplicity of presentation,
we will no longer distinguish between the operators $\tilde{K}_D$ and $\iota(\tilde{K}_D)$
whenever there is no ambiguity,
and by an abuse of notation, we denote both operators by $\tilde{K}_D$,
likewise for the operators  $\tilde{K}^{(i)}_m$ and $\iota(\tilde{K}^{(i)}_m)$ for $m \in \mathbb{Z}, i = 1,2$.

Then we move to our main focus of this subsection, which is to obtain the shape derivative of the scattering coefficients for a domain $D$ along a perturbation $ h\in \mathcal{C}^1(\p D)$.
Now let  $\varepsilon^*$ be given. For any bounded $\mathcal{C}^2$-domain $D$ in $\mathbb{R}^2$,
let $D^\delta$ be a $\delta$-perturbation of $D$ along the variational direction $h \in \mathcal{C}^1(\p D)$, i.e.,
\beqn
\p D^\delta := \bigg \{ \tilde{x} = x + \delta h(x) \nu(x) \, : \, x \in \p D \bigg\} \,,
\eqn
where $\nu(x)$ is the outward unit normal at $\p D$.
For such perturbations of the domain $D$, we  investigate the difference between $W_{nm}(D^{\delta}, \varepsilon^*)$ and $W_{nm}(D, \varepsilon^*)$.
We first estimate the difference $\tilde{K}_{D^\delta} -  \tilde{K}_{D}$, where both operators $\tilde{K}_{D^\delta}$ and $\tilde{K}_{D}$ are regarded as the extended operators on $L^2 \left(D^{\delta} \bigcup D \right)$. Indeed, from the fact that the singularity type of the function $H^{(1)}_0$ is logarithmic, there exists a constant $C_{R}$ depending only on the radius $R$ such that the estimate
\beqn
|| \tilde{K}_{D^\delta} -  \tilde{K}_{D} ||_{L^2 \left(B(0,R) \right) } \leq C_{R} \, \delta
\label{Est} % \delta \log \delta .
\eqn
holds for $\delta$ small enough with $R$ being such that $D \Subset B(0,R)$.
Therefore, we can repeatedly apply the following resolvent equalities
%{\small
\beqn
\left({\varepsilon^*}^{-1} +  \tilde{K}_{D^\delta}\right)^{-1}  - \left( {\varepsilon^*}^{-1} + \tilde{K}_{D} \right)^{-1} &=
\left({\varepsilon^*}^{-1} +  \tilde{K}_{D^\delta} \right)^{-1}
  ( \tilde{K}_{D} - \tilde{K}_{D^\delta})
 \left( {\varepsilon^*}^{-1} + \tilde{K}_{D}\right)^{-1} \\
&=
\left( {\varepsilon^*}^{-1} + \tilde{K}_{D}\right)^{-1}  ( \tilde{K}_{D} - \tilde{K}_{D^\delta}) \left({\varepsilon^*}^{-1} +  \tilde{K}_{D^\delta}\right)^{-1}
\eqn
%}
to obtain the following expression of the difference of scattering coefficients for any $n,m \in \mathbb{Z}$,
{\
\beqn
& & W_{nm}({D^\delta}, \varepsilon^* ) - W_{nm}(D, \varepsilon^* ) \notag \\
&=&
\left\langle \left({\varepsilon^*}^{-1} +  \tilde{K}^*_{D^\delta }\right)^{-1}[J_n(r) e^{i n \theta}] , J_m (r) e^{i m \theta} \right\rangle_{L^2(D^\delta ) } - \left\langle \left({\varepsilon^*}^{-1} +  \tilde{K}^*_{D}\right)^{-1}[J_n(r) e^{i n \theta}], J_m (r) e^{i m \theta} \right\rangle_{L^2(D ) }  \notag \\
&=& \left\langle  J_n(r) e^{i n \theta},  \left[\left({\varepsilon^*}^{-1} +  \tilde{K}_{D^\delta}\right)^{-1}  - \left( {\varepsilon^*}^{-1} + \tilde{K}_{D}\right)^{-1}\right][J_m (r) e^{i m \theta}] \right\rangle_{L^2(D) } \notag\\
& & + \left\langle \left({\varepsilon^*}^{-1} +  \tilde{K}^*_{D^\delta}\right)^{-1}[ J_n(r) e^{i n \theta}], \text{sgn}(h) \, J_m (r) e^{i m \theta} \right\rangle_{L^2(D \bigcup {D^\delta}  \backslash D \bigcap {D^\delta} )  }
\notag\\
&=&
- \left\langle \left({\varepsilon^*}^{-1} +  \tilde{K}^*_{D}\right)^{-1}[J_n(r) e^{i n \theta}], (\tilde{K}_{D^\delta} -  \tilde{K}_{D} ) \left( {\varepsilon^*}^{-1} + \tilde{K}_{D^\delta }\right)^{-1}[J_m (r) e^{i m \theta}] \right\rangle_{L^2(D) } \notag \\
& & + \left\langle \left({\varepsilon^*}^{-1} +  \tilde{K}^*_{D^\delta}\right)^{-1}[J_n(r) e^{i n \theta}],  \text{sgn}(h) \,  J_m (r) e^{i m \theta} \right\rangle_{L^2(D \bigcup {D^\delta}  \backslash D \bigcap {D^\delta} )  } \notag \\
&=&
- \left\langle \left({\varepsilon^*}^{-1} +  \tilde{K}^*_{D}\right)^{-1}[J_n(r) e^{i n \theta}], (\tilde{K}_{D^\delta} -  \tilde{K}_{D} ) \left( {\varepsilon^*}^{-1} + \tilde{K}_{D}\right)^{-1}[J_m (r) e^{i m \theta}] \right\rangle_{L^2(D) } \notag \\
& & + \left\langle \left({\varepsilon^*}^{-1} +  \tilde{K}^*_{D}\right)^{-1}[J_n(r) e^{i n \theta}], \text{sgn}(h)\,  J_m (r) e^{i m \theta} \right\rangle_{L^2(D \bigcup {D^\delta}  \backslash D \bigcap {D^\delta} )  }
 + O(\delta^2),
\label{devirative}
\eqn
}where the last equality comes from \eqref{Est}. Now for any $L^1$ function $f$, considering the fact that the shape derivative of the integral
\beqn
I [D] = \int_{D} f(x) d x
\eqn
is given by the following boundary integral
\beqn
\mathcal{D} \, I [D ] (h) = \int_{\p D} f(x) h(x) \,ds(x) \,,
\eqn
we have for $x \in D \bigcup {D^\delta}$ and $\phi \in L^2 (D \bigcup {D^\delta} )$ that
\beqn
(\tilde{K}_{D^\delta} -  \tilde{K}_{D} )[ \phi ] (x) &=& -\f{i}{4} \int_{ (D \bigcup {D^\delta} ) \backslash (D \bigcap {D^\delta}) } \text{sgn}(h) \, H^{(1)}_{0}(|x-y|) \phi(y) dy \notag \\
&=& - \delta \f{i}{4} \int_{\p D} H^{(1)}_{0}(|x-y|) h(y) \phi(y) \,ds(y) + O(\delta^2) \,. \,
\eqn
Therefore, by substituting the above expression into \eqref{devirative}, a direct
expansion of the integral together with the Fubini's theorem yields the following expression for the first term in \eqref{devirative}:
\beqn
& & - \left\langle \left({\varepsilon^*}^{-1} +  \tilde{K}^*_{D}\right)^{-1}  [J_n(r) e^{i n \theta}] , (\tilde{K}_{D^\delta} -  \tilde{K}_{D} ) \left( {\varepsilon^*}^{-1} + \tilde{K}_{D}\right)^{-1}[J_m (r) e^{i m \theta}] \right\rangle_{L^2(D ) }  \notag \\
&=&
\delta \f{i}{4} \int_{D}
\int_{\p D} H^{(1)}_{0}(|x-y|) h(y) \left[ \left( {\varepsilon^*}^{-1} + \tilde{K}_{D}\right)^{-1} [J_m (r) e^{i m \theta} ] \right] (y) \, d y \overline{\left[ \left({\varepsilon^*}^{-1} +  \tilde{K}^*_{D}\right)^{-1} [ J_n(r) e^{i n \theta} ]\right] (x)} \, dx \notag \\
& & + O(\delta^2)
\notag\\
&=& - \delta
\int_{\p D}
  h(y)  \left[ \left( {\varepsilon^*}^{-1} + \tilde{K}_{D}\right)^{-1} [J_m (r) e^{i m \theta} ] \right] (y)  \overline{ \left[ \tilde{K}^*_{D} \left({\varepsilon^*}^{-1} +  \tilde{K}^*_{D}\right)^{-1}  [J_n(r) e^{i n \theta}] \right] (y)}  \, d y + O(\delta^2)
\notag\\
&=& - \delta
\left\langle \overline{ \left[ \left( {\varepsilon^*}^{-1} + \tilde{K}_{D}\right)^{-1} [ J_m (r) e^{i m \theta} ]\right] }  \left[ \tilde{K}^*_{D} \left({\varepsilon^*}^{-1} +  \tilde{K}^*_{D}\right)^{-1}[J_n(r) e^{i n \theta}]\right]  , h \, \right\rangle_{L^2(\p D)}  + O(\delta^2).
\label{devirative21}
\eqn
Likewise, for the second term in \eqref{devirative}, we derive that
\beqn
& & \left\langle \left({\varepsilon^*}^{-1} +  \tilde{K}^*_{D}\right)^{-1}[ J_n(r) e^{i n \theta}], \text{sgn}(h) \, J_m (r) e^{i m \theta} \right\rangle_{L^2(D \bigcup {D^\delta}  \backslash D \bigcap {D^\delta} )  } \notag \\
&=&
\delta \int_{\p D}
  h(y)  \left[ \left( {\varepsilon^*}^{-1} + \tilde{K}_{D}\right)^{-1} [J_m (r) e^{i m \theta} ] \right] (y)  \overline{ \left[ J_n(r) e^{i n \theta} \right] (y) } \, d y + O(\delta^2) \notag\\
&=&
  \delta \left\langle \overline{ \left[ \left( {\varepsilon^*}^{-1} + \tilde{K}_{D}\right)^{-1} [J_m (r) e^{i m \theta} ] \right] }  \left[ J_n(r) e^{i n \theta} \right]  , h \, \right\rangle_{L^2(\p D)} + O(\delta^2).
\label{devirative22}
\eqn
Therefore, combining the above two estimates shows that
\beqn
& & W_{nm}( {D^\delta}, \varepsilon^*) - W_{nm}( D, \varepsilon^*) \notag \\
&=&
 \delta {\varepsilon^*}^{-1} \left\langle \overline{ \left[ \left( {\varepsilon^*}^{-1} + \tilde{K}_{D}\right)^{-1}[J_m (r) e^{i m \theta} ] \right] }  \left[ \left({\varepsilon^*}^{-1} +  \tilde{K}^*_{D}\right)^{-1}[J_n(r) e^{i n \theta}] \right]  , h \, \right\rangle_{L^2(\p D)} \notag \\
 && + O(\delta^2) \,.
\label{devirativetotal}
\eqn
Hence, if we define the following $L^2(\partial D)$-duality gradient function $\nabla W_{nm} (D, \varepsilon^*)$ of the form of
\beqn
\nabla W_{nm} ( D, \varepsilon^*) :=   {\varepsilon^*}^{-1} \overline{ \left[ \left( {\varepsilon^*}^{-1} + \tilde{K}_{D}\right)^{-1}[J_m (r) e^{i m \theta}] \right] } \left[ \left({\varepsilon^*}^{-1} +  \tilde{K}^*_{D}\right)^{-1} [J_n(r) e^{i n \theta}] \right] \,,
\label{gradient}
\eqn
then the shape derivative of the scattering coefficient $W_{nm}(D, \varepsilon^*)$ along the variational direction $h$ is given by
\beqn
\mathcal{D} \,W_{nm}(\varepsilon^* , D)  [h]  = \left\langle  \nabla W_{nm} (\varepsilon^*, D) , h \, \right\rangle_{L^2(\p D)}.
\eqn
%where $\nabla W_{nm} (\varepsilon^*, D)$ is given as in the previous formula.
In particular, for the case where $D$ is a circular domain $D = B(0,R)$, we have from the decomposition of the operator $\tilde{K}_D$ the following simple expression of
$\nabla W_{nm} ( D, \varepsilon^*)$:
\beqn
\nabla W_{nm} ( B(0,R), \varepsilon^*) &=& {\varepsilon^*}^{-1}  \overline{ \left[ \left( {\varepsilon^*}^{-1} + \tilde{K}_{B(0,R)}\right)^{-1}[J_m (r) e^{i m \theta}] \right] }  \left[ \left({\varepsilon^*}^{-1} +  \tilde{K}^*_{B(0,R)}\right)^{-1} [J_n(r) e^{i n \theta}] \right] \,, \notag\\
&=& {\varepsilon^*}^{-1}  \left[ \left( {\varepsilon^*}^{-1} + \tilde{K}^{(2)}_m\right)^{-1} [J_m] \right]  (R)   \left[ \left({\varepsilon^*}^{-1} + \tilde{K}^{(2)}_n \right)^{-1} [J_n]\right]  (R)  \, e^{i (n-m) \theta}  \,. \notag
\label{gradientcircle}
\eqn
Consequently,
\beqn
\mathcal{D} \,W_{nm}( B(0,R), \varepsilon^*) [h]  &=& {\varepsilon^*}^{-1}  \left[ \left( {\varepsilon^*}^{-1} + \tilde{K}^{(1)}_m\right)^{-1}[J_m] \right]  (R)   \left[ \left({\varepsilon^*}^{-1} + \tilde{K}^{(1)}_n\right)^{-1}[J_n] \right]  (R) \left\langle    \, e^{i (n-m) \theta}  , h \, \right\rangle_{L^2(\p D)}  \notag \\
&=&  {\varepsilon^*}^{-1}  \left[ \left( {\varepsilon^*}^{-1} + \tilde{K}^{(1)}_m\right)^{-1}[J_m] \right]  (R)   \left[ \left({\varepsilon^*}^{-1} + \tilde{K}^{(1)}_n\right)^{-1}[J_n] \right]  (R)\mathfrak{F}_\theta \left[ h \right](n-m)  \,,
\label{divcircle}
\eqn
where $\mathfrak{F}_\theta \left[h\right](n-m)$ is the $(n-m)$-th Fourier coefficient of the function $h$ on $L^2(\mathbb{S}^1)$. This gives the following  key result on the shape derivative of $W_{nm}( D, \varepsilon^* )$ .
\begin{Theorem}
\label{vari}
Suppose that $\varepsilon^* >0$ is given.
For any $\mathcal{C}^2$-domain $D$ and $n,m \in \mathbb{Z}$, the shape derivative of the scattering coefficient $W_{nm}( D, \varepsilon^*)$ along the variational direction $h \in L^2(\p D)$ is given by
\beqn
\mathcal{D} \,W_{nm}( D, \varepsilon^*) [h]  = \left\langle  \nabla W_{nm} ( D, \varepsilon^*) , h \, \right\rangle_{L^2(\p D)},
\eqn
where $\nabla W_{nm}$ is defined by
\beqn
\nabla W_{nm} ( D, \varepsilon^*) = {\varepsilon^*}^{-1} \overline{ \left[ \left( {\varepsilon^*}^{-1} + \tilde{K}_{D}\right)^{-1} [J_m (r) e^{i m \theta}] \right] } \left[  \left({\varepsilon^*}^{-1} +  \tilde{K}^*_{D}\right)^{-1}  [J_n(r) e^{i n \theta}] \right] \,.
\label{gradients}
\eqn
In particular, if the domain $D$ is a circular domain $D=B(0,R)$, then for any $D^\delta$ as a $\delta$-perturbation of $D$ along the variational direction $h \in \mathcal{C}^1(\p D)$, we have
\beqn
W_{nm}( D^\delta, \varepsilon^*) -W_{nm}( D, \varepsilon^*) = \delta \, C(\varepsilon^*,n,m) \mathfrak{F}_\theta \left[ h \right](n-m) + O(\delta^2),
\label{divcircle2}
\eqn
with
\beqn
C(\varepsilon^*,n,m) :=  {\varepsilon^*}^{-1}
\left[ \left( {\varepsilon^*}^{-1} + \tilde{K}^{(1)}_m\right)^{-1} [J_m] \right]  (R)   \left[ \left({\varepsilon^*}^{-1} + \tilde{K}^{(1)}_n\right)^{-1}  [J_n] \right]  (R)\,.
\label{coefficient}
\eqn
\end{Theorem}

From the above theorem, we obtain in the linearized case that the scattering coefficient $W_{nm}$ gives us precise information about the $(m-n)$-th Fourier mode of the perturbation $h$.

Therefore, the magnitude of the coefficients $W_{nm}$ and  $C(\varepsilon^*,n,m)$ shall be responsible for the resolution in imaging $D^\delta$. Note that the function $C(\varepsilon^*,n,m)$ depends now on the spectra of both $\tilde{K}^{(1)}_m$ and $\tilde{K}^{(1)}_n$. The change and growth of the coefficients $W_{nm}$ and  $C(\varepsilon^*,n,m)$ with respect to $\varepsilon^*$ will be the main focus of the next section.

\section{Asymptotic behaviors of eigenvalues over a circular domain and the phenomenon of super-resolution}

In the previous section, we have obtained a relationship between the coefficients $W_{nm}$ of a perturbed circular domain $D^\delta$ and the Fourier coefficients of the perturbation $h$. In this section, we investigate the decay of the eigenvalues of $\tilde{K}^{(1)}_m$ and analyze the behavior  with respect to $\varepsilon^*$ of $W_{nm}$ and $C(\varepsilon^*,n,m)$ for different values of $n$ and
$m$. For this purpose, we introduce the following Riesz decomposition.

\subsection{Riesz decomposition of the operators} \label{sec4_1}

To continue our analysis on the operators $\widetilde{K}_D$ and $\tilde{K}^{(1)}_m$, we first recall the following classical spectral theorem for compact operators in a Hilbert space  \cite{spectral}.
\begin{Theorem} \label{thm1}
Let $K$ be a compact operator on a Hilbert space $H$ and $\sigma(K)$ its spectrum.
 Then the following results hold:
\begin{enumerate}
\item
 $\lambda \in \sigma(K)$ if and only if $\lambda$ is an eigenvalue (Fredholm alternative).
\item
 For all $\lambda \in \sigma(K)$, there exists a smallest $m_{\lambda} $ such that $Ker (\lambda  - K)^{m_{\lambda}} = Ker (\lambda  - K)^{m_{\lambda}+1}$. Denoting the space $Ker (\lambda  - K)^{m_{\lambda}}$ by $E_\lambda$, we have $dim(E_\lambda) < \infty$.
\item
 $\sigma(K)$ is countable and $0$ is the only accumulation point of $\sigma(K)$ for $dim(H) = \infty$.
\item
 The map $z \mapsto ( z - K)^{-1}$ admits poles at $ z \in \sigma(K)$.
\end{enumerate}
\end{Theorem}
\noindent Applying the above theorem to $\widetilde{K}_D$, which is compact but not normal, we can decompose $$L^2(D) = \bigoplus_{\lambda \in \sigma(\widetilde{K}_D)} E_{\lambda}$$ with $E_\lambda = \bigoplus_{1\leq i \leq N_\lambda} E_\lambda^i$ for some $N_\lambda$ such that the operator $\widetilde{K}_D$ can be written as
\beqn
\widetilde{K}_D = \sum_{\lambda \in \sigma(\widetilde{K}_D)} \sum_{1 \leq i \leq N_\lambda}  \widetilde{K}_{i,\lambda},
\label{hahadecomp}
\eqn
where the operators $\widetilde{K}_{i,\lambda}: E_\lambda^i \rightarrow E_\lambda^i$ admit the action of the following Jordan block under a choice of basis ${\bf e}^i_\lambda$ in $E^i_{\lambda}$:
%\beqn
%\begin{pmatrix}
%J^1_\lambda & 0 & \dots & 0\\
%0 & J^2_\lambda & \dots & 0 \\
%\vdots & \vdots & \ddots & \vdots \\
%0 & \dots & \dots & J_\lambda^{N_\lambda}
%\end{pmatrix}
%\eqn
%for some $N_\lambda$ with $J^i_\lambda$ for all $ 1\leq i \leq N_\lambda$ as a Jordan block
\beqn
J^i_{\lambda} :=
\begin{pmatrix}
\lambda& 1 &\dots &\dots & 0\\
0 & \lambda & 1 & \dots &0 \\
\vdots & \vdots & \ddots & \vdots & \vdots \\
0 & \dots & \dots & \lambda & 1 \\
0 & \dots & \dots & \dots & \lambda
\end{pmatrix} \, ,
\eqn
as matrices of size smaller than or equal to $m_\lambda$. For the sake of simplicity, for a given $n \in \mathbb{N}$ and a given Riesz basis $\bf{w}$, i.e., a complete frame in $L^2(D)$, supposing that ${\bf v} $ is a finite subset of $\bf{w}$, we shall often write, for any $\phi \in L^2(D)$, $\left(\phi \right)_{{\bf v},L^2(D)} \in \mathbb{C}^n$ as the coefficients of $\phi$ in front of the vectors in ${\bf v}$ when expressed in the Riesz basis $\bf{w}$, i.e., if
\beqn
\phi = \sum_{w_i \in \bf{w}} b_i w_i  \,,
\eqn
for coefficients $b_i \in \mathbb{C}$ and ${\bf v} = (w_{k_1}, w_{k_2}, \ldots, w_{k_n})$, then $\left(\phi \right)_{{\bf v},L^2(D)} = (b_{k_1}, b_{k_2}, \ldots, b_{k_n}) $.
%\beqn
%\langle {\bf v}, \phi \rangle_{L^2(D)} :=\left(\langle v_1, \phi \rangle_{L^2(D)}, \langle v_2, \phi \rangle_{L^2(D)}, ..., \langle v_n, \phi \rangle_{L^2(D)} \right)  \in \mathbb{C}^n \,.
%\eqn
Also, for any $a = (a_1,\ldots,a_n) \in \mathbb{C}^n$, and any given finite frame ${\bf v} = (v_1 , v_2 , \ldots, v_n )$ in $L^2(D)$, we write
\beqn
{\bf v}^T a := \sum_{i=1}^n a_i v_i \, ,
\label{transpose}
\eqn
and, for any $\phi \in L^2(D)$, the $L^2$ inner product of ${\bf v}$ and $\phi$ as
\beqn
\langle {\bf v}, \phi \rangle_{L^2(D)} :=\left(\langle v_1, \phi \rangle_{L^2(D)}, \langle v_2, \phi \rangle_{L^2(D)}, \ldots, \langle v_n, \phi \rangle_{L^2(D)} \right)  \in \mathbb{C}^n \,.
\eqn
With these notations, we can write \eqref{hahadecomp} in terms of the frame $\bigcup_{\lambda \in \sigma({\tilde{K}_D})} \bigcup_{1 \leq i \leq N_\lambda} {\bf e}_\lambda$ as follows:
\beqn
\widetilde{K}_D = \sum_{\lambda \in \sigma(\widetilde{K}_D)} \sum_{1 \leq i \leq N_\lambda} \left( {\bf e}^i_\lambda \right)^T J^{i}_\lambda  (  \, \cdot\, )_{{\bf e}^i_\lambda,L^2(D)},
\label{decomposition}
\eqn
where the superscript $T$ denotes the transpose as described in \eqref{transpose}.
Therefore, substituting the above expression of $\widetilde{K}_D$ into \eqref{fundamentalexpression}, we have
\beqn
W_{nm} (D,\varepsilon^*) &=&   \left\langle J_n ( r) e^{i n  \theta},  ( {\varepsilon^*}^{-1} + \tilde{K}_D )^{-1}[J_m ( r) e^{i m  \theta}] \right\rangle_{L^2(D)} \notag \\
&=&
\sum_{\lambda \in \sigma(\widetilde{K}_D)} \sum_{1 \leq i \leq N_\lambda} \left[  \left\langle J_n ( r) e^{i n  \theta}, {\bf e}^i_\lambda  \right\rangle_{L^2(D)}   \right]^T [ J^{i}_{ {\varepsilon^*}^{-1} + \lambda }]^{-1} \left[  J_m ( r ) e^{i m  \theta} \right]_{{\bf e}^i_\lambda,L^2(D)}.
\label{farfield222}
\eqn
The above expression gives a general decomposition of the scattering coefficient $W_{nm} (D,\varepsilon^*)$.

Next, we consider the special domain $D = B(0,R)$. Then from Theorem \ref{fundamentalhaha} we shall focus on the operators $\widetilde{K}^{(1)}_m$ for $m \in \mathbb{Z}$.
Similarly to the previous argument, we can see that the operators $\widetilde{K}^{(1)}_m$ are compact on $ L^2((0,R), r \, dr ) $.  Then by Theorem \ref{thm1} there exists a complete basis $\bigcup_{\lambda} \bigcup_{ 0 \leq i \leq N_\lambda^m} {\bf e}^i_{m,\lambda}$ over $L^2((0,R), r \, dr )$ with each ${\bf e}^i_{m,\lambda} $ spanning the subspace
$E^i_{m,\lambda}$
such that $\widetilde{K}_m$ admits the action of a Jordan block, denoted by $J^i_{m,\lambda}$, with respect to the basis when acting on the invariant subspace $E^i_{m,\lambda}$.
Moreover, adopting the same notations as previously introduced, we can write
\beqn
\left({\varepsilon^*}^{-1} +  \tilde{K}^{(1)}_m\right)^{-1} =  \sum_{\lambda \in \sigma(\widetilde{K}^{(1)}_m)}
\sum_{1 \leq i \leq N^m_\lambda}    \left({\bf e}^i_{m,\lambda} \right)^T  [J^{i}_{m,{\varepsilon^*}^{-1} + \lambda}]^{-1}  (  \, \cdot\, )_{{\bf e}^i_{m,\lambda}, L^2((0,R), r \, dr ) },
\label{sumlong}
\eqn
and a similar expansion holds for $\tilde{K}^{(2)}_m$.
Now, using the orthogonality of $\{ e^{i m \theta} \}_{m \in \mathbb{Z}}$ on $L^2(\mathbb{S}^1)$, for a given contrast $\varepsilon^*$ such that $-{\varepsilon^*}^{-1} $ is not an eigenvalue of $ \tilde{K}^{(1)}_m$, we have that
\beqn
& &W_{nm} (D,\varepsilon^*) \notag \\
&=&   \delta_{nm} \left\langle  J_n, \left({\varepsilon^*}^{-1} +  \tilde{K}^{(1)}_m\right)^{-1} [J_m] \right\rangle_{L^2((0,R), r \, dr ) } \notag \\
&=&
 \delta_{nm}
\sum_{\lambda \in \sigma(\widetilde{K}^{(1)}_m)} \sum_{1 \leq i \leq N^m_\lambda} [ \langle J_n ( r),  {\bf e}^i_{m,\lambda}  \rangle_{L^2((0,R), r dr)} ]^T  [J^{i}_{m,{\varepsilon^*}^{-1} + \lambda}]^{-1}
 (  \, J_m(r) \, )_{{\bf e}^i_{m,\lambda}, L^2((0,R), r \, dr ) }.
\label{sumsubpart}
\eqn

Finally, the following remarks are in order.  For $D = B(0,R)$, the action of $\tilde{K}_D$ on each of the subspace $E^i_{m,\lambda} e^{i m \theta}$ of $L^2(D)$  is invariant and admits the same Jordan block representation as $\widetilde{K}^{(1)}_m$ acting on $E^i_{m,\lambda}$ of $L^2((0,R),r dr)$.  Hence,  the decomposition
$$
L^2(D) = \bigoplus_{m\in\mathbb{Z}} \bigoplus_{\lambda \in \sigma\left(\widetilde{K}^{(1)}_m\right)} \bigoplus_{1\leq i \leq N^m_{\lambda} } E^i_{m,\lambda} e^{i m \theta}
$$
coincides with the original Jordan block decomposition of $\tilde{K}_D$,
$$
L^2(D) = \bigoplus_{\lambda \in \sigma(\widetilde{K})} \bigoplus_{1\leq i \leq N_{\lambda} } E^i_{\lambda} \,.
$$
Therefore, we readily get $ \bigcup_{m\in \mathbb{Z}} \sigma(\widetilde{K}^{(1)}_m) = \sigma(\tilde{K}_D) $, and the sum \eqref{sumsubpart} constitutes a part of the sum \eqref{farfield222} with all the other terms in \eqref{farfield222} being zero.
In the next section, we will focus on the decay of the eigenvalues of $\tilde{K}_m$ and the asymptotic expansion for the eigenvalues and eigenfunctions of the operators. This will allow us to better understand the behavior of $W_{nm}$ and $C(\varepsilon^*,n,m)$.

\subsection{Asymptotics of the eigenvalues and eigenfunctions of $\widetilde{K}^{(i)}_m$ } \label{sec4_2}

Intuitively we can expect that the eigenvalues of $\widetilde{K}^{(1)}_m$ are distributed closer to $0$ as $|m|$ increases for the following reason.
Considering \eqref{symmetryeq} together with the asymptotic expressions \eqref{asymptotics} and \eqref{asymptotics2} of $J_m$ and $Y_m$ as $m \rightarrow \infty$, we have the following bound for the operator norm of $\tilde{K}_m$ for $m \in \mathbb{Z}$:
\beqn
|| \widetilde{K}^{(1)}_m ||_{L^2((0,R) r dr)} \leq  \f{C'_R}{m^2}
\eqn
for some constant $C'_R$ depending on $R$. Then we obtain the estimate for the spectral radius of $\tilde{K}_m$ from the Gelfand theorem:
\beqn
 \sup_{\lambda\in\sigma( \widetilde{K}^{(1)}_m )} | \lambda|  = \lim_{n \rightarrow \infty} \bigg|\bigg| \left( \widetilde{K}^{(1)}_m\right)^n\bigg|\bigg|^{\f{1}{n}} \leq \f{C'_R}{m^2}. \label{normest}
\eqn
This implies that the spectrum $\sigma(\widetilde{K}^{(1)}_m)$ actually lies inside
$ \sigma(\tilde{K}) \bigcap B(0,\f{C'_R}{m^2})$.

However, the above argument is a bit heuristic, and we intend to obtain a formal asymptotic expansion of the eigenvalues for the operators $\tilde{K}_m$.
For this purpose, we first restrict ourselves to the discussion of the operators for $m \in \mathbb{N}$, and consider the equation $\tilde{K}^{(1)}_m f = \lambda f$ with $\lambda \neq 0$.
Since we have
\beqn
\left( \frac{1}{r} \partial_r r \partial_r + 1 - \frac{m^2}{r^2} \right) \left( \tilde{K}^{(1)}_m f \right) e^{i m \theta} = (\Delta +1) \tilde{K}^{(1)}  \left( f  e^{i m \theta} \right) = f  e^{i m \theta},
\eqn
we obtain for $m \neq 0$ the following equivalence
\beqn
\tilde{K}^{(1)}_m f =  \lambda f \quad \Leftrightarrow
\begin{cases}
\left( \frac{1}{r} \partial_r r \partial_r + 1 - \frac{1}{\lambda}  - \frac{m^2}{r^2} \right) f &= 0 \,,\\
\quad \quad \quad \quad f(0) & = 0 \,, \\
\quad \quad \quad \quad f(R) &=  - \f{i}{4}\int_0^R r J_m (r) f(r) dr H^{(i)}_m (R) \,.
\end{cases}
\label{eigenvalue}
\eqn
Enumerating the eigenvalues $\lambda$ of $\tilde{K}^{(1)}_m$ as $\lambda_{m,l}$ in descending order of their magnitudes, and writing $e^i_{m,l}$ as the unique eigenfunction in the Jordan basis ${\bf e}^i_{m,\lambda_{m,l}}$ for each $i$, we are bound to have the following form for the eigenpair of the operator for all $i$,
\beqn
(\lambda_{m,l}, e^i_{m,l} ) =  \left( \lambda_{m,l}, J_m \left( \sqrt{1 - \frac{1}{ \lambda_{m,l} }} r  \right) \right)  \,.
\label{eigeigeig}
\eqn
The above statement implies that the geometric multiplicities of all the eigenvalues of $\tilde{K}^{(1)}_m$ should be $N_{\lambda} =1$ (while the algebraic multiplicities are still unknown for the time being). For the sake of simplicity, we denote the frame ${\bf e}^1_{m,\lambda_{m,l}}$ by ${\bf e}_{m,l}$, and also the eigenfunction $e^1_{m,l}$ by $ e_{m,l}$.
Substituting \eqref{eigeigeig} into \eqref{eigenvalue},
together with the following well-known property of Lommel's integrals \cite{handbook} that for all $n \in \mathbb{N}$ and for all $a, b >0$ with $a \neq b$:
\beqn
 \int_{0}^R [J_n(a r )]^2 r dr &=& \frac{R^2}{2} [ J_{n}(a R) ^2 - J_{n-1}(a R) J_{n+1}(a R)  ] \,, \label{lommel1} \\
  \int_{0}^R J_n(a r ) J_n(b r ) r dr &=& \frac{R}{a^2 -b^2} [ b J_n(a R) J_{n-1}(b R) - a J_{n-1}(a R) J_{n}(b R)  ] \,,
  \label{lommel2}
\eqn
we get the following equation for $\lambda_{m,l}$:
\beqn
 & &J_m \left( \sqrt{1 - \frac{1}{ \lambda_{m,l} }} R  \right)  \notag \\
&=&  - \f{i}{4}\int_0^R r J_m (r)  J_m \left( \sqrt{1 - \frac{1}{ \lambda_{m,l} }} r  \right)  dr H^{(i)}_m (R) \notag \\
 &  = &
  \f{i}{4} R \lambda_{m,l} H^{(i)}_m (R) \bigg[ \sqrt{1 - \frac{1}{ \lambda_{m,l} }} J_m(R) J_{m-1}\left(\sqrt{1 - \frac{1}{ \lambda_{m,l} }} R \right) \notag \\
  && -  J_{m-1}( R) J_{m}\left(\sqrt{1 - \frac{1}{ \lambda_{m,l} }} R\right) \bigg]   \,.
  \label{root}
\eqn
Now since $\lambda_{m,l} \rightarrow 0 $ as $l \rightarrow \infty$, from the following well-known asymptotic of $J_n$ \cite{handbook} for all $n$:
\beqn
J_n \left( z  \right) = \sqrt{\frac{2}{ \pi z}} \cos \left( z - \frac{2 n + 1}{ 4} \pi \right) + O( {z}^{-3/2}) \,,
\label{order0}
\eqn
we obtain the following estimate for $m,n,l \in \mathbb{N}$:
\beqn
J_n \left( \sqrt{1 - \frac{1}{ \lambda_{m,l} }} R  \right) = \sqrt{\frac{2}{ \pi R \sqrt{1 - \frac{1}{ \lambda_{m,l} }}}} \cos \left( \sqrt{1 - \frac{1}{ \lambda_{m,l} }} R - \frac{2 n + 1}{ 4} \pi \right) + O( {| \lambda_{m,l}|}^{3/4}) \,.
\label{order}
\eqn
Hence, substituting this expression into  \eqref{root}, we shall directly infer that the eigenvalues $\lambda_{m,l}$ satisfy the following bound:
\beqn
J_m \left( \sqrt{1 - \frac{1}{ \lambda_{m,l} }} R  \right)
=  O( {|\lambda_{m,l}|}^{3/4} )  \,,
\label{importantbound}
\eqn
which has a decay order higher than the one in \eqref{order}. With this observation, we shall expect that the terms $\sqrt{1 - \frac{1}{ \lambda_{m,l} }} R$ should be close to the $l$-th zeros of the Bessel functions of $J_m$ as $l$ grows, which is indeed the case following the argument below.

For the sake of exposition, we shall often denote by $a_{m,l}$ the zeros of the $m$-th Bessel function of the first kind, i.e., $J_m(a_{m,l}) = 0$, arranged in ascending order.
Then it follows from \eqref{order0}, the inverse function theorem and the Taylor expansion that
\beqn
\left| a_{m,l} -  \frac{2 m  + 4 l - 1}{ 4} \pi   \right| < C {\left( m  + 2 l  \right) }^{-1/2} \rightarrow 0 \quad \text{ as } l \rightarrow \infty\,.
\eqn
Then, again from \eqref{order0}, we have
\beqn
 J_m' \left( a_{m,l} \right) - (-1)^{l} \sqrt{\frac{2}{ \pi a_{m,l}}}     =  O( { a_{m,l}}^{-3/2}) \,,
\eqn
which, combined with \eqref{importantbound}, leads to
\beqn
R \sqrt{1 - \frac{1}{ \lambda_{m,l} }} -  a_{m,l}   =  O( { a_{m,l}}^{-1/2} ) \,. %\quad \text{ as } l \rightarrow \infty\,.
\label{estdone}
\eqn
This gives us the following estimate for $\lambda_{m,l}$:
\beqn
R \sqrt{1 - \frac{1}{ \lambda_{m,l} }}   \bigg \slash \left( \frac{ (m + 2l) \pi }{2} - \frac{ \pi }{4} \right) \rightarrow 1 \quad \text{ as } l \rightarrow \infty\,.
\eqn
Therefore, we obtain the following decay rate of the eigenvalues,
\beqn
\lambda_{m,l}  \bigg \slash \left( \frac{4 R^2}{ \pi^2}\frac{ 1 }{(m + 2l)^2} \right) \rightarrow - 1 \quad \text{ as } l \rightarrow \infty\,.
\label{decayeig}
\eqn
Moreover, using \eqref{estdone} and the fact that $J_m$ is holomorphic, we have the following uniform estimate for the eigenfunctions:
\beqn
\left|\left| J_m \left( \sqrt{1 - \frac{1}{ \lambda_{m,l} }} r  \right) - J_m \left( \frac{a_{m,l}}{R} r  \right) \right|\right|_{\mathcal{C}^0 ((0,R))}
\leq C ||  J_m'||_{L^\infty ((0,R)) }  { a_{m,l}}^{-1/2}  < C {\left( m  + 2 l  \right) }^{-1/2} \,.
\eqn
Note that the set $\{ J_m \left( \frac{a_{m,l}}{R} r  \right) \}_{l = 1}^\infty$ forms a complete orthogonal basis in  $L^2((0,R), \,r \,dr)$. Hence, the above estimate actually implies that the eigenfunctions of $\tilde{K}^{(1)}_m $ approach in the sup-norm to an orthogonal basis in $L^2((0,R), \,r \,dr)$ for all $m \in \mathbb{N}$. From \eqref{symmetryeq}, together with the fact that $a_{-m,l} = a_{m,l}$ from \eqref{symmetry}, the above analysis also holds for $\tilde{K}^{(1)}_{-m}$.

The following theorem summarizes the main eigenvalue and eigenfunction estimates for the operator $\tilde{K}^{(1)}_m $.

\begin{Theorem}
For all $m \in \mathbb{Z}\backslash\{0\}$, the eigenpairs of the operator $\tilde{K}^{(1)}_m $ are of the form
\beqn
(\lambda_{m,l}, e_{m,l} ) =  \left( \lambda_{m,l}, J_m \left( \sqrt{1 - \frac{1}{ \lambda_{m,l} }} r  \right) \right)  \quad \text{ for } l \in \mathbb{N} \,,
\eqn
where the eigenvalues $\lambda_{m,l}$ satisfy the following asymptotic behavior
\beqn
\lambda_{m,l}  \bigg \slash \left( \frac{4 R^2}{ \pi^2}\frac{ 1 }{(|m| + 2l)^2} \right) \rightarrow - 1 \quad \text{ as } l \rightarrow \infty\,.
\label{decayeigeigeig}
\eqn
Moreover, the eigenfunctions also have the following uniform estimate:
\beqn
\left|\left| J_m \left( \sqrt{1 - \frac{1}{ \lambda_{m,l} }} r  \right) - J_m \left( \frac{a_{m,l}}{R} r  \right) \right|\right|_{\mathcal{C}^0 ((0,R))}= O ( \left( |m|  + 2 l  \right)^{-1/2} ) \,.
\label{esteigfun}
\eqn
\end{Theorem}
This theorem %, although very simple and elementary to obtain,
is very important for the analysis of the behaviors of $W_{nm}$ and $C(\varepsilon^*,n,m)$.
Figure \ref{eigendistribution} shows the distribution of eigenvalues of $\tilde{K}^{(1)}_m$ for $R=10$ with different values of $m$.
It not only illustrates that  the spectral radius decreases as the value of $m$ increases (which agrees with the estimate \eqref{normest}); but also that, for a fixed number $l \in \mathbb{N}$, the magnitude of the $l$-th eigenvalue of $\tilde{K}^{(1)}_m$ decreases in general monotonically with respect to increment of $m$ (which agrees with \eqref{decayeigeigeig}).
Eigenfunctions of $\tilde{K}^{(1)}_m$ for some values of $m$ are also plotted in Figure \ref{eigenfunction} for a better illustration of the behaviour of eigenfunctions.

\begin{figurehere}
\center

        \hfill{}   \hfill{} \hfill{}   \hfill{}

        \hfill{}\includegraphics[clip,width=0.47\textwidth]{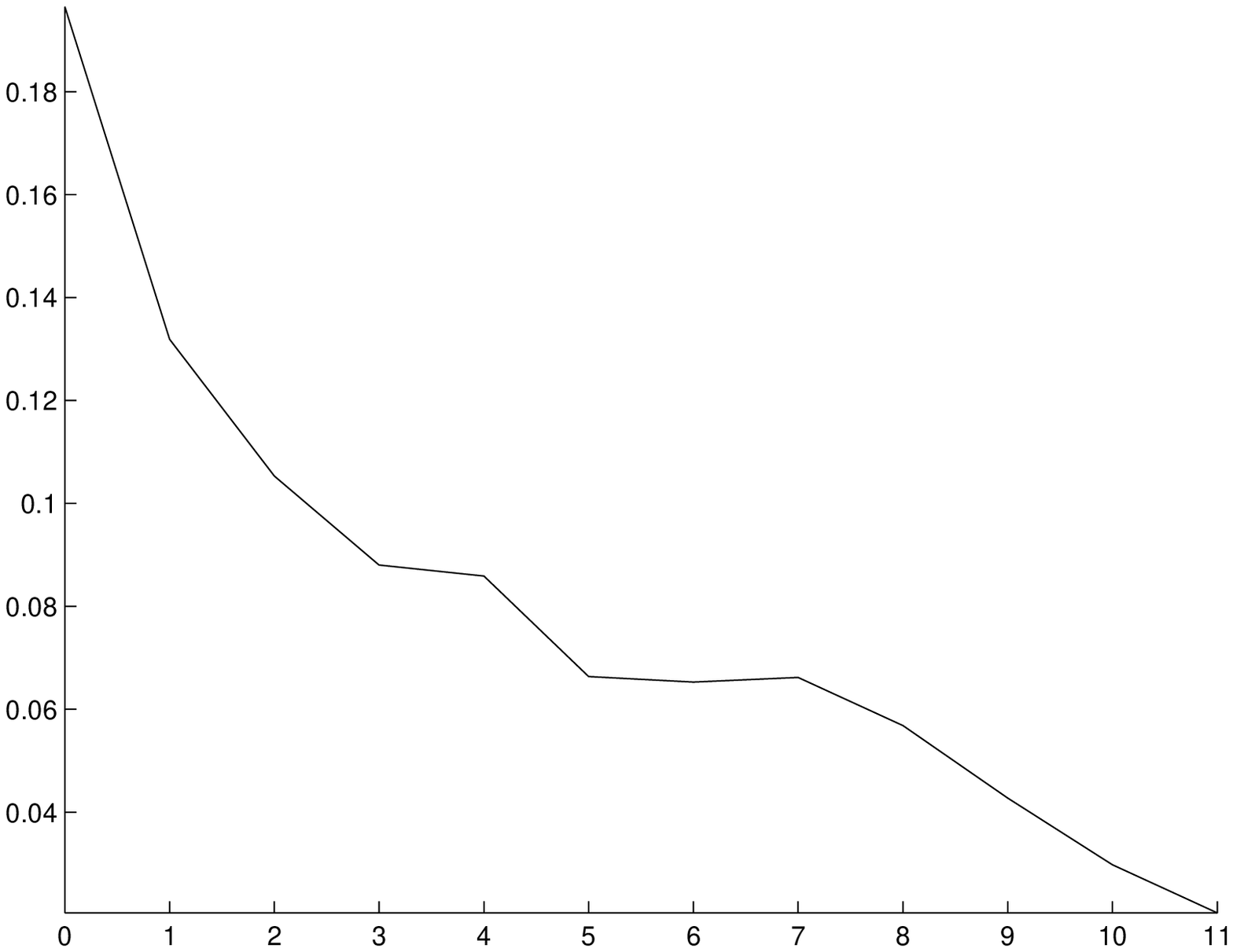}\hfill{}
        \hfill{}\includegraphics[clip,width=0.47\textwidth]{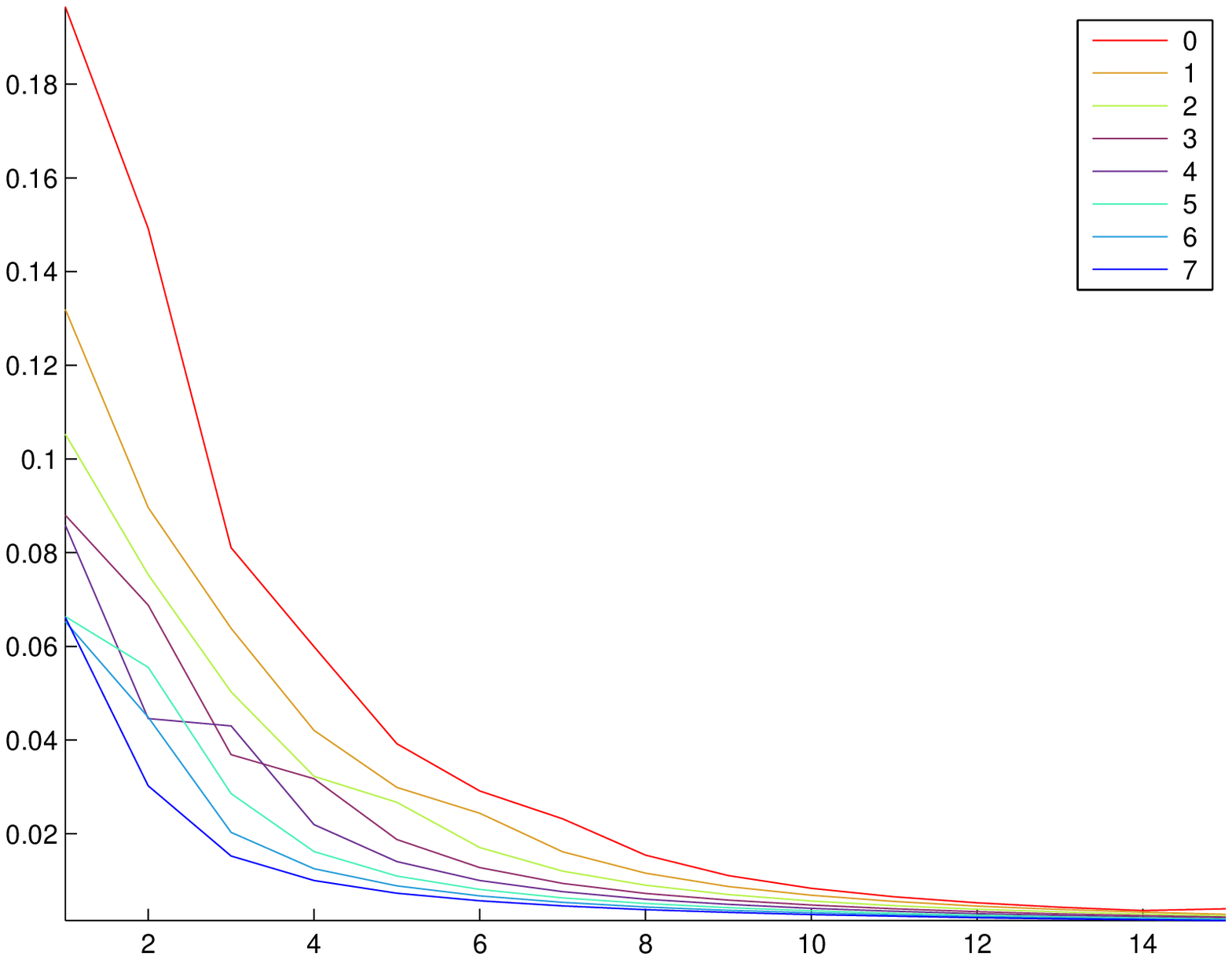}\hfill{}

        \hfill{}(a)\hfill{} \hfill{}(b)\hfill{}

\caption{
(a) Spectral radius of $\tilde{K}^{(1)}_m$ for $m=0,1,\ldots,11$.
(b) Norms of eigenvalues $\lambda_{m,l}, l = 1,2,\ldots,15,$ for operators $\tilde{K}^{(1)}_m, m=0,1,\ldots,7,$ as in the legend.
}\label{eigendistribution}
\end{figurehere}

\begin{figurehere}
\center

        \hfill{}   \hfill{} \hfill{}   \hfill{}

        \hfill{}\includegraphics[clip,width=0.47\textwidth]{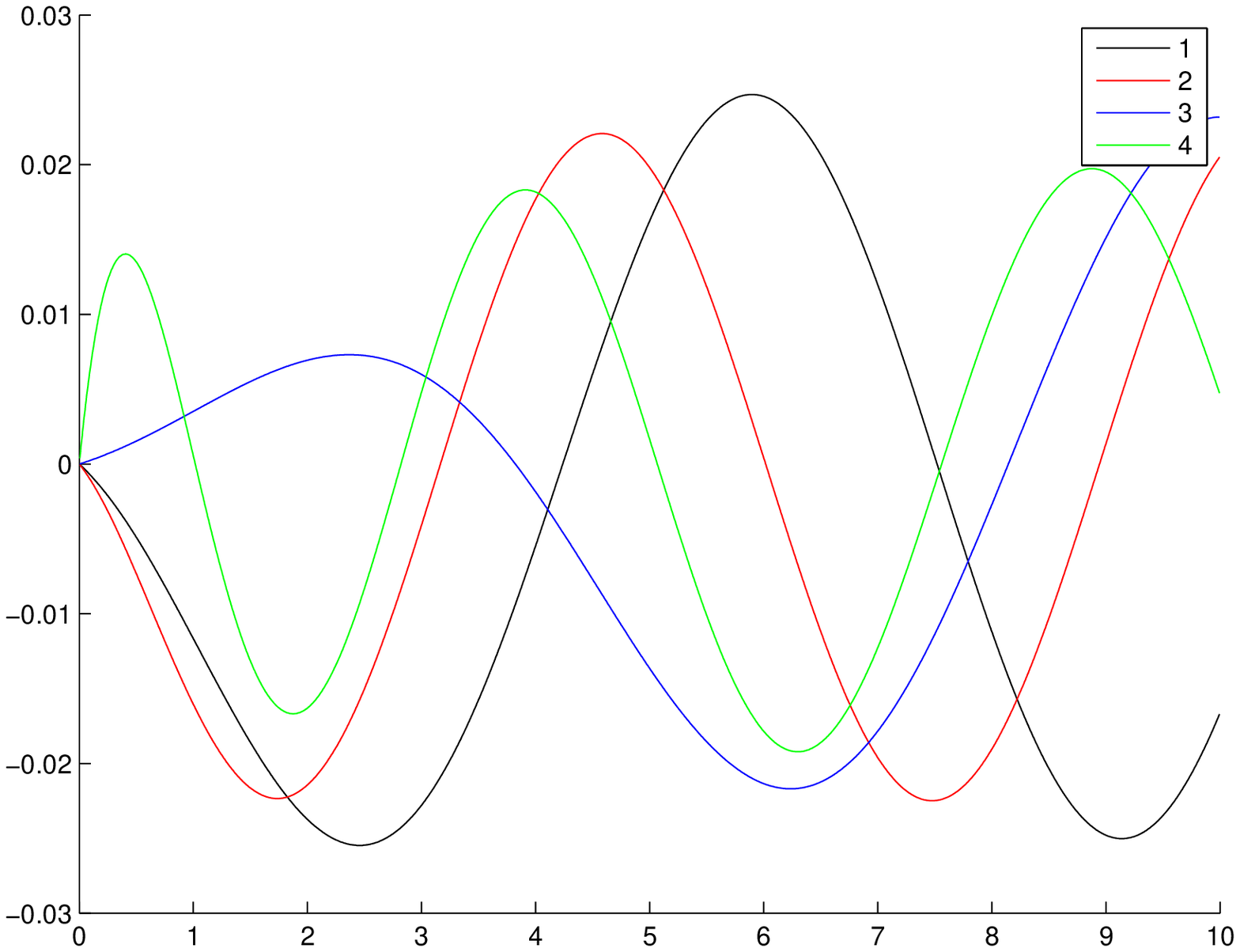}\hfill{}
        \hfill{}\includegraphics[clip,width=0.47\textwidth]{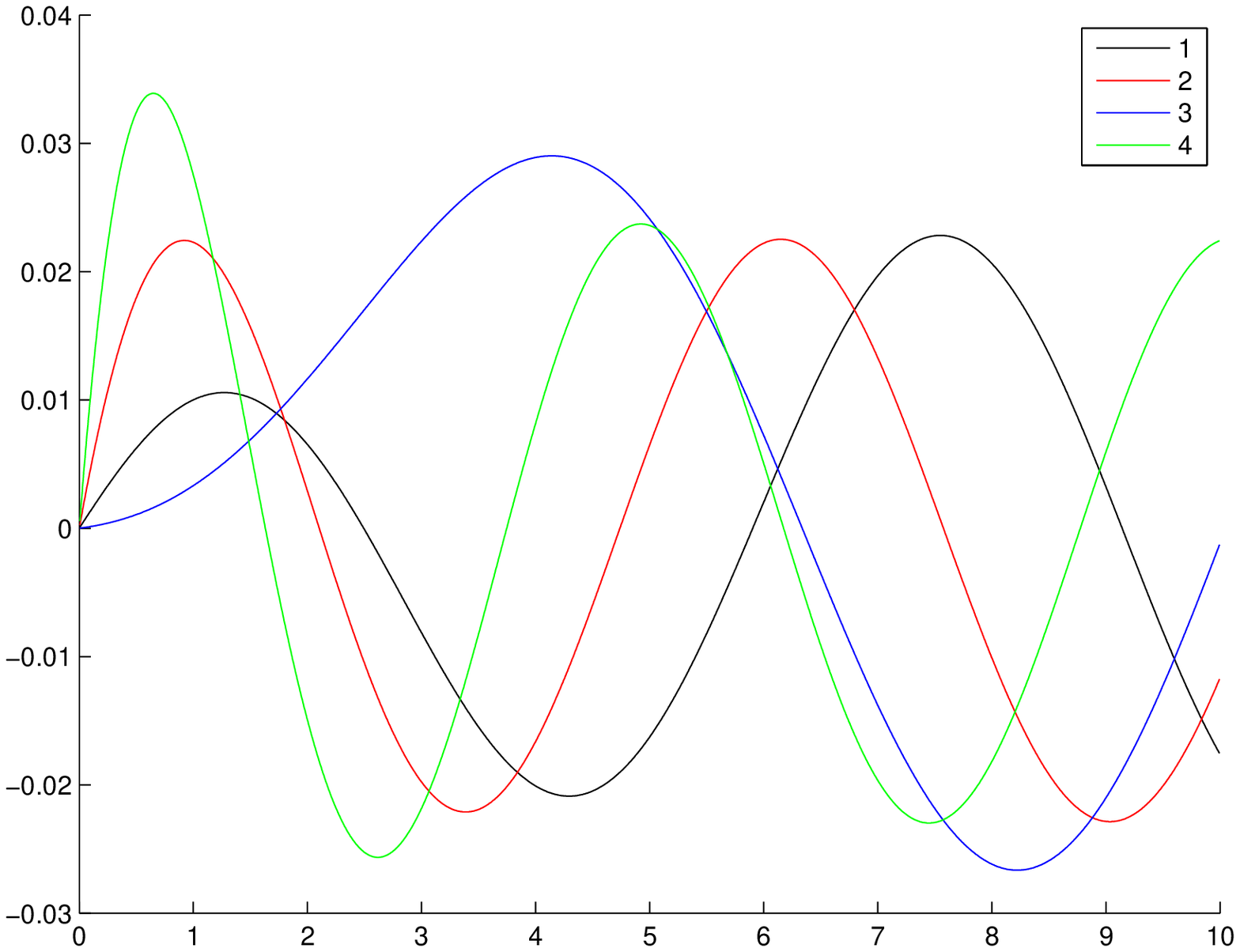}\hfill{}

        \hfill{}(1a)\hfill{} \hfill{}(1b)\hfill{}

        \hfill{}\includegraphics[clip,width=0.47\textwidth]{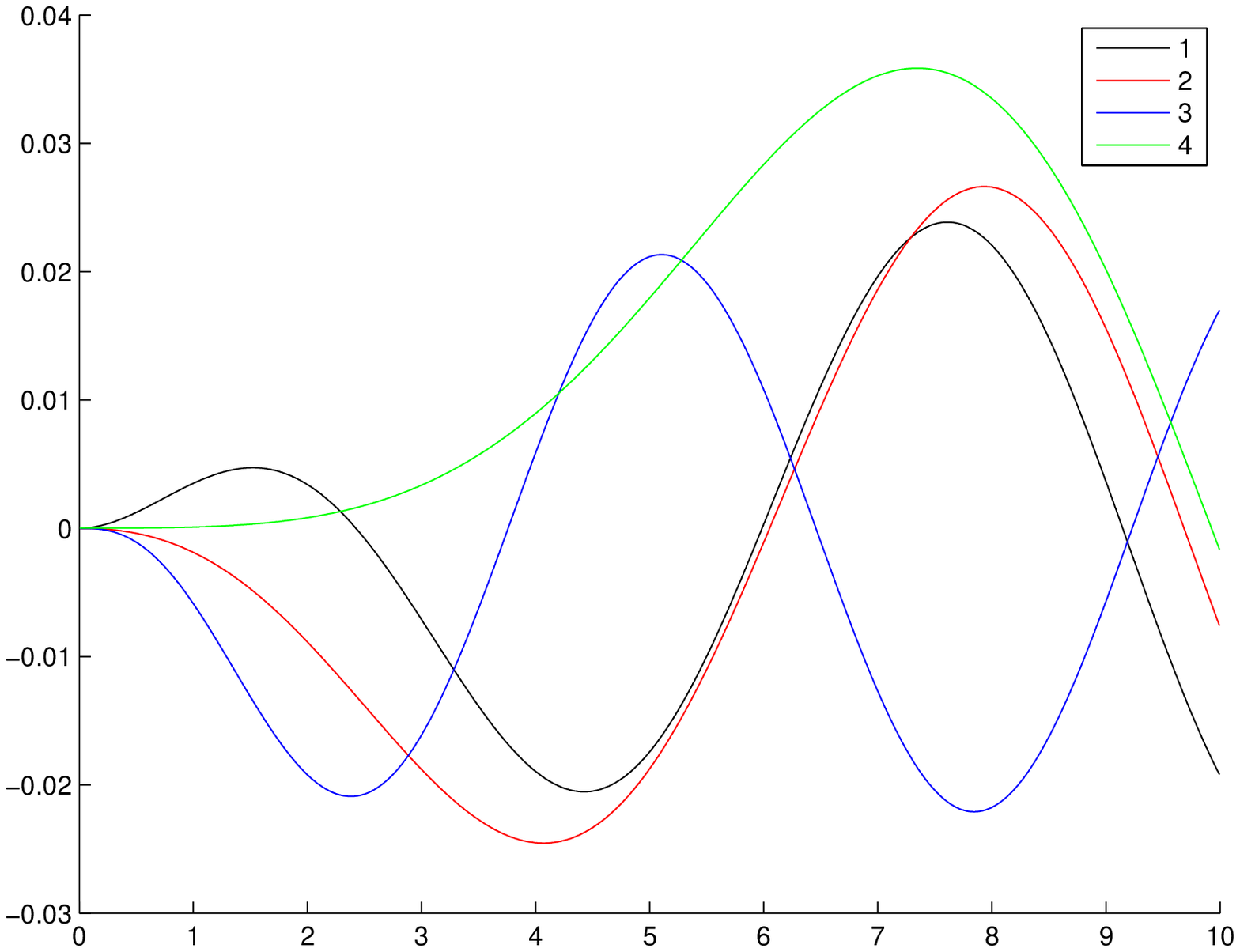}\hfill{}
        \hfill{}\includegraphics[clip,width=0.47\textwidth]{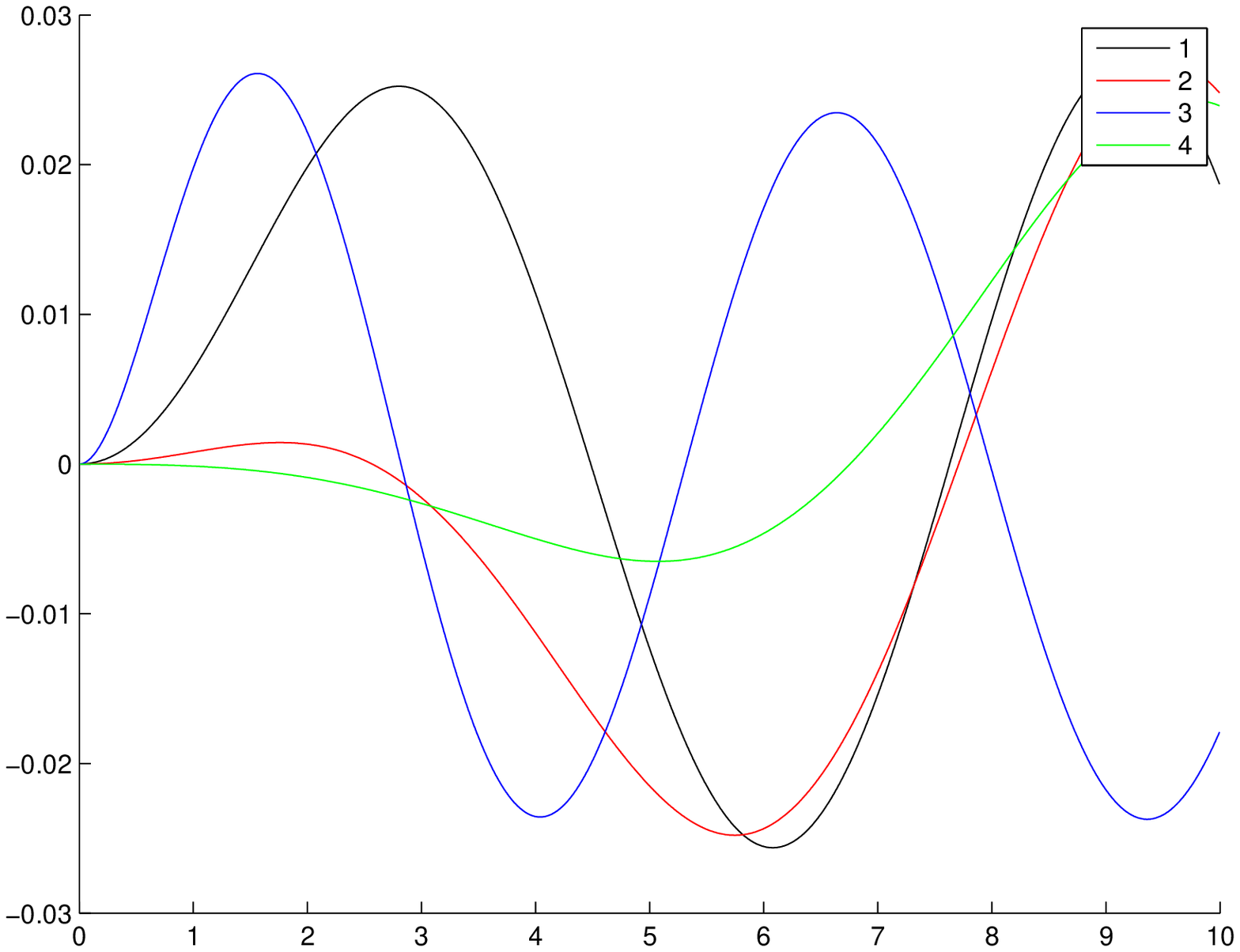}\hfill{}

        \hfill{}(2a)\hfill{} \hfill{}(2b)\hfill{}

        \hfill{}\includegraphics[clip,width=0.47\textwidth]{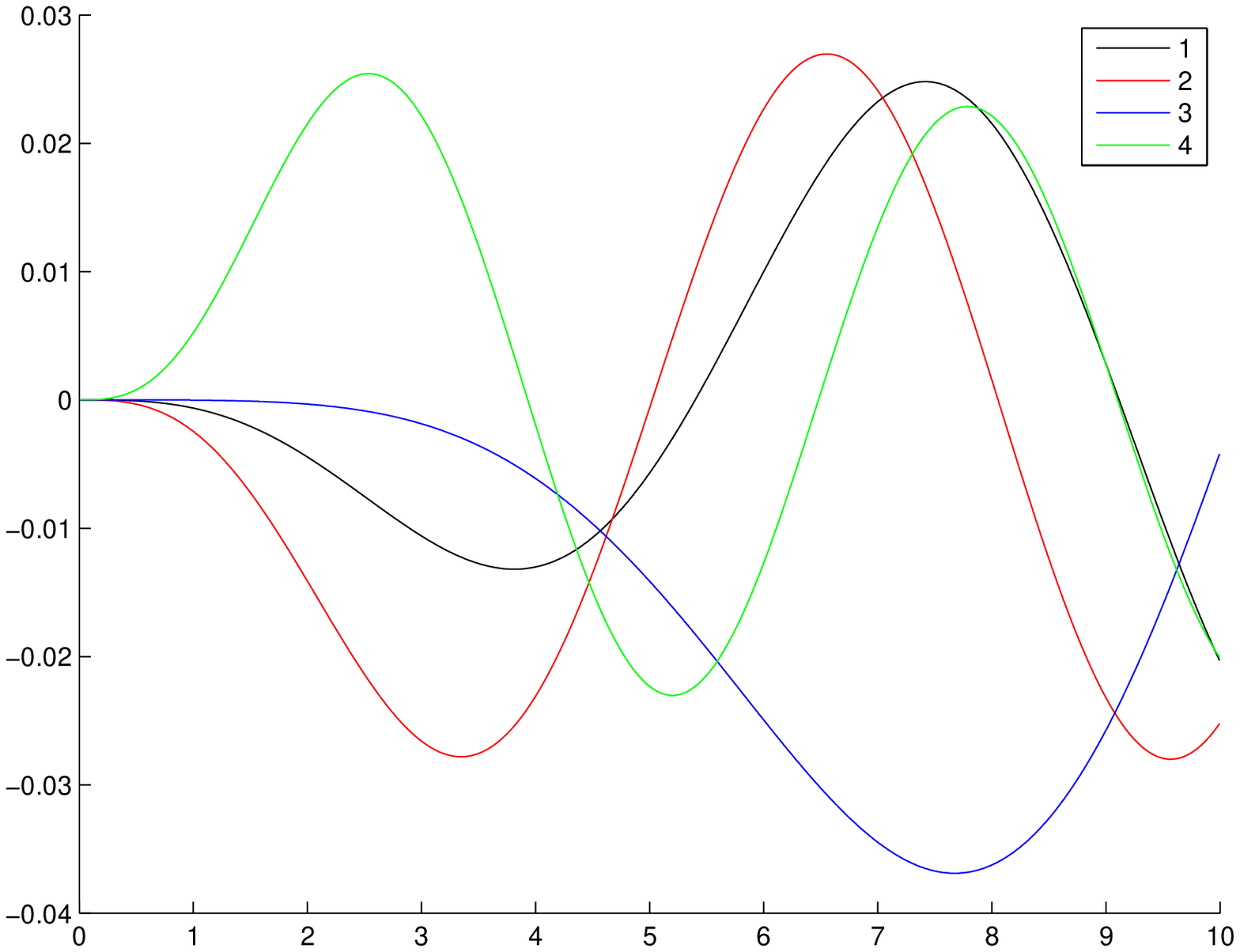}\hfill{}
        \hfill{}\includegraphics[clip,width=0.47\textwidth]{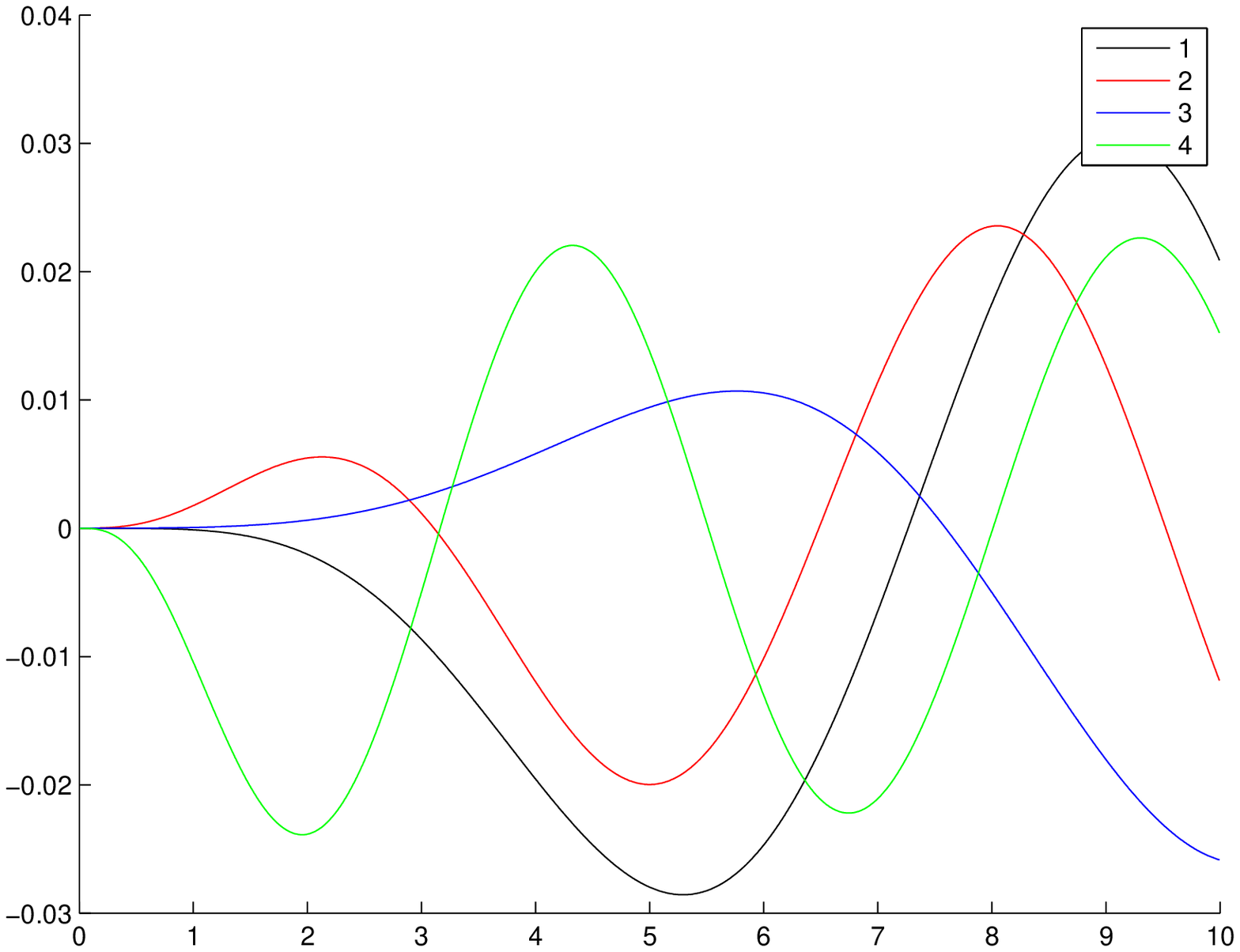}\hfill{}

        \hfill{}(3a)\hfill{} \hfill{}(3b)\hfill{}

\caption{
Real and imaginary parts of the first $4$ eigenfunctions of $\tilde{K}^{(1)}_m$, $m = 1,2,3$.
(1a) Real parts of eigenfunctions of $\tilde{K}^{(1)}_1$; (1b) imaginary parts of eigenfunctions of $\tilde{K}^{(1)}_1$;
(2a) real parts of eigenfunctions of $\tilde{K}^{(1)}_2$, and so forth.
}\label{eigenfunction}
\end{figurehere}

\subsection{Tail behavior of the series representation of $W_{nm}$ and $C(\varepsilon^*,n,m)$ and the super-resolution phenomenon}  \label{sec4_3}

In this subsection,  we deduce very useful information on the behaviors of $W_{nm}$ and $C(\varepsilon^*,n,m)$ from the asymptotic behaviors of eigenpairs of $\tilde{K}^{(1)}_m$ derived in the previous subsection.

\subsubsection{Tail behavior of the series representation of $W_{nm}$}

We first focus on the scattering coefficients $W_{nm} (D,\varepsilon^*)$ when $D = B(0,R)$.
Form \eqref{useful}, it is known that $W_{nm} = 0$ when $n \neq m$, therefore the only interesting case is when $n = m$.
Again, we shall first consider $m\in \mathbb{N}$.  From the analysis in the previous subsection that the geometric multiplicities of all the eigenvalues of $\tilde{K}^{(1)}_m$ are $N_{\lambda}^m =1$ , we already obtain %, after denoting ${\bf e}_{m,l} := {\bf e}^1_{m,\lambda_{m,l}}$
from \eqref{sumsubpart} that
\beqn
W_{mm} (D,\varepsilon^*) =
\sum_{ l = 0}^{\infty}  [ \langle J_m ( r), {\bf e}_{m,l}   \rangle_{L^2((0,R), r dr)} ]^T  [J_{m,{\varepsilon^*}^{-1} + \lambda_{m,l}}]^{-1}
 (  \, J_m(r) \, )_{{\bf e}_{m,l}, L^2((0,R), r \, dr ) } .
\notag
\eqn
For the sake of simplicity, from now on we shall often denote
\beqn
\widetilde{\lambda_{m,l}} := \frac{1}{1 - \frac{a_{m,l}^2}{R^2} }\quad \text{ and } \quad
\widetilde{{e}_{m,l}} := J_m \left( \frac{a_{m,l}}{R} r  \right)  \,.
\label{defdef}
\eqn
From \eqref{lommel1} and \eqref{esteigfun}, together with the completeness and orthogonality of $\widetilde{{e}_{m,l}}$ in $L^2((0,R) , \, r \, dr)$ and the Parseval's identity, we readily obtain that, fixing any $m \in \mathbb{N}$ and for any given $\epsilon$, there exists $N(m)$ such that for all $i > N(m) $, we have
\beqn
 \left| \langle  { e}_{m,i}  ,  \widetilde{{e}_{m,j}}  \rangle_{L^2((0,R), r dr)}  - \delta_{ij} \frac{R^2}{2} J_{m+1}^2(a_{m,j})   \right|  < \epsilon_{ij} \,,
 \label{almostorth}
\eqn
where $\sum_{j} \epsilon_{ij}^{2}  < \epsilon^2$ . Therefore, for a large $N_1(m)$, the span of $\{   { e}_{m,l} \}_{l= N_1(m)}^\infty$ has a finite dimensional orthogonal complement. This follows that there exists a large $N_2(m) > N_1(m)$ such that the algebraic multiplicity of $\lambda_{m,l}$ is $1$.
Therefore, we directly obtain
\beqn
W_{mm} (D,\varepsilon^*)  &=& S_{1,m}(\varepsilon^*) + S_{2,m} (\varepsilon^*),
\label{serieswnn}
\eqn
where the sums $S_{i,m}(\varepsilon^*)$, $i = 1,2,$ are defined by
\beqn
S_{1,m}(\varepsilon^*)
&:=&
\sum_{ l = 0}^{N_2(m)}  [ \langle J_m ( r), {\bf e}_{m,l}   \rangle_{L^2((0,R), r dr)} ]^T  [J_{m,{\varepsilon^*}^{-1} + \lambda_{m,l}}]^{-1}
 (  \, J_m(r) \, )_{{\bf e}_{m,l}, L^2((0,R), r \, dr ) }  \\
S_{2,m}  (\varepsilon^*) &:= & \sum_{ l = N_2(m) +1 }^{\infty} \frac{ \alpha_{m,l} }{ {\varepsilon^*}^{-1} + \lambda_{m,l} } \,,
\label{defSm}
\eqn
with the coefficients $\alpha_{m,l}$ being defined, for all $m,l$, as
\beqn
\alpha_{m,l} :=\langle J_m ( r), { e}_{m,l}   \rangle_{L^2((0,R), r dr)} (  \, J_m(r) \, )_{{ e}_{m,l}, L^2((0,R), r \, dr ) }.
\eqn
Note that for any $\varepsilon^* \geq - 2 \, \text{Re} \left( \lambda^{-1}_{m,N_2(m)} \right)$, we have $|S_{1,m}(\varepsilon^*)| < C_m $ for some constant $C_m$.
Therefore, if we want to investigate the behavior of \eqref{serieswnn} for large $\varepsilon^*$, we shall focus on the term $S_{2,m}(\varepsilon^*)$.
For this purpose, we analyse the limiting behavior of $\alpha_{m,l}$ as $l$ increases.
Now, from \eqref{lommel2} and \eqref{esteigfun}, we have the following estimate for the inner product:
\beqn
\langle J_m ( r), { e}_{m,l}   \rangle_{L^2((0,R), r dr)} -   \widetilde{\lambda_{ml}} a_{m,l} J_m(R) J_{m-1}( a_{m,l} )  = O (a_{m,l}^{-1/2} ) \,.
\eqn
From \eqref{order0} we get
\beqn
 J_{m \pm 1} \left( a_{m,l} \right) - (-1)^{l} \sqrt{\frac{2}{ \pi a_{m,l}}}     =  O( { a_{m,l}}^{-3/2}) \,,
\eqn
and hence it follows that
\beqn
\langle J_m ( r), { e}_{m,l}   \rangle_{L^2((0,R), r dr)} \bigg \slash   (-1)^{l} \widetilde{\lambda_{m,l}} a_{m,l}^{1/2} \sqrt{\frac{2}{ \pi }}  J_m(R)   \rightarrow 1 \quad \text{ as } l \rightarrow \infty \,.
\label{limit1}
\eqn
From \eqref{almostorth}, we obtain that the coefficient of $J_m(r)$ of ${e}_{m,l}$ with respect to the Jordan basis approaches to the orthogonal project of $J_m(r)$ on the subspace ${e}_{m,l}$, whence the following holds
\beqn
(  \, J_m(r) \, )_{{ e}_{m,l}, L^2((0,R), r \, dr ) }
\bigg \slash
\frac{\langle J_m ( r), { e}_{m,l}   \rangle_{L^2((0,R), r dr)} }{ \frac{R^2}{2} [ J_{m-1}^2(a_{m,l})  ] } \rightarrow 1 \quad \text{ as } l \rightarrow \infty \,.
\label{limit2}
\eqn
Combining the above several limiting behaviors \eqref{limit1} and \eqref{limit2} yields
\beqn
\alpha_{m,l} \bigg \slash
2 \widetilde{\lambda_{m,l}}^2 a_{m,l}^2
\frac{ J_m^2(R) }{ R^2 }
\rightarrow 1 \quad \text{ as } l \rightarrow \infty \,,
\eqn
which can further be reduced to the following asymptotic behavior by combining  \eqref{estdone},\eqref{decayeig} and \eqref{defdef},
\beqn
\alpha_{m,l}
\bigg \slash
2 \lambda_{m,l} J_m^2(R)
\rightarrow - 1 \quad \text{ as } l \rightarrow \infty \,.
\eqn
From \eqref{symmetry} and \eqref{symmetryeq}, the conclusions also hold for the case with $-m \in \mathbb{N}$.

The above analysis can be summarized in the following theorem.
\begin{Theorem}
\label{important}
Let $D = B(0,R)$ be a circular domain.
For all $m \in \mathbb{Z} \backslash \{0\}$, there exist constants $N(m) \in \mathbb{N}$ and $C_m > 0$ such that, for any given contrast value $\varepsilon^* > -2 \, \text{Re} \left(  \lambda^{-1}_{m,N(m)} \right)$, the scattering coefficient $W_{mm}( D, \varepsilon^* )$ has the following decomposition
\beqn
W_{mm} (D,\varepsilon^*)  &=&  S_{1,m}(\varepsilon^*) + S_{2,m} (\varepsilon^*),
\eqn
where $S_{1,m}(\varepsilon^*)$ has a uniform bound
\beqn
|S_{1,m}(\varepsilon^*)| < C_m  \,,
\eqn
whereas $S_{2,m} (\varepsilon^*) $ is of the form
\beqn
S_{2,m}  (\varepsilon^*) = \sum_{ l = N_2(m) +1 }^{\infty} \frac{ \alpha_{m,l} }{ {\varepsilon^*}^{-1} + \lambda_{m,l} },
\eqn
where the coefficients $\alpha_{m,l}$ have the following limiting behavior
\beqn
\alpha_{m,l}
\bigg \slash
2 \lambda_{m,l} J_m^2(R)
\rightarrow - 1 \quad \text{ as } l \rightarrow \infty \,.
\label{convalpha}
\eqn
\end{Theorem}

This decomposition of the coefficient $W_{mm}$ gives us a clear picture of the behavior of $W_{mm}$ as $\varepsilon^*$ grows.
When $\varepsilon^*$ increases, ${\varepsilon^*}^{-1}$ passes through the values $ - \text{Re} (\lambda_{m,l}) \sim (|m| + 2l)^{-2} $ for large $l$.
If $\lambda_{m,l} \in \mathbb{R}$, ${\varepsilon^*}^{-1}$ directly passes through the pole. Therefore $W_{mm}$ grows from a finite value rapidly to a directional complex infinity $\infty e^{i\theta}$ for some $\theta$, and then comes back from $ - \infty e^{i\theta}$  to a finite value after ${\varepsilon^*}^{-1}$ passes through it.
Otherwise, if $\lambda_{m,l} \notin \mathbb{R} $, then ${\varepsilon^*}^{-1}$ does not directly hit the pole.
However, since $\lambda_{m,l} \sim - (|m| + 2l)^{-2} $ where $ (|m| + 2l)^{-2}$ are real, $\text{Im} (\lambda_{m,l})$ is very small for large $l$. Hence, as ${\varepsilon^*}^{-1}$ moves close to $- \text{Re} (\lambda_{m,l})$, it comes close to the pole.
Therefore, $W_{mm}$ grows from a comparably small value very rapidly to a complex value of very large modulus, and then drops back to a small value after passing through $- \text{Re} (\lambda_{m,l})$.
The behavior of $W_{mm}$ is consequently very oscillatory as $\varepsilon^*$ grows.
Moreover, from \eqref{convalpha} we have for a fixed pair of $m,l$ that
\beqn
\frac{ \alpha_{m,l} }{ {\varepsilon^*}^{-1} + \lambda_{m,l} } \rightarrow - 2 J_m^2(R)
\eqn
as $\varepsilon^* \rightarrow \infty$, therefore that there is no hope on any convergence behavior of $W_{mm}$ as $\varepsilon^*$ grows to infinity.

Furthermore, from \eqref{decayeigeigeig} that the asymptotic $\lambda_{m,l} \sim - (|m| + 2l)^{-2} $ holds and the limit comparison test,  we have for a fixed $\varepsilon^* > - 2 \, \text{Re} \left( \lambda^{-1}_{m,N(m)} \right)$ that
\beqn
| W_{nm} (D,\varepsilon^*)| &\leq& \delta_{nm} \left( C_m +
\f{C_m'}{d( - {\varepsilon^*}^{-1} , \sigma (\tilde{K}^{(1)}_m) ) }
J_m^2(R) \sum_{l=0}^\infty  | \lambda_{m,l}| \right)
\\ &\leq&  \delta_{nm}  \left( C_m +
\f{C_m' }{d( - {\varepsilon^*}^{-1} , \sigma (\tilde{K}^{(1)}_m) ) }
\frac{R^{|m|+|n|}}{{|m|}^{|m|} {|n|}^{|n|}} \right) \,.
\eqn

\begin{Corollary}
Let $D = B(0,R)$.
For all $m \in \mathbb{Z} \backslash \{0\}$, there exist constants $N(m) \in \mathbb{N}$ and $C_{i,m}$, $i=1,2$ such that, for any given contrast value $\varepsilon^* > - 2 \, \text{Re} \left( \lambda^{-1}_{m,N(m)} \right)$, the scattering coefficient $W_{nm}(D, \varepsilon^* )$ satisfies the following estimate  for all $n \in \mathbb{Z}$,
\beqn
| W_{nm} (D,\varepsilon^*)|
\leq  \delta_{nm} \left( C_{1,m} +
\f{C_{2,m}}{d( - {\varepsilon^*}^{-1} , \sigma (\tilde{K}^{(1)}_m) ) }
\frac{R^{|m|+|n|}}{{|m|}^{|m|} {|n|}^{|n|}} \right) \,.
\eqn
\end{Corollary}
This clearly improves the estimate \eqref{boundspecific}.
%Note that we can now totally drop the exponential term that we obtained in the inequality \eqref{boundspecific}.

\subsubsection{Tail behavior of the series representation of $C(\varepsilon^*,n,m)$} %and the phenomenon of super-resolution

We now focus on the behaviours of the coefficients $C(\varepsilon^*,n,m)$, which will
help us to understand the phenomenon of super-resolution.
We first focus on the case when $n,m \in \mathbb{N}$. %The case when $n,m \in \mathbb{Z}$ is similar.
We recall the expression of the coefficient $C(\varepsilon^*,n,m)$ in \eqref{coefficient}:
\beqn
C(\varepsilon^*,n,m) := {\varepsilon^*}^{-1}
\left[ \left( {\varepsilon^*}^{-1} + \tilde{K}^{(1)}_m\right)^{-1}  [J_m] \right]  (R)   \left[ \left({\varepsilon^*}^{-1} + \tilde{K}^{(1)}_n\right)^{-1}  [J_n] \right]  (R)\,.
\notag
\eqn
It remains to study the term $\left( {\varepsilon^*}^{-1} + \tilde{K}^{(1)}_m\right)^{-1} [J_m]  (R)$. From the previous subsection, the geometric multiplicities of all the eigenvalues of $\tilde{K}^{(1)}_m$ are $N_{\lambda}^m =1$, and the algebraic multiplicities of eigenvalues $\lambda_{m,l}$ of $\tilde{K}^{(1)}_m$ are also $1$ for $l> N_2(m)$ (see Theorem \ref{important}).
Together with the regularity of $J_m$, we readily obtain as in the previous subsection that
\beqn
C(\varepsilon^*,n,m) =  {\varepsilon^*}^{-1} ( s_{1,n}(\varepsilon^*) + s_{2,n} (\varepsilon^*)) ( s_{1,m}(\varepsilon^*) + s_{2,m} (\varepsilon^*)),
\label{seriesCnm}
\eqn
where the sums $s_{i,m}(\varepsilon^*)$, ($i = 1,2$) are defined by
\beqn
s_{1,m}(\varepsilon^*)
&:=&
\sum_{ l = 0}^{N_2(m)}  ( {\bf e}_{m,l}(R)  )^T  [J_{m,{\varepsilon^*}^{-1} + \lambda_{m,l}}]^{-1}
 (  \, J_m(r) \, )_{{\bf e}_{m,l}, L^2((0,R), r \, dr ) }  ,\\
s_{2,m}  (\varepsilon^*) &:= & \sum_{ l = N_2(m) +1 }^{\infty} \frac{ \beta_{m,l} }{ {\varepsilon^*}^{-1} + \lambda_{m,l} } \,
\label{defsm}
\eqn
with the coefficients $\beta_{m,l}$ being given for all $m,l$ by
\beqn
\beta_{m,l} :=  (  \, J_m(r) \, )_{{ e}_{m,l}, L^2((0,R), r \, dr )  }  J_m \left( \sqrt{1 - \frac{1}{ \lambda_{m,l} }} R  \right).
\eqn
Similarly to the previous subsection, for any $\varepsilon^* \geq - 2 \, \text{Re} \left( \lambda^{-1}_{m,N_2(m)} \right)$, we have $|s_{1,m}(\varepsilon^*)| < C_m $ for some constant $C_m$.
Therefore, we can study the behavior of \eqref{seriesCnm} for large $\varepsilon^*$ by investigating the limiting behavior of $\beta_{m,l}$ in the series $s_{2,m} (\varepsilon^*)$ .

Substituting \eqref{order}, \eqref{estdone}  and \eqref{decayeig} into \eqref{root}, we readily derive
\beqn
J_m \left( \sqrt{1 - \frac{1}{ \lambda_{m,l} }} R  \right) \bigg\slash
(-1)^l \f{i}{4} \lambda_{m,l} a_{m,l}^{1/2}  H^{(1)}_m (R) J_m(R)  \sqrt{\frac{2 R }{ \pi} }
\rightarrow 1 \quad \text{ as } l \rightarrow \infty \,.
\label{limit11}
\eqn
Together with \eqref{limit1} and \eqref{limit2}, we conclude that
\beqn
\beta_{m,l}
\bigg \slash
\f{i}{2} \sqrt{R} \lambda_{m,l} J_m^2(R) H^{(1)}_m (R)
\rightarrow - 1 \quad \text{ as } l \rightarrow \infty \,.
\eqn
Combining the above results with \eqref{symmetry} and \eqref{symmetryeq}, we obtain the following decomposition of $C(\varepsilon^*,n,m)$.
\begin{Theorem}
\label{impthm}
Let $D = B(0,R)$ be a circular domain.
For all $p \in \mathbb{Z} \backslash \{0\}$, there exist constants $N(p) \in \mathbb{N} $ and $C_p > 0$ such that, for
any $n, m \in \mathbb{Z} \backslash \{0\}$
and any contrast value $\varepsilon^* > - 2 \, \max \left\{ \text{Re}  \left( \lambda^{-1}_{n,N(n)} \right), \text{Re} \left(  \lambda^{-1}_{m,N(m)} \right) \right\}$, the coefficient $C(\varepsilon^*,n,m)$
\eqref{coefficient} admits the following decomposition:
\beqn
C(\varepsilon^*,n,m) = {\varepsilon^*}^{-1} ( s_{1,n}(\varepsilon^*) + s_{2,n} (\varepsilon^*)) ( s_{1,m}(\varepsilon^*) +s_{2,m} (\varepsilon^*))\,.
\eqn
For all $p \in \mathbb{Z} \backslash \{0\}$, $s_{1,p}(\varepsilon^*)$ satisfies the uniform bound
\beqn
|s_{1,p}(\varepsilon^*)| < C_p \,,
\eqn
whereas $s_{2,p} (\varepsilon^*)$ is given by
\beqn
s_{2,p}  (\varepsilon^*) = \sum_{ l = N_2(p) +1 }^{\infty} \frac{ \beta_{p,l} }{ {\varepsilon^*}^{-1} + \lambda_{p,l} }
\eqn
where the coefficients $\beta_{p,l}$ have the following limiting behavior
\beqn
\beta_{p,l}
\bigg \slash
\f{i}{2} \sqrt{R} \lambda_{p,l} J_p^2(R) H^{(1)}_p (R)
\rightarrow -1 \quad \text{ as } l \rightarrow \infty \,.
\label{convalpha2}
\eqn
\end{Theorem}

Similarly to the previous subsection, the aforementioned decomposition of $C(\varepsilon^*,n,m)$ clearly illustrates the behavior of $C(\varepsilon^*,n,m)$ as $\varepsilon^*$ grows and ${\varepsilon^*}^{-1}$ passes through the values $ - \text{Re} (\lambda_{p,l}) \sim (|p| + 2l)^{-2} $ with $p = n,m$.
If $\lambda_{p,l} \in \mathbb{R}$, ${\varepsilon^*}^{-1}$ directly hits the pole. Therefore $C(\varepsilon^*,n,m)$ first grows from a finite value rapidly to a directional complex infinity $\infty e^{i\theta} $ for some $\theta$, then back from $- \infty e^{i\theta}$ to a finite value after passing through it.
Otherwise if $\lambda_{p,l} \notin \mathbb{R} $ and when $l$ is large, ${\varepsilon^*}^{-1}$ does not pass through the pole, but comes very close to it. Hence, $C(\varepsilon^*,n,m)$ grows rapidly from a considerably small value to a complex value of very large modulus, then drops to a small value after passing through $- \text{Re} (\lambda_{p,l})$.
Moreover, for a fixed pair of $p,l$, we have
\beqn
\frac{ \beta_{p,l} }{ {\varepsilon^*}^{-1} + \lambda_{p,l} } \rightarrow - \f{i}{2} \sqrt{R} J_p^2(R) H^{(1)}_p (R)
\eqn
as $\varepsilon^* \rightarrow \infty$.  Therefore, we can see that $C(\varepsilon^*,n,m)$ has very oscillatory behavior as $\varepsilon^*$ grows.

\subsection{The super-resolution phenomenon}
Although $C(\varepsilon^*,n,m)$ is very oscillatory as $\varepsilon^*$ grows, the aforementioned behavior and series decomposition of $C(\varepsilon^*,n,m)$ gives a clear explanation of the super-resolution phenomenon for high-contrast inclusions.
It is because, what we have actually proved is that, in the shape derivative of the scattering coefficients of a circular domain,
there are simple poles corresponding to the complex resonant states, and therefore peaks at the real parts of these resonances.
Hence, as the material contrast $\varepsilon^*$ increases to infinity and is such that it hits the real part of a resonance,
the sensitivity in the scattering coefficients becomes very large and super-resolution for imaging occurs.

To put it more accurately, let us recall \eqref{divcircle2}. Suppose $D=B(0,R)$, then for any $\delta$-perturbation of $D$, $D^\delta$, along the variational direction $h \in \mathcal{C}^1(\p D)$, we have
\beqn
W_{nm}( D^\delta, \varepsilon^*) -W_{nm}( D, \varepsilon^*) = \delta \, C(\varepsilon^*,n,m) \mathfrak{F}_\theta \left[ h \right](n-m) + O(\delta^2) \,.
\notag
\eqn
As one might recall from \eqref{estimate}, $W_{nm}( D^\delta, \varepsilon^*)$ always decays exponentially as $|n|,|m|$ increase. Hence, it is always of exponential ill-posedness to recover the higher order Fourier modes of the perturbation $h$.
The inversion process to recover the $k$-th Fourier mode $\mathfrak{F}_\theta \left[ h \right](k)$ becomes less ill-posed if $C(\varepsilon^*,n,m)$ is large for some $n,m \in \mathbb{Z}$ such that $k = n-m$. This not only makes the respective scattering coefficients more apparent than the others, but also lowers the condition number of the inverse process to reconstruct the respective Fourier mode.
From the analysis in the previous subsection, this can only be made possible when $ {\varepsilon^*}^{-1} $ comes close to $ - \text{Re} (\lambda_{p,l} ) $ for some $p = n,m$ and for some $l \in \mathbb{N}$.
%Therefore, the phenomenon of super-resolution up to the $k$-th Fourier mode is only possible after $\varepsilon^* > -\text{Re}(\lambda_{k,1}^{-1})$.

Now, suppose $\varepsilon^*$ is close to the following resonant value $\left( \frac{ K \pi  }{ 2 R } - \frac{\pi}{4 R}\right)^2$ where $K \in \mathbb{N}$ is large.  Then, from the fact that the eigenvalues $\lambda_{p,l}$ of the operators $\widetilde{K}^{(1)}_{p}$ follow the asymptotics:
\beqn
- \lambda_{p,l}^{-1} \sim \left( \frac{\pi (|p| + 2l) }{ 2 R } - \frac{\pi}{4 R}\right)^2 \,,
\eqn
we see that  ${\varepsilon^*}^{-1}$ is close to $-\text{Re}(\lambda_{p,l(p)}^{-1})$ for all $p \in \mathbb{Z}$ such that $|p| + 2 \,l(p) = K$ for some $l(p) \in \mathbb{N}$.
Therefore, ${\varepsilon^*}^{-1}$ comes close to $-\text{Re}(\lambda_{K,0}^{-1}), -\text{Re}(\lambda_{K-2,1}^{-1}), -\text{Re}(\lambda_{K-4,2}^{-1}), \ldots,  -\text{Re}(\lambda_{K-2[\frac{K}{2}],[\frac{K}{2}]}^{-1})$ simultaneously where $[\cdot ]$ is the floor function.
This in turn boosts up the magnitudes of all the terms $ \frac{ \beta_{p,l(p)} }{ {\varepsilon^*}^{-1} + \lambda_{p,l(p)} }$ whenever $p$ is of the form $p = - K +2 s$, $s = 0,2,\ldots,K$. These terms dominate the series $s_{2,p}(\varepsilon^*) $, hence we obtain the following approximations of $s_{2,p}(\varepsilon^*)$ for all $p = - K +2 s$, $s = 0,2,\ldots,K$:
$$s_{2,p}(\varepsilon^*) \approx  - \f{i}{2} \sqrt{R}  J_p^2(R) H^{(1)}_p (R) \frac{ (K-0.5)^{-2} }{ 4^{-1} \pi^2  R^{-2} {\varepsilon^*}^{-1} - (K-0.5)^{-2} } \,.$$
Now we see from Theorem \ref{impthm} that the coefficients $C(\varepsilon^*,n,m)$ have the following approximations for $n,m \in \mathbb{Z}$ when $\varepsilon^*$ is very close to the resonant values $\left( \frac{ K \pi  }{ 2 R } - \frac{ \pi}{4 R}\right)^2$ for large $K$:
%\beqn
%C(\varepsilon^*,n,m) &
%\begin{cases}
%\approx    & \pi^{-2} R^3 J_m^2(R) J_n^2(R) H^{(1)}_m (R) H^{(1)}_n (R) (K-0.5)^{-6} \left( 4^{-1} \pi^2  R^{-2} {\varepsilon^*}^{-1} - (K-0.5)^{-2} \right)^{-2} \\
%& \quad \text{ if both of $n,m$ have the form of $- K +2 s$ with $s = 0,2,\ldots,K $} \,; \\
%\approx  & \pi^{-2} R^3 J_m^2(R) J_n^2(R) H^{(1)}_m (R) H^{(1)}_n (R) (K-0.5)^{-4} \left( 4^{-1} \pi^2  R^{-2} {\varepsilon^*}^{-1} - (K-0.5)^{-2} \right)^{-1} \\
%& \quad \text{ if only one of $n,m$ has the form of $- K +2 s$ with $s = 0,2,\ldots,K $} \,; \\
%& \text{is very small otherwise}.
%\end{cases}
%\notag
%\eqn
\beqn
C(\varepsilon^*,n,m) &
\begin{cases}
\approx    & M_{n,m,R} \, (K-0.5)^{-6} \left( 4^{-1} \pi^2  R^{-2} {\varepsilon^*}^{-1} - (K-0.5)^{-2} \right)^{-2} \\
& \quad \text{ if both of $n,m$ have the form $- K +2 s$, $s = 0,2,\ldots,K $} \,; \\
\approx  & M_{n,m,R} \,  (K-0.5)^{-4} \left( 4^{-1} \pi^2  R^{-2} {\varepsilon^*}^{-1} - (K-0.5)^{-2} \right)^{-1} \\
& \quad \text{ if only one of $n,m$ has the form $- K +2 s$, $s = 0,2,\ldots,K $} \,; \\
& \text{is very small otherwise},
\end{cases}
\notag
\eqn
where $M_{n,m,R}$ are some constants depending only on $n,m$ and $R$.
Here, the term $\left( 4^{-1} \pi^2  R^{-2} {\varepsilon^*}^{-1} - (K-0.5)^{-2} \right)^{-1}$ is very large, and makes the Fourier coefficients $\mathfrak{F}_\theta \left[ h \right]( n-m)$ visible for $n,m \in \{ - K +2 s \, : \, s = 1,2,\ldots,K \} $ for accurate classification of the shapes.
The above mechanism is possible only when $\varepsilon^*$ increases up to one of the resonant values $\left( \frac{ K \pi  }{ 2 R } - \frac{\pi}{4 R}\right)^2$ when $K$ is large.
This explains the increasing likelihood of obtaining super-resolution as $\varepsilon^*$ increases.

Now, for a given $\varepsilon^*$, consider the following bounded linear map over the space $l_{\pm}^2(\mathbb{C})$ of two-sided sequences $(a_l)_{l=-\infty}^{\infty}$ such that $\sum_{l=-\infty}^{\infty} a_l^2 < \infty$,
\beqn
A(\varepsilon^*) : l_{\pm}^2(\mathbb{C}) &\rightarrow& l_{\pm}^2(\mathbb{C}) \otimes l_{\pm}^2 (\mathbb{C}) \notag \\
(a_l)_{l=-\infty}^{\infty} & \mapsto &
\left( C(\varepsilon^*,n,m) \, a_{n-m} \right)_{n,m=-\infty}^{\infty} \,.
\eqn
By Theorem \ref{vari}, we know the shape derivative of $ \left( W_{nm}( D, \varepsilon^* ) \right)_{n,m=-\infty}^{\infty}$ in the variational direction $h$ is given by
\beqn
\mathcal{D} \,W ( D, \varepsilon^*) [h] = A(\varepsilon^*) \mathfrak{F}_\theta \left[ h \right] \,.
\eqn
It is ready to conclude that the least-squared map
 \beqn
[A(\varepsilon^*)]^* [A(\varepsilon^*)] : l_{\pm}^2(\mathbb{C}) &\rightarrow& l_{\pm}^2(\mathbb{C}) \notag \\
(a_l)_{l=-\infty}^{\infty} & \mapsto & \left( \sum_{n-m = l}  |C(\varepsilon^*,n,m)|^2  \, a_l \right)_{l=-\infty}^{\infty}  \,
\eqn
is a diagonal operator, and the $l$-th singular value $s_l(A)$ is of the form
\beqn
 s_l(A) = \sqrt{ \sum_{n-m = l}  |C(\varepsilon^*,n,m)|^2} \,.
\eqn
Therefore, from the above analysis on $C(\varepsilon^*,n,m)$ when $\varepsilon^*$ is close to the resonant values $\left( \frac{ K \pi  }{ 2 R } - \frac{ \pi}{4 R}\right)^2$, we can observe that the singular values $ s_{l}$ become large and comparable to each other,
making the inversion of many Fourier modes % $\mathfrak{F}_\theta \left[ h \right]( - K +2 s)$, $s = 1,2,\ldots,K, $
well-conditioned.
This implies a much higher resolution of the modes of $h$, and also for reconstructing the geometry of $D^\delta$ in the linearized case.
This provides a good understanding towards the recently observed phenomenon of super-resolution in the physics and engineering communities.

\section{Numerical experiments} \label{numerical}

In this section, we present some numerical experiments on the behaviors of the scattering coefficients for some domains as the contrast $\varepsilon^*$ grows, and numerically illustrate the phenomenon of super-resolution.

In the following $2$ examples, we consider an infinite domain of homogeneous background medium with its material coefficient being $1$. An inclusion $D^\delta$ is then introduced as a perturbation of a circular domain $D = B(0,R)$ for some $R>0$ and $\delta>0$ lying inside the homogeneous background medium, with its contrast chosen to be
$\varepsilon^* = a_{m,l}^2/R^2 -1 $ running over all $m,l$ such that $a_{m,l} \leq 18.901$. The exact values of the zeros of Bessel functions are found in \cite{zerozero}.

In order to generate the far-field data for the forward problem and the observed scattering coefficients, we use the SIES-master package  developed by H.~Wang \cite{hansite}.

The forward problem is solved by computing the solutions $(\phi_m, \psi_m)$ of \eqref{defint} for $|m|\leq 25$ using rectangular quadrature rule with mesh-size $s/1024$ along the boundary of the target, where $s$ denotes the length of the inclusion boundary. The scattering coefficients of $D^{\delta}$ of orders $(n,m)$ for $|n|,|m| \leq 25 $ are then calculated as the Fourier transform of the far-field data.

In order to test the robustness of the super-resolution phenomenon, we introduce some multiplicative random noise in the scattering coefficients in the form:
\begin{equation} W_{nm}^\gamma ( D^{\delta}, \varepsilon^*) = W_{nm} ( D^{\delta}, \varepsilon^*) \left( 1 + \gamma (\eta_1 + i \eta_2) \right), \label{noise} \end{equation}
where $\eta_i, \, i=1,2,$ are uniformly distributed between $[-1,1]$ and $ \gamma$ refers to the relative noise level.
In both examples below, we always set the noise level to be $\gamma = 5 \%$.

Since the purpose of our numerical experiments is to illustrate the phenomenon of super-resolution as $\varepsilon^*$ increases, we assume that both $R$ and $\varepsilon^*$ are known and
 use the following regularized inversion method suggested from the linearized problem \eqref{divcircle2} to recover the $k$-th Fourier mode for $|k | \leq 50$ from the observed noisy scattering coefficients $W^{\gamma}_{nm} ( D^{\delta}, \varepsilon^*), \, |n|,|m| \leq 25 $:
\beqn
\delta \mathfrak{F}_\theta \left[ h \right]^{\text{recovered}}(k) = \sum_{n-m = k, \, |n|,|m| \leq 25} \frac{ W^{\gamma}_{nm} ( D^{\delta}, \varepsilon^*) - W_{nm} ( D , \varepsilon^*) } { C(\varepsilon^{*}, n,m) + \alpha},
\eqn
where %$W^{\gamma}_{nm} (D^\delta )$ are the observed noisy scattering coefficients and
$\alpha$ is a regularization parameter.
%\footnote{
%Numerically we surprisingly observe that this inversion algorithm gives more feature than the following least squared algorithm:
%\beqn
%\delta \left[\mathfrak{F}_\theta h \right]^{\text{recovered}}(m-n) = \frac{ \sum_{n-m = l} \overline{ C(\varepsilon^{*}, n,m) } ( W^{\gamma}_{nm} (D^\delta ) - W_{nm} (D) ) }{\sum_{n-m = l}| C(\varepsilon^{*}, n,m)|^2 + \alpha} \,.
%\eqn
%}
The coefficients $ W_{nm} ( D , \varepsilon^*) $ used in the inversion process are calculated using the same method as previously mentioned for the forward problem without adding noise, and the coefficients $ C( \varepsilon^{*}, n,m) $ are calculated by the following approximations
\beqn
C( \varepsilon^{*}, n,m) \approx \left( W_{nm} (D^{\delta_0}(n-m) , \varepsilon^*) - W_{nm} (D , \varepsilon^*) \right)/\delta_0
\eqn
for $|n|,|m| \leq 25$, where $D^{\delta_0}(k)$ are defined as domains with the following boundaries for $|k|\leq 50$,
\beqn
\p D^{\delta_0}(k) := \{ \tilde{x} = R (1 + \delta_0 e^{i k \theta } ) \, : \, \theta \in (0,2\pi] \} \,,
\eqn
with $\delta_0$ chosen to be $\delta_0 = 0.1$.

\textbf{Example 1}
As a toy example, we first consider a flori-form shape $D^\delta$ described by the following parametric form (with $\delta = 0.1$ ):
    \begin{equation}
        r = 0.3 ( 1 + \delta \cos (3 \theta) + 2 \delta \cos (6\theta) + 4 \delta \cos (9\theta)   )\, , \quad \theta \in (0,2\pi]\,,
        \label{circle_perturb_form}
    \end{equation}
which is a perturbation of the domain $D:=B(0, 0.3)$; see Figure \ref{Example 1_domain} (left) for the domain and Figure \ref{Example 1_domain} (right) for the comparison between the domains $D^\delta$ and $D$.

%$51$ receivers are put to receive the far-field data from the forward problem for the scattering coefficient as the Fourier transform of the far-field data, see Figure \ref{Example 1_domain} (right).

\begin{figurehere}
\center
       \includegraphics[clip,width=0.3\textwidth]{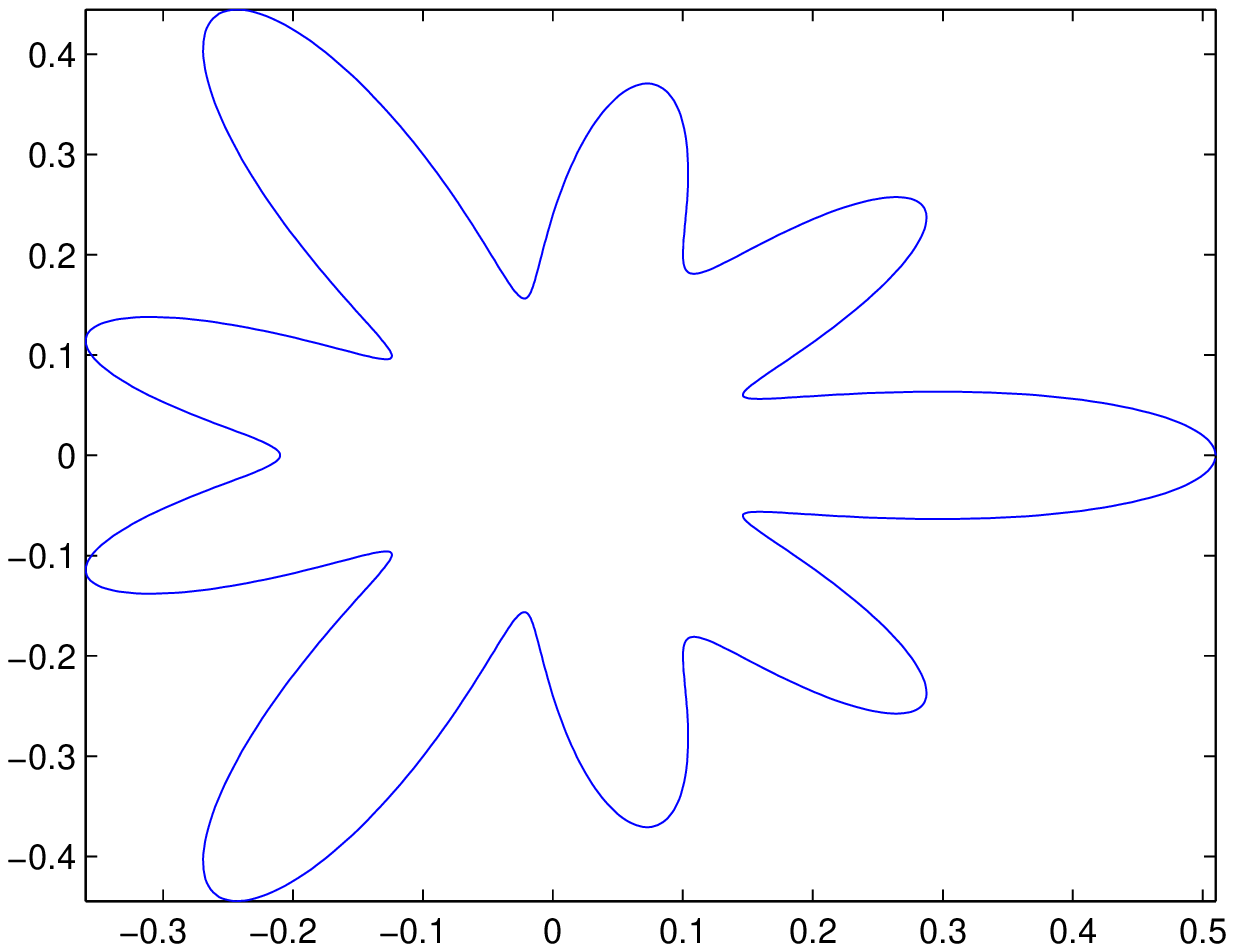}
        \includegraphics[clip,width=0.3\textwidth]{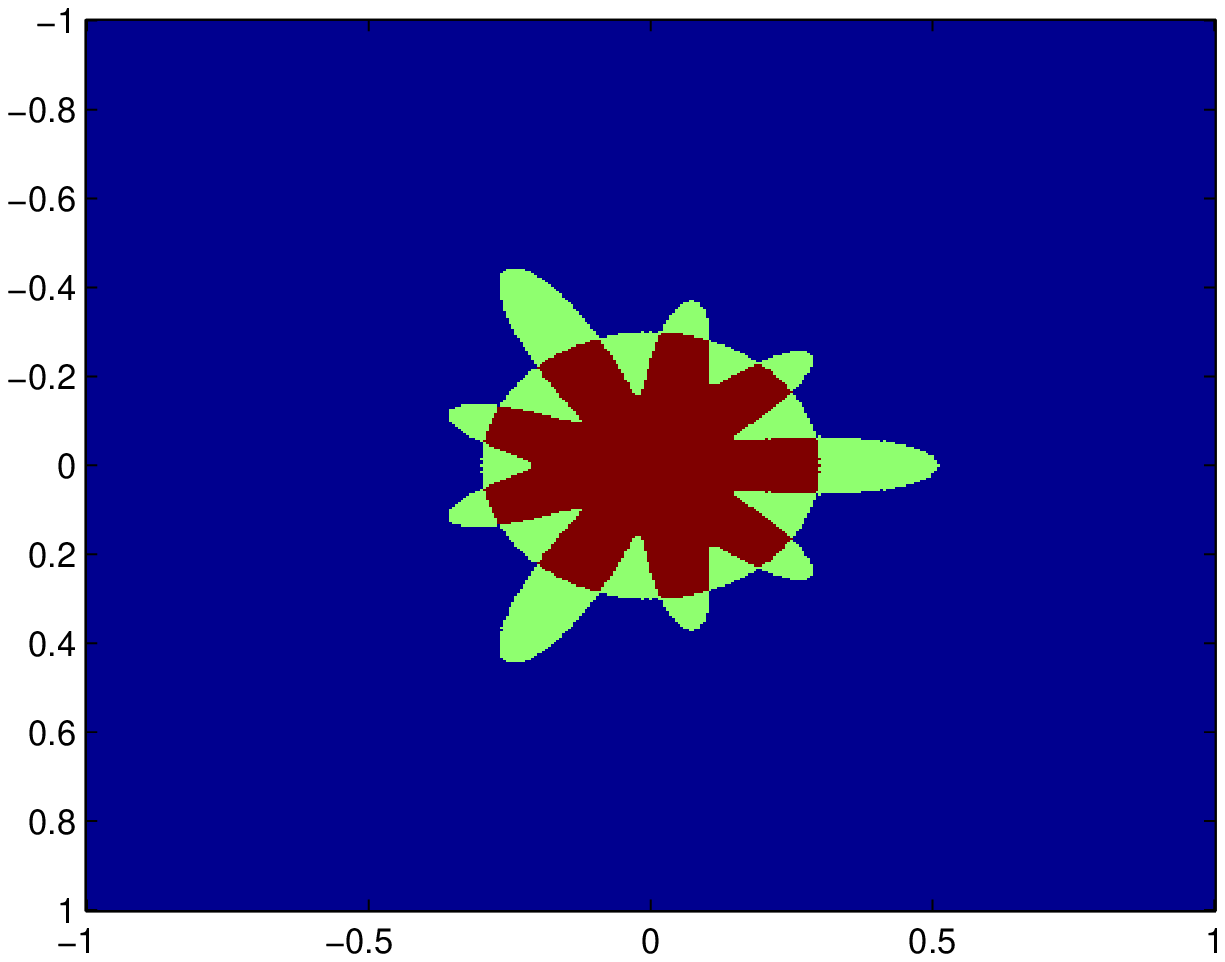}
%        \hfill{}\includegraphics[clip,width=0.3\textwidth]{figures/NEW_SC_1/NEW_SC_1_rec}\hfill{}

     %   \hfill{}(5 a)\hfill{} \hfill{}(5 b)\hfill{} \hfill{}(5 c)\hfill{}
\caption{
 Inclusion shape in Example 1.
}\label{Example 1_domain}
\end{figurehere}

The relative magnitudes of the scattering coefficients $\max_{|m-n| = k }  |W_{nm}(D^\delta,\epsilon^*)| / \max_{ m \neq n }  |W_{nm}(D^\delta,\epsilon^*)| $ are plotted for $k = 3,6,9$ in Figure \ref{Example 1_coeff}.

%where $\varepsilon^* = a_{m,l}^2/R^2 -1 $ running over all $m,l$ such that $a_{m,l} \leq 18.901$ as previously mentioned.
From Figure \ref{Example 1_coeff}, we can clearly observe that, as $\varepsilon^*$ grows, the relative magnitude of the scattering coefficient corresponding to the $\pm k$-th Fourier mode grows from a smaller magnitude to larger magnitude, and the peaks become apparent when $\varepsilon^*$ hits the respective zeros of the Bessel functions.

\begin{figurehere}
\center

        \hfill{}   \hfill{} \hfill{}   \hfill{}  \hfill{}   \hfill{}
        \hfill{}\includegraphics[clip,width=0.3\textwidth]{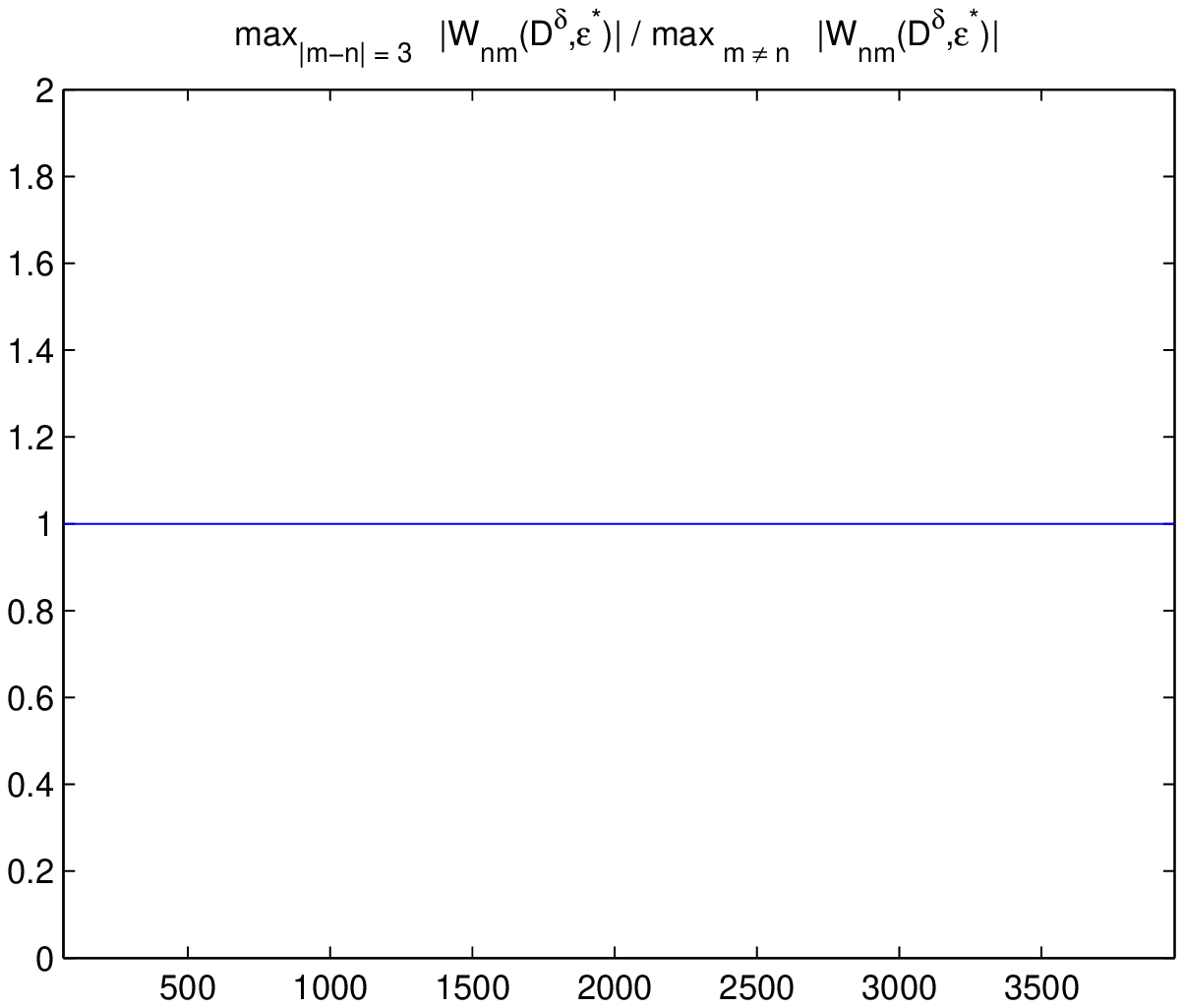}\hfill{}
        \hfill{}\includegraphics[clip,width=0.3\textwidth]{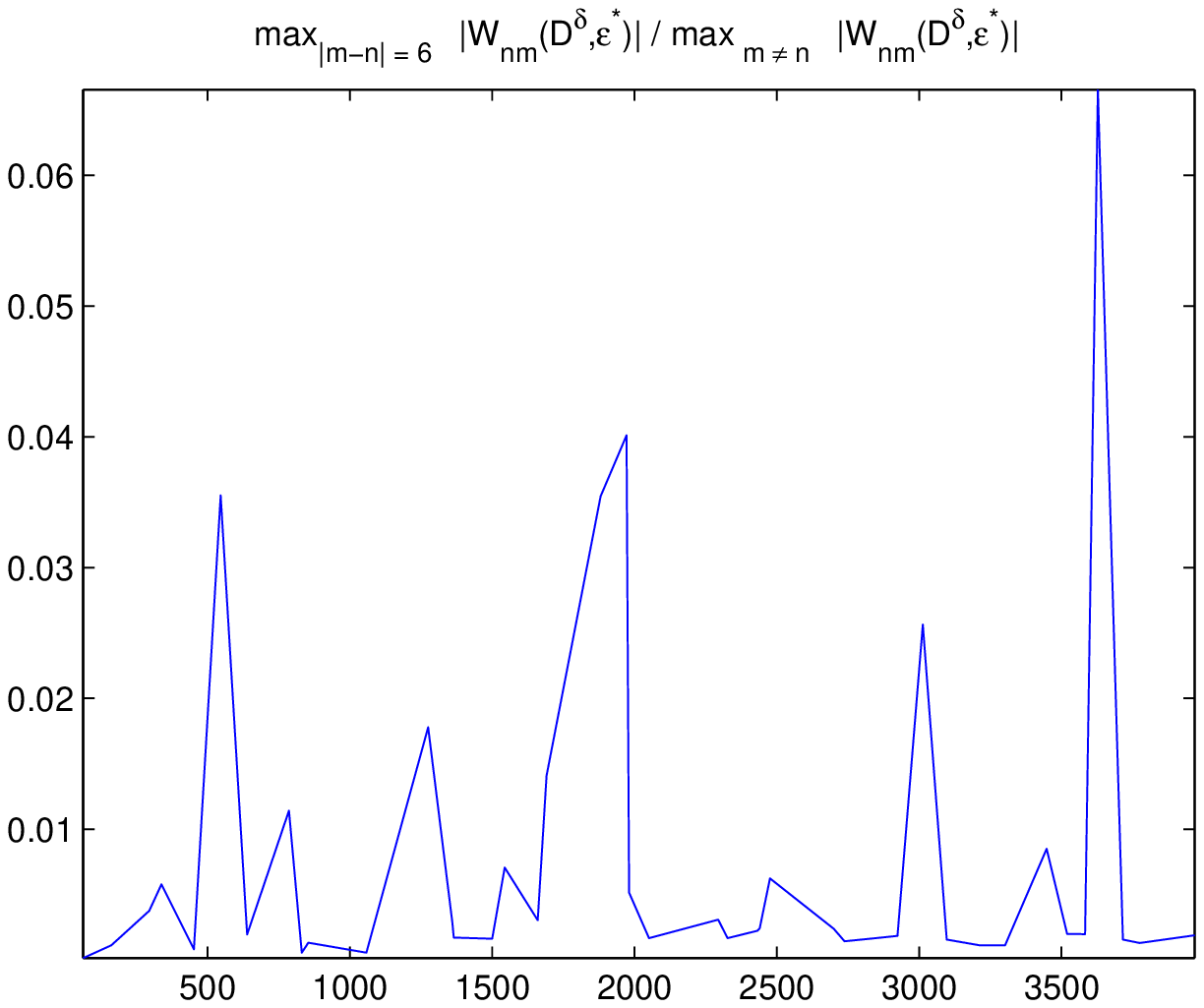}\hfill{}
        \hfill{}\includegraphics[clip,width=0.3\textwidth]{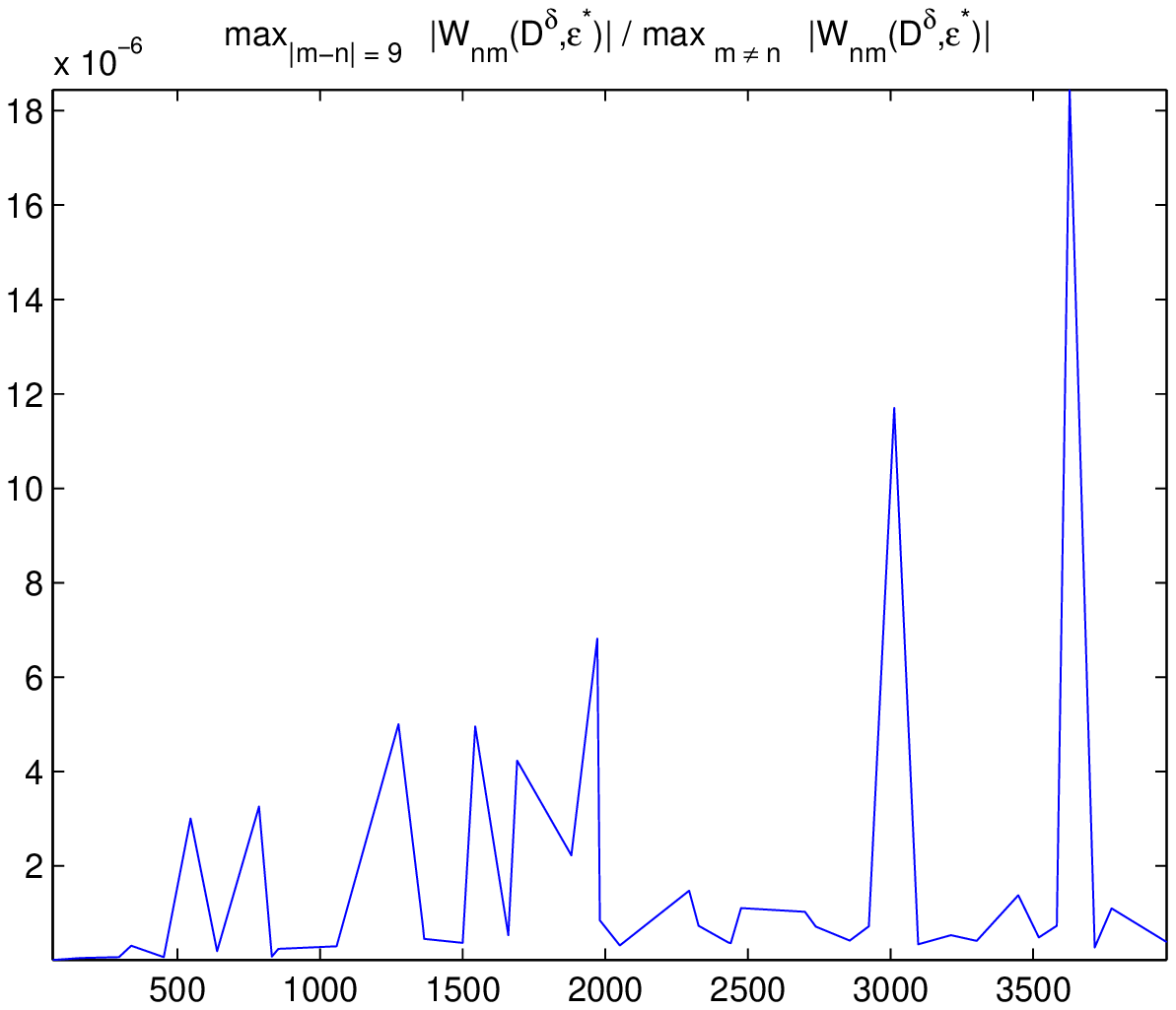}\hfill{}

     %   \hfill{}(5 a)\hfill{} \hfill{}(5 b)\hfill{} \hfill{}(5 c)\hfill{}
\caption{
Relative magnitudes of the scattering coefficients in Example 1.
}\label{Example 1_coeff}
\end{figurehere}

From the relative magnitudes shown in the above figures, we observe that the scattering coefficients are best conditioned for inversion when $\varepsilon^* =  1971.2481, 3627.456$. The scattering coefficients of the respective contrasts are plotted in Figure \ref{Example 1_recovery} (left), together with $\varepsilon^* = 63.2669$ corresponding to the first zero of $J_0$ as a comparison. We notice from the figures that the scattering coefficients corresponding to higher Fourier modes become more apparent as $\varepsilon^*$ increases. We then apply the aforementioned inversion process, with the regularization parameter chosen as $\alpha = 1 \times 10^{-8}$. The magnitudes of the recovered Fourier modes and the reconstructed domains are shown in Figure \ref{Example 1_recovery} (middle) and (right) respectively. We can clearly see that the fine features are more and more apparent as $\varepsilon^*$ grows along the specific contrasts that we choose. Notice also that the fine features are of a magnitude smaller than $0.4$, which is much smaller than half of the operating wavelength, $\pi$.

\begin{figurehere}
\center

        \hfill{}   \hfill{} \hfill{}   \hfill{}  \hfill{}   \hfill{}

        \hfill{}\includegraphics[clip,width=0.3\textwidth]{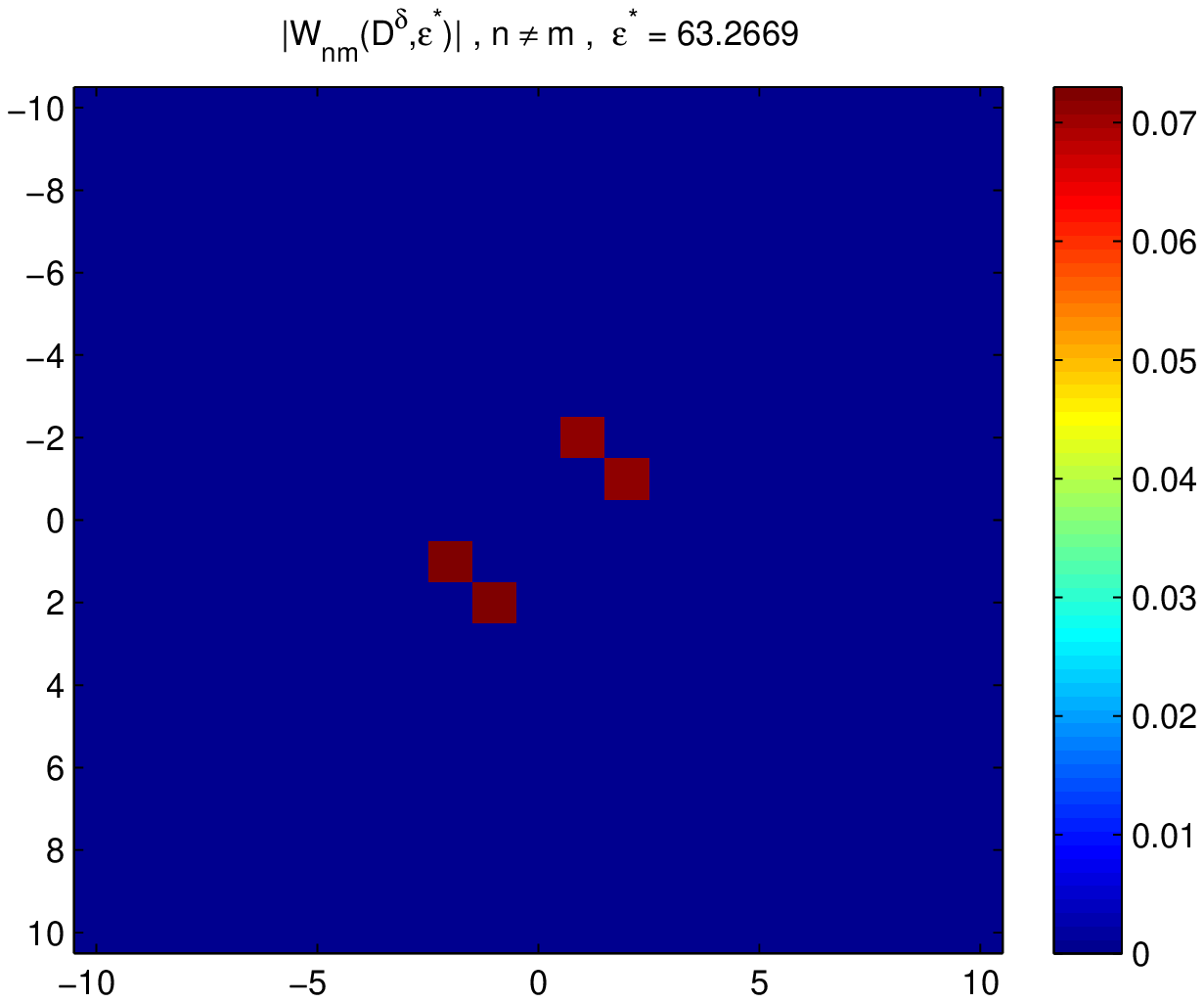}\hfill{}
        \hfill{}\includegraphics[clip,width=0.3\textwidth]{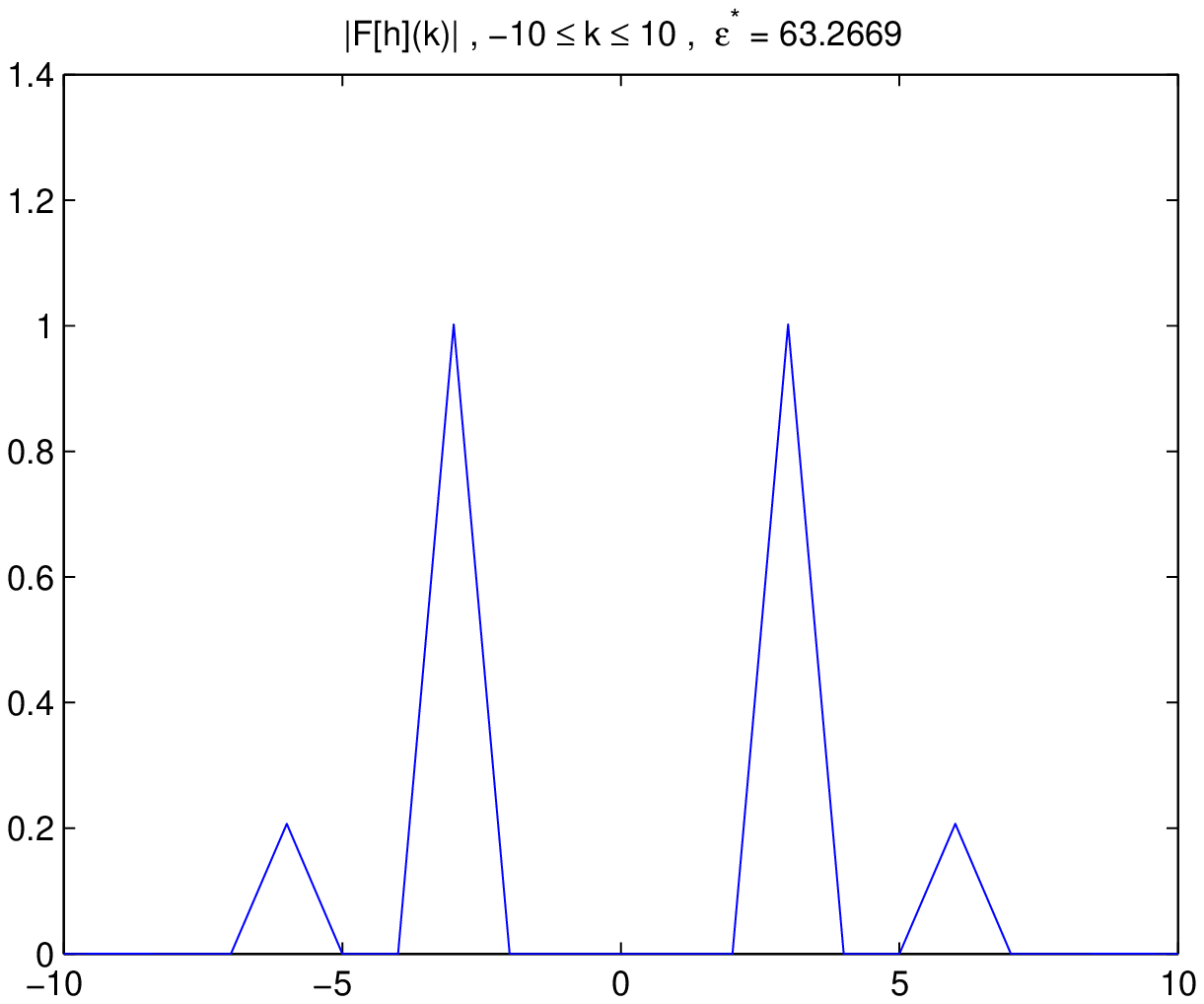}\hfill{}
        \hfill{}\includegraphics[clip,width=0.3\textwidth]{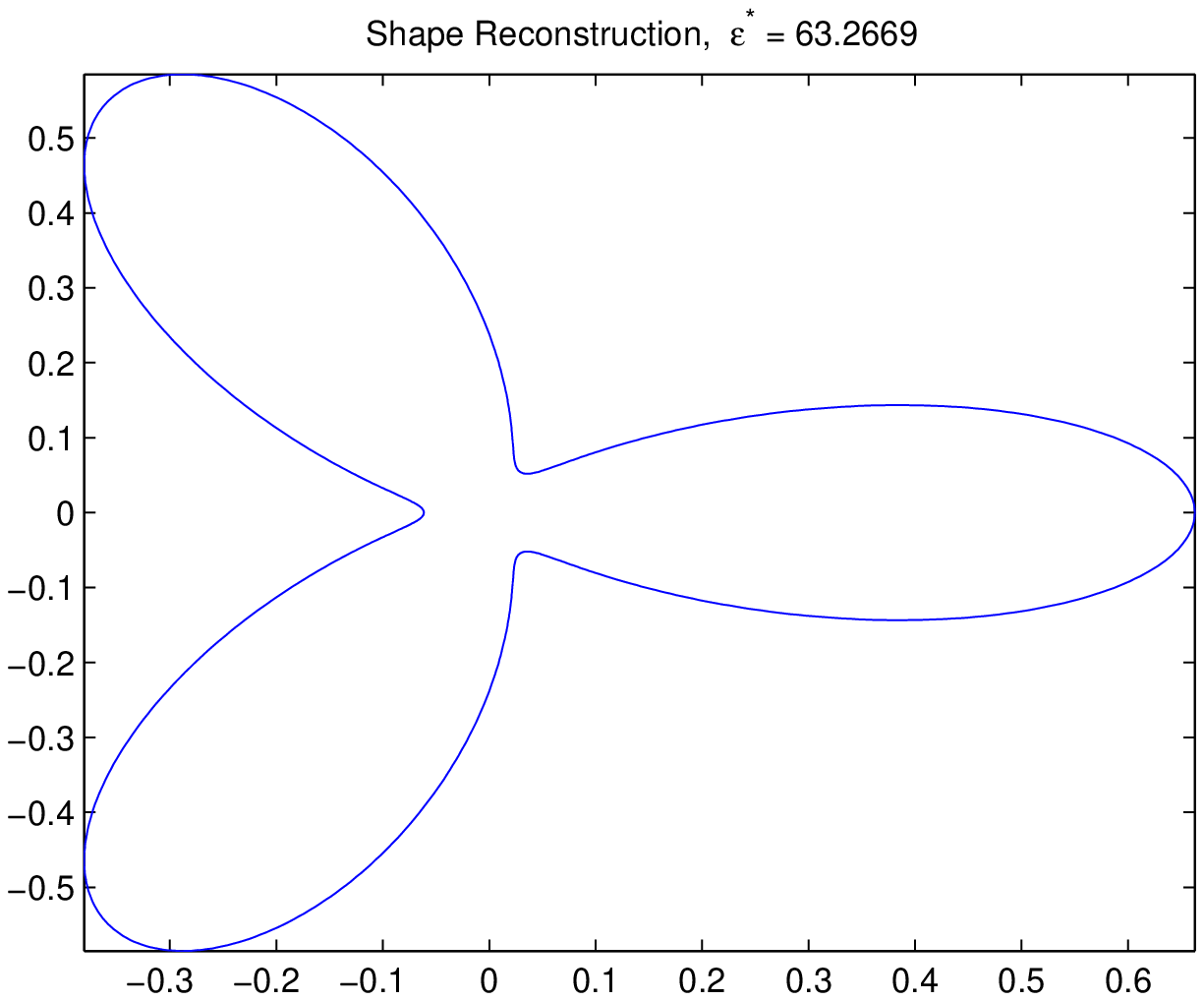}\hfill{}

      %  \hfill{}(3 a)\hfill{} \hfill{}(3 b)\hfill{} \hfill{}(3 c)\hfill{}

         \hfill{}\includegraphics[clip,width=0.3\textwidth]{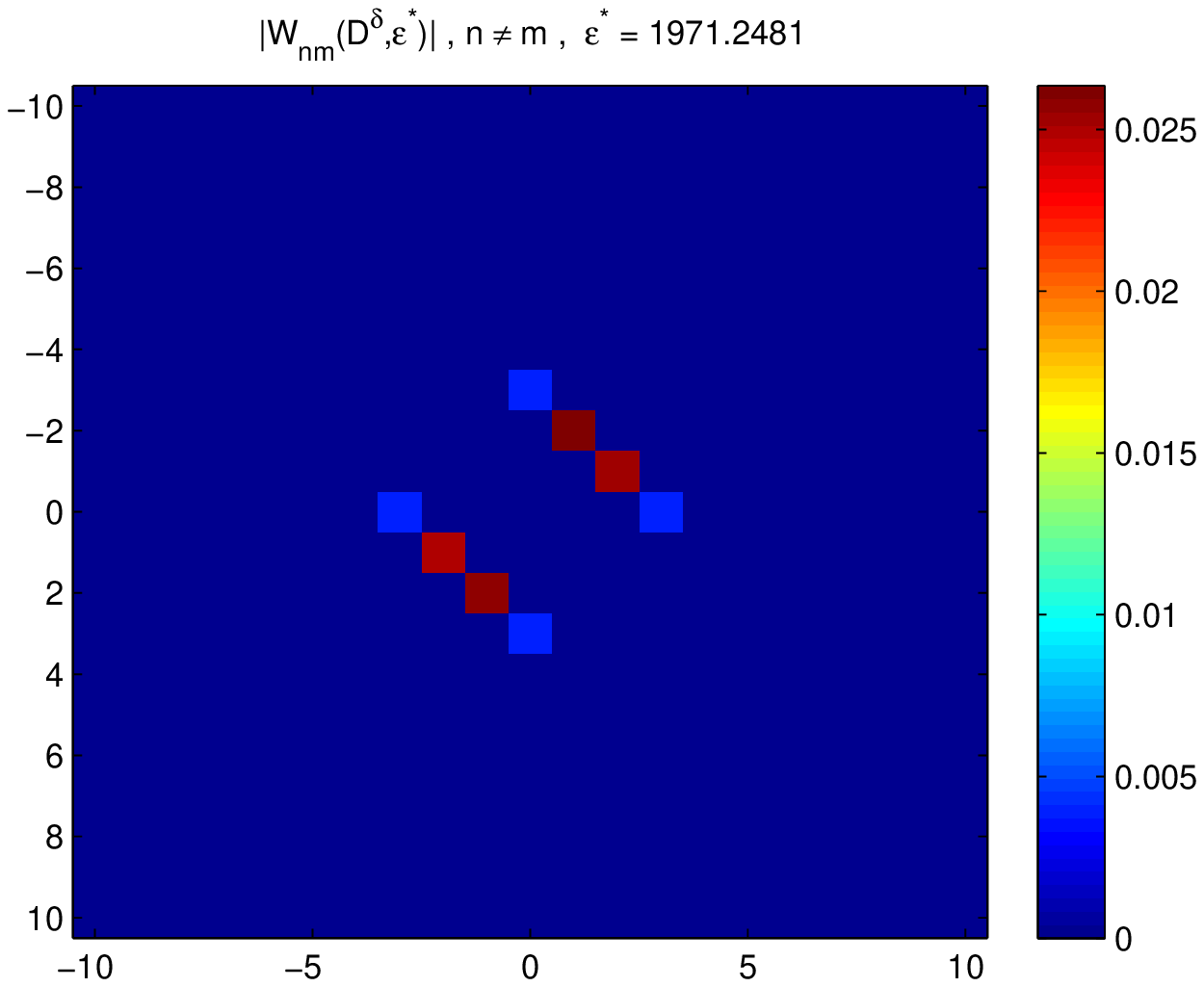}\hfill{}
        \hfill{}\includegraphics[clip,width=0.3\textwidth]{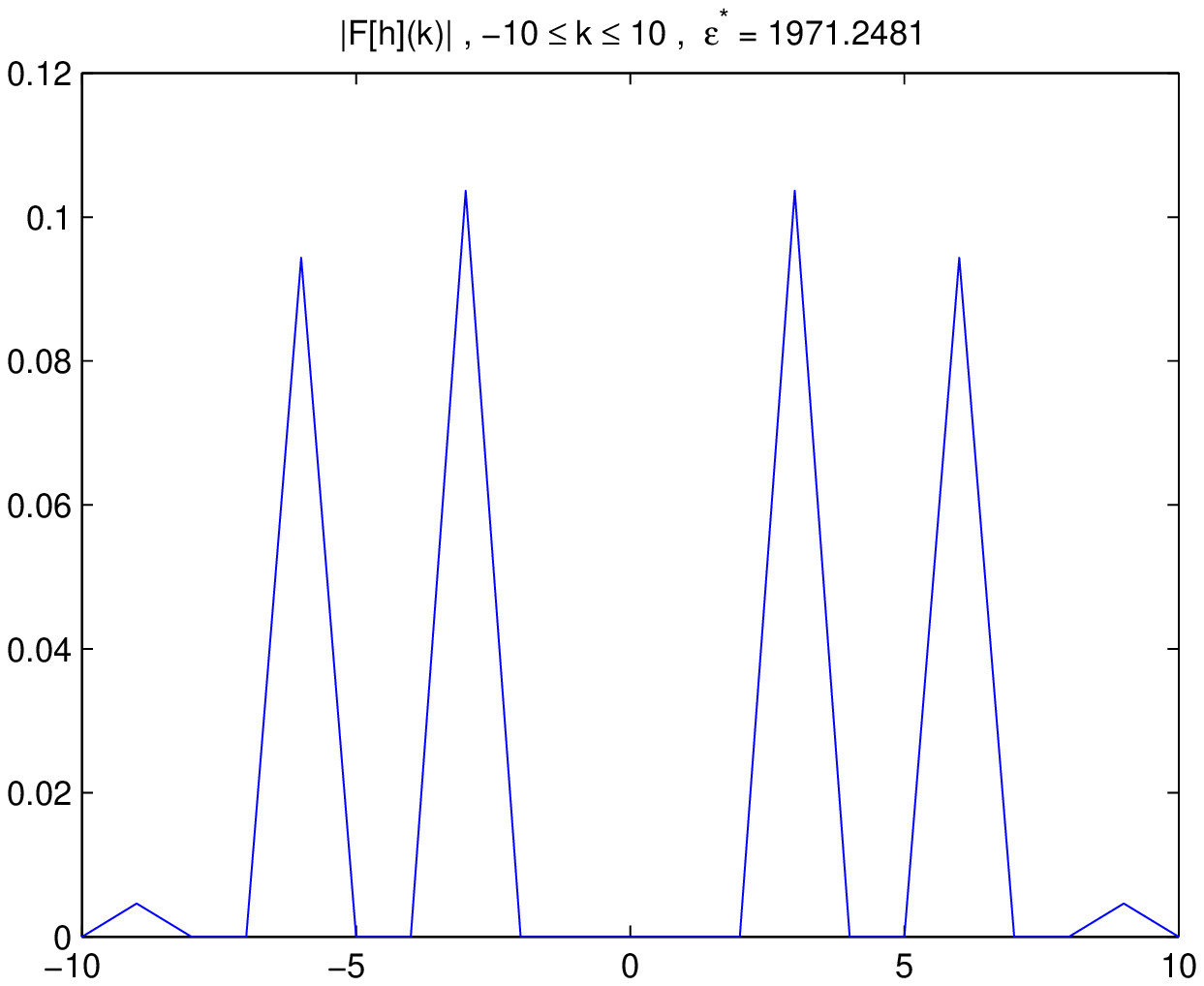}\hfill{}
        \hfill{}\includegraphics[clip,width=0.3\textwidth]{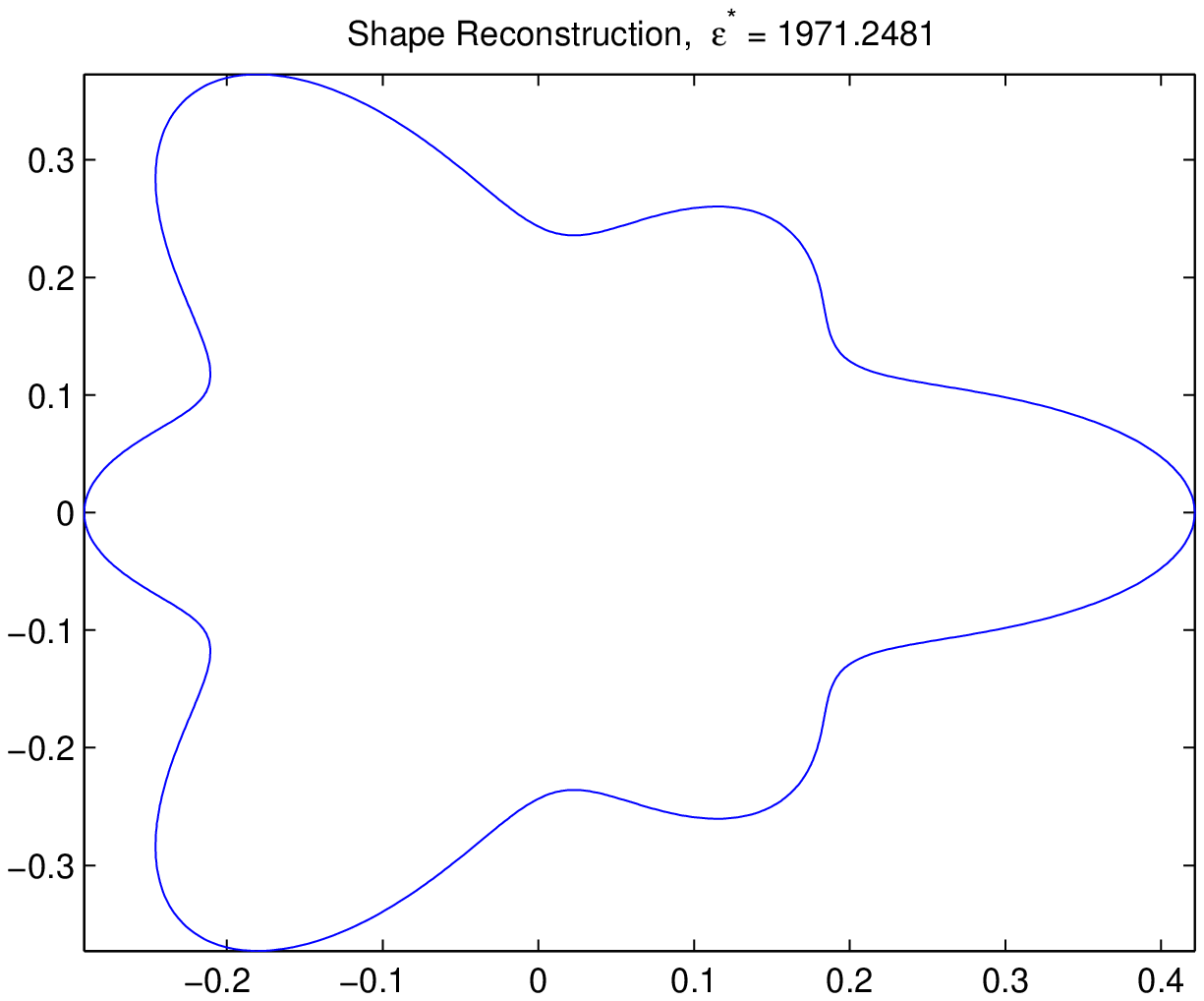}\hfill{}

      %  \hfill{}(4 a)\hfill{} \hfill{}(4 b)\hfill{} \hfill{}(4 c)\hfill{}
%
%         \hfill{}\includegraphics[clip,width=0.3\textwidth]{figures/NEW_SC_1/NEW_SC_1_cond_32}\hfill{}
%        \hfill{}\includegraphics[clip,width=0.3\textwidth]{figures/NEW_SC_1/NEW_SC_1_cond_32_fourier}\hfill{}
%        \hfill{}\includegraphics[clip,width=0.3\textwidth]{figures/NEW_SC_1/NEW_SC_1_cond_32_fourier_shape}\hfill{}

    %    \hfill{}(5 a)\hfill{} \hfill{}(5 b)\hfill{} \hfill{}(5 c)\hfill{}

        \hfill{}\includegraphics[clip,width=0.3\textwidth]{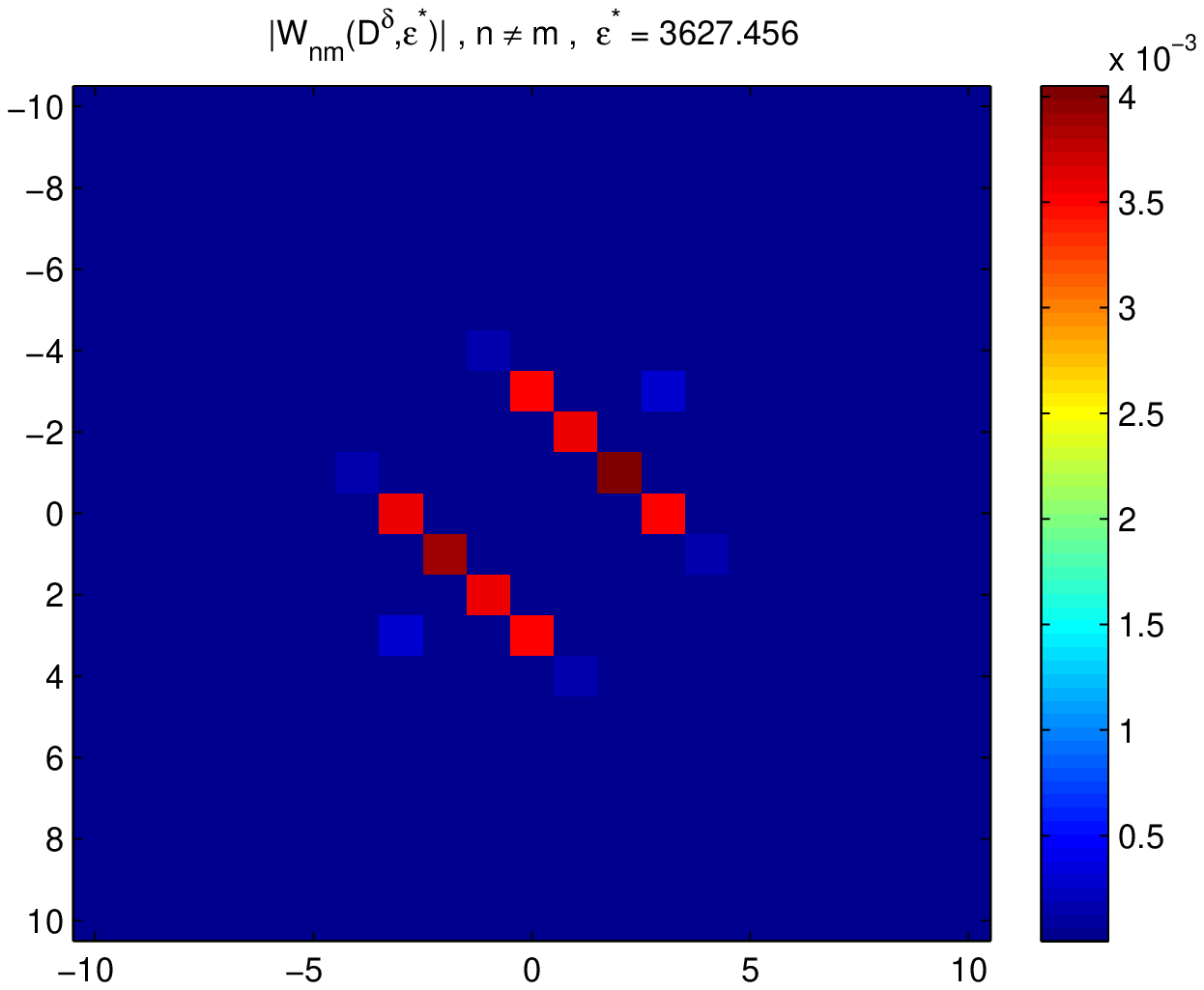}\hfill{}
        \hfill{}\includegraphics[clip,width=0.3\textwidth]{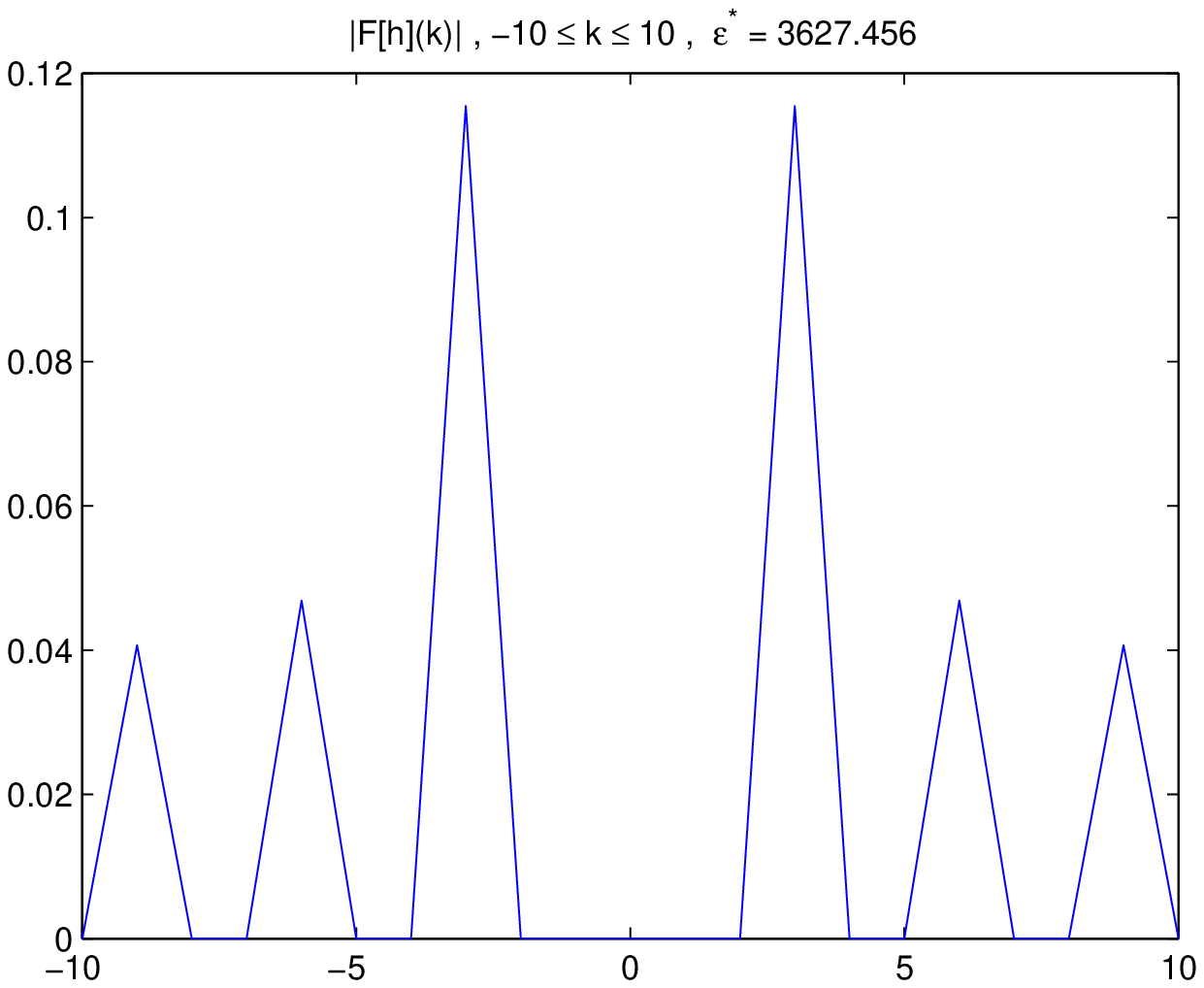}\hfill{}
        \hfill{}\includegraphics[clip,width=0.3\textwidth]{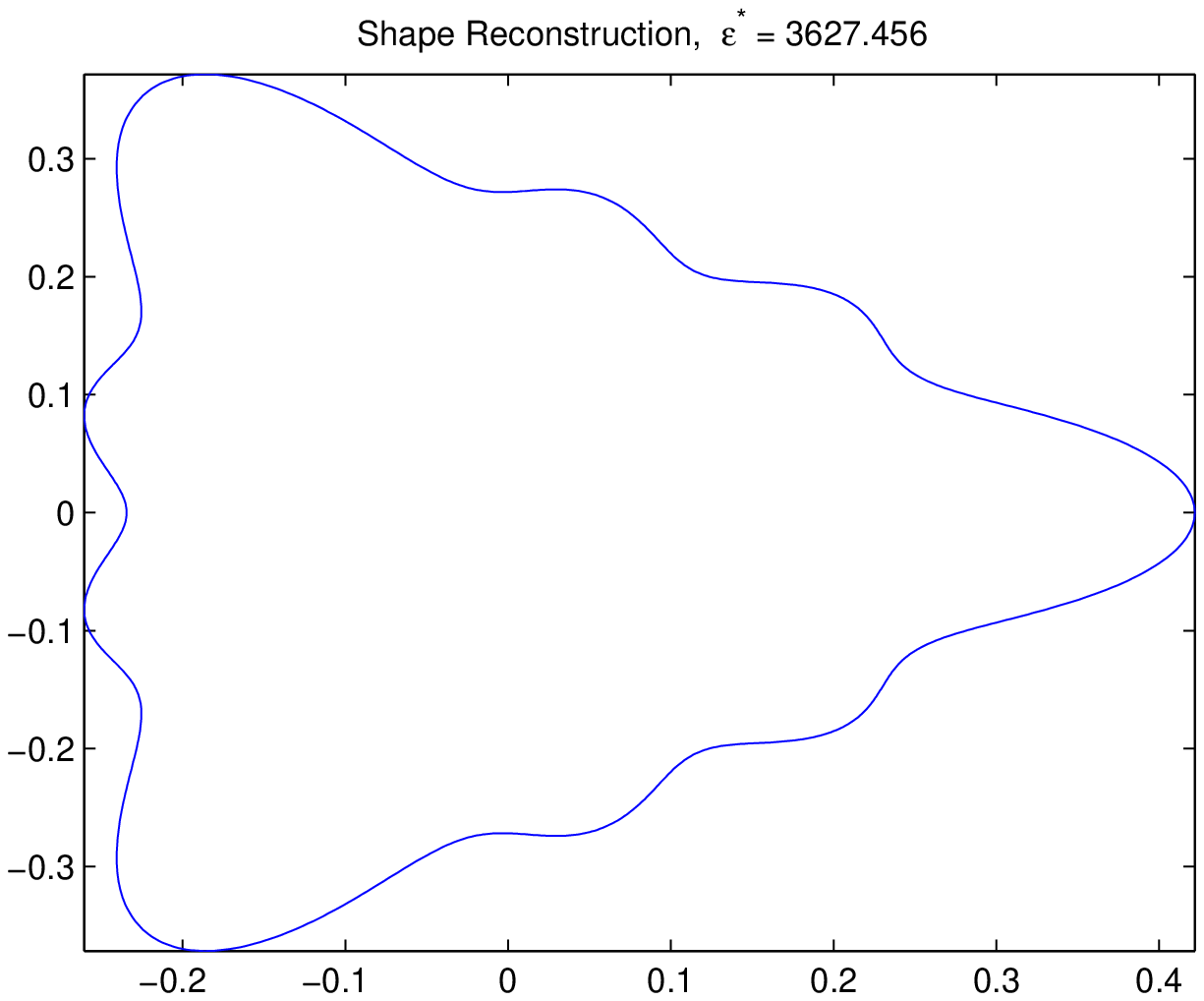}\hfill{}

     %   \hfill{}(5 a)\hfill{} \hfill{}(5 b)\hfill{} \hfill{}(5 c)\hfill{}

\caption{
Illustration of super-resolution in Example 1. Left: magnitude of scattering coefficients; middle: magnitude of recovered Fourier coefficients; right: recovered domain.
}\label{Example 1_recovery}
\end{figurehere}

\textbf{Example 2}
We try the following right-angled isosceles triangle $D^\delta$,
which is a perturbation of the domain $D:=B(0,0.2)$; see Figure \ref{Example 2_domain} (left) for the domain and Figure \ref{Example 2_domain} (right) the comparison between the domains $D^\delta$ and $D$. This case is substantially harder, since the perturbation $h$ consists of many Fourier modes and is no longer smooth.

%$51$ receivers are put to receive the data from the forward problem for the scattering coefficient as the Fourier transform of the far-field data, see Figure \ref{Example 2_domain} (right).

\begin{figurehere}
\center
\includegraphics[clip,width=0.3\textwidth]{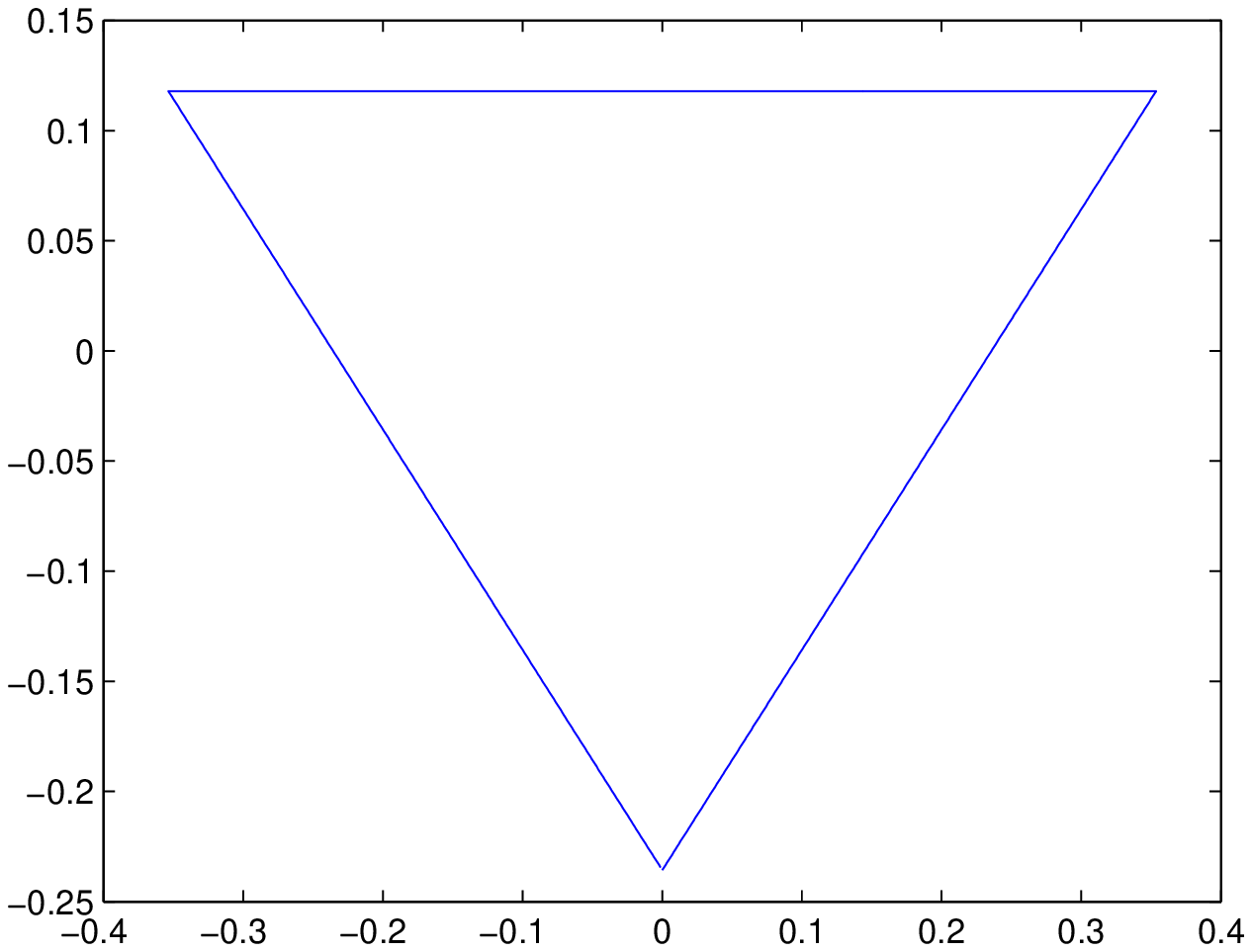}
 \includegraphics[clip,width=0.3\textwidth]{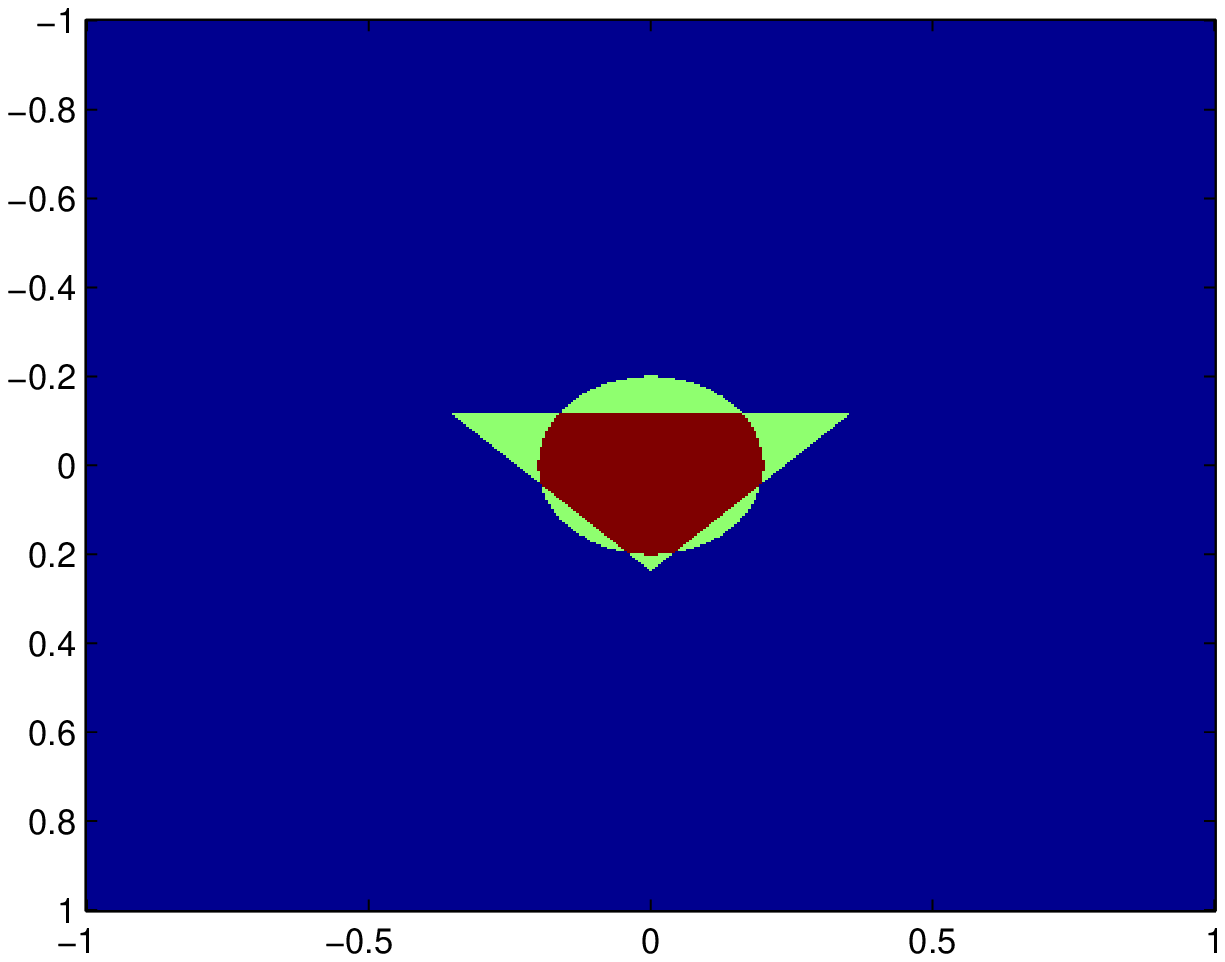}
    %    \hfill{}\includegraphics[clip,width=0.3\textwidth]{figures/NEW_SC_2/NEW_SC_2_rec}\hfill{}
     %   \hfill{}(5 a)\hfill{} \hfill{}(5 b)\hfill{} \hfill{}(5 c)\hfill{}
\caption{
Inclusion shape in Example 2.
}\label{Example 2_domain}
\end{figurehere}

The relative magnitudes of the scattering coefficients $\max_{|m-n| = k }  |W_{nm}(D^{\delta},\epsilon^*)| / \max_{ m \neq n }  |W_{nm}(D^{\delta},\epsilon^*)| $ are plotted for $k = 1,2,\ldots,6,$ in Figure \ref{Example 2_coeff}.
From this figure, we can see that the relative magnitude of the scattering coefficient corresponding to the $\pm k$-th Fourier mode comes out more often when $\varepsilon^*$ becomes large.

\begin{figurehere}
\center

        \hfill{}\includegraphics[clip,width=0.3\textwidth]{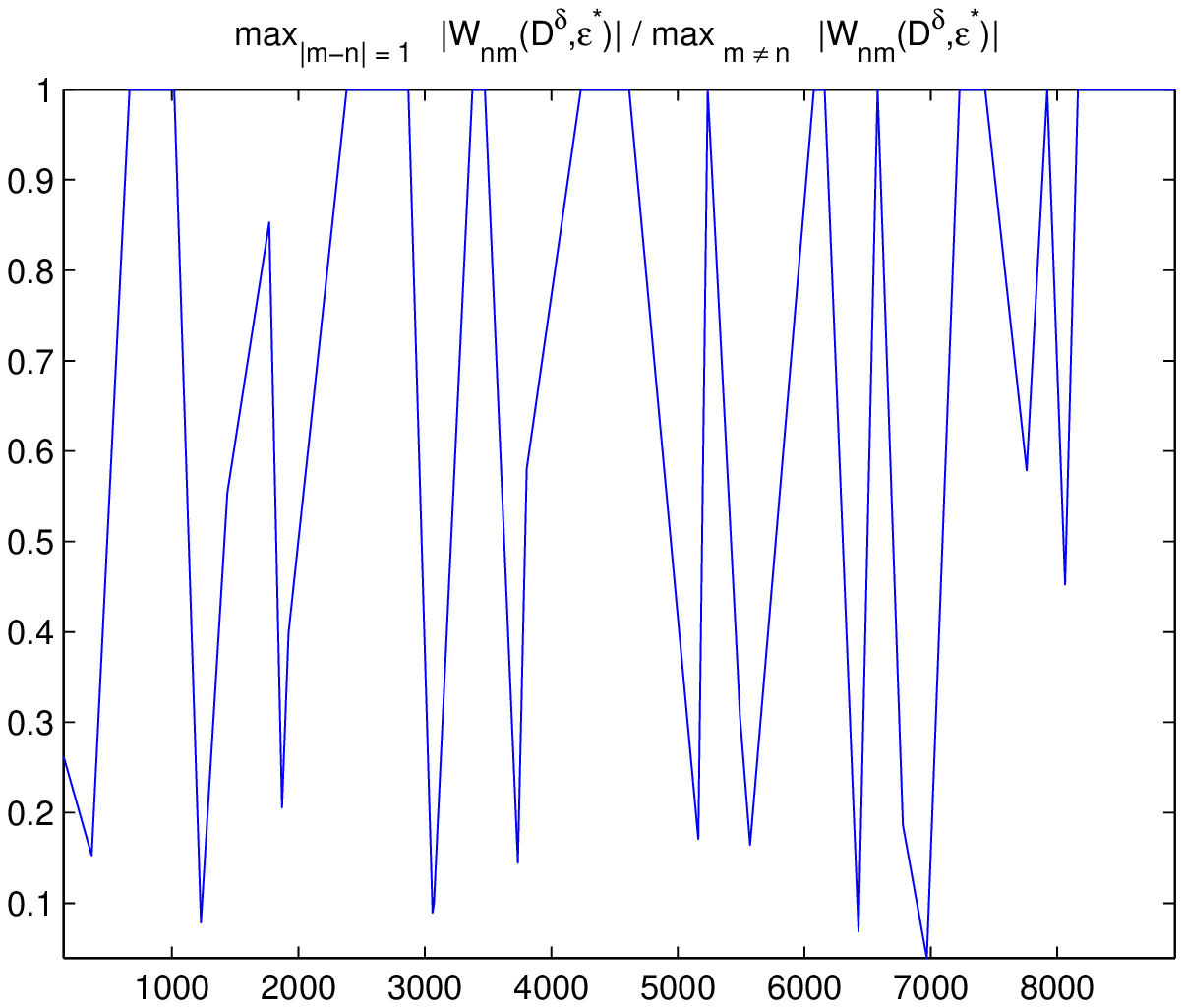}\hfill{}
        \hfill{}\includegraphics[clip,width=0.3\textwidth]{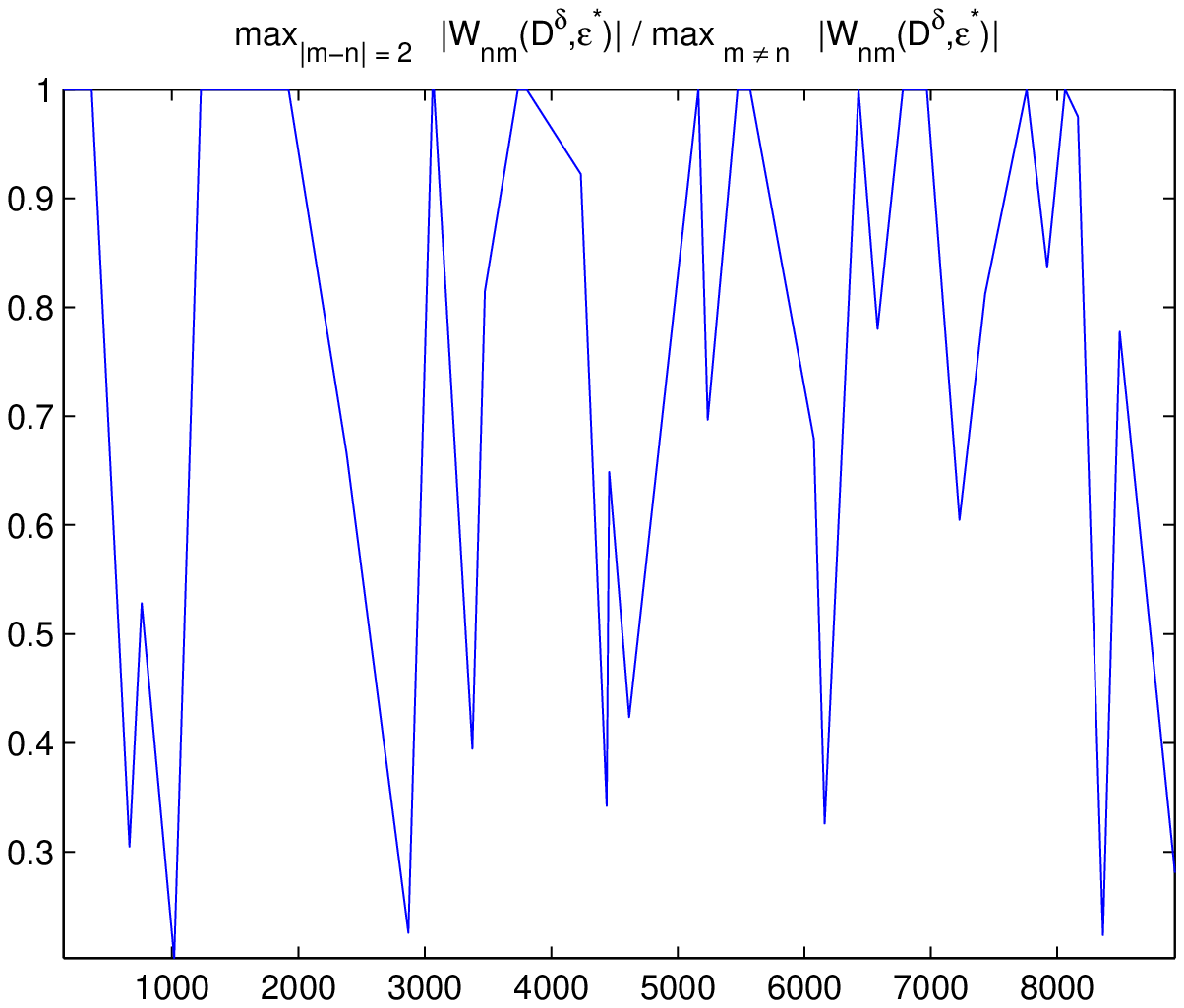}\hfill{}
        \hfill{}\includegraphics[clip,width=0.3\textwidth]{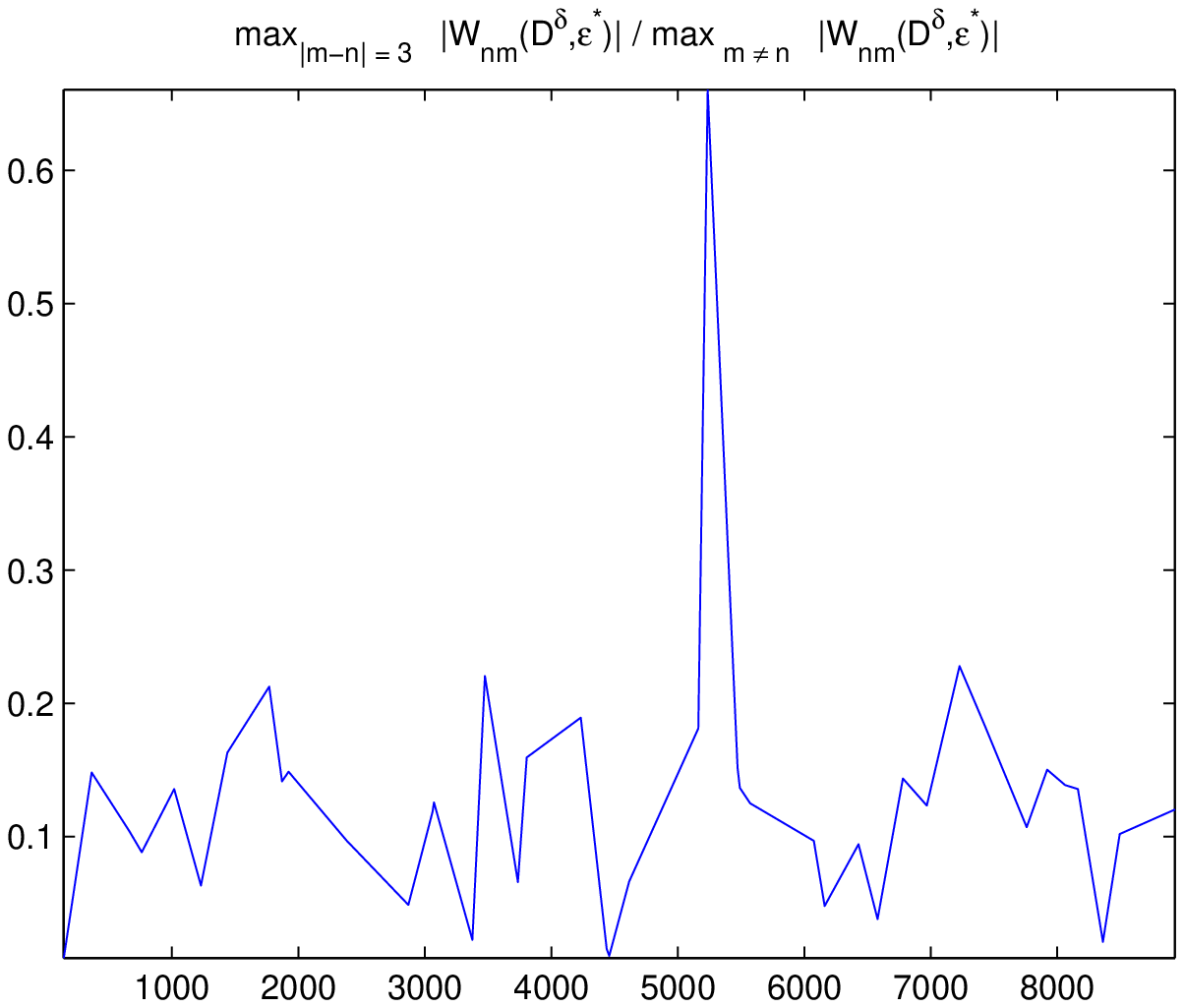}\hfill{}

       % \hfill{}(2 a)\hfill{} \hfill{}(2 b)\hfill{} \hfill{}(2 c)\hfill{}

       \hfill{}\includegraphics[clip,width=0.3\textwidth]{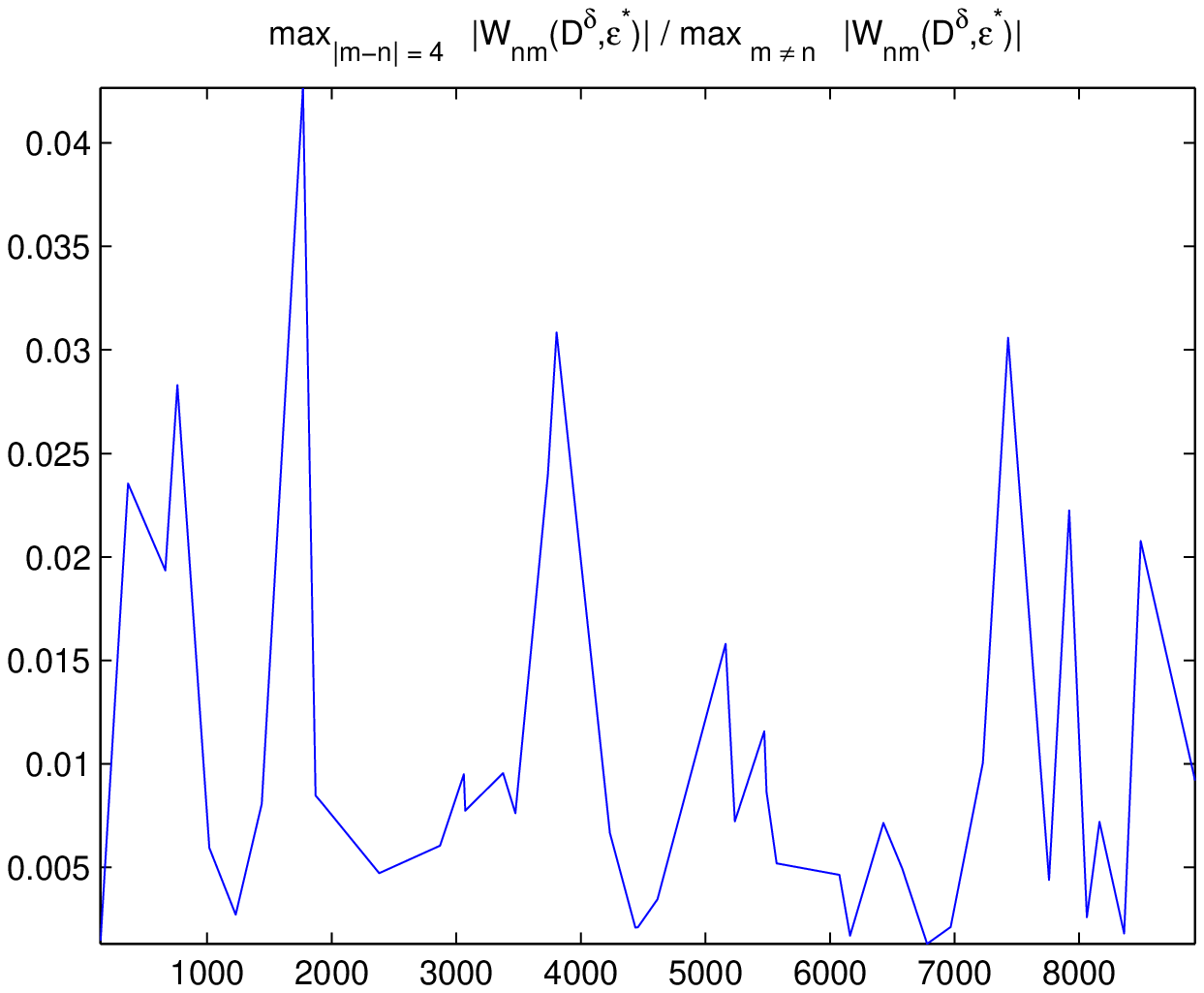}\hfill{}
        \hfill{}\includegraphics[clip,width=0.3\textwidth]{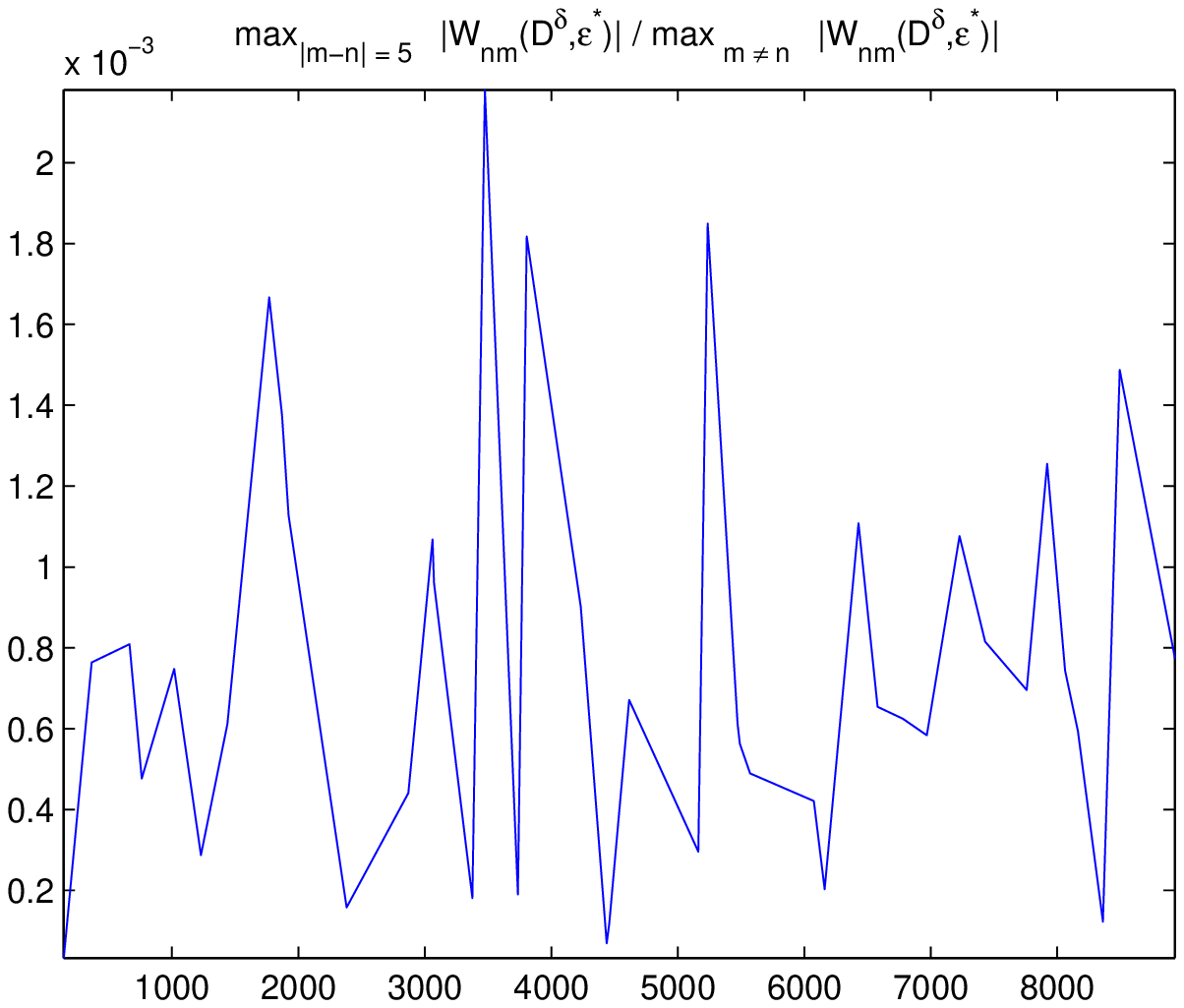}\hfill{}
        \hfill{}\includegraphics[clip,width=0.3\textwidth]{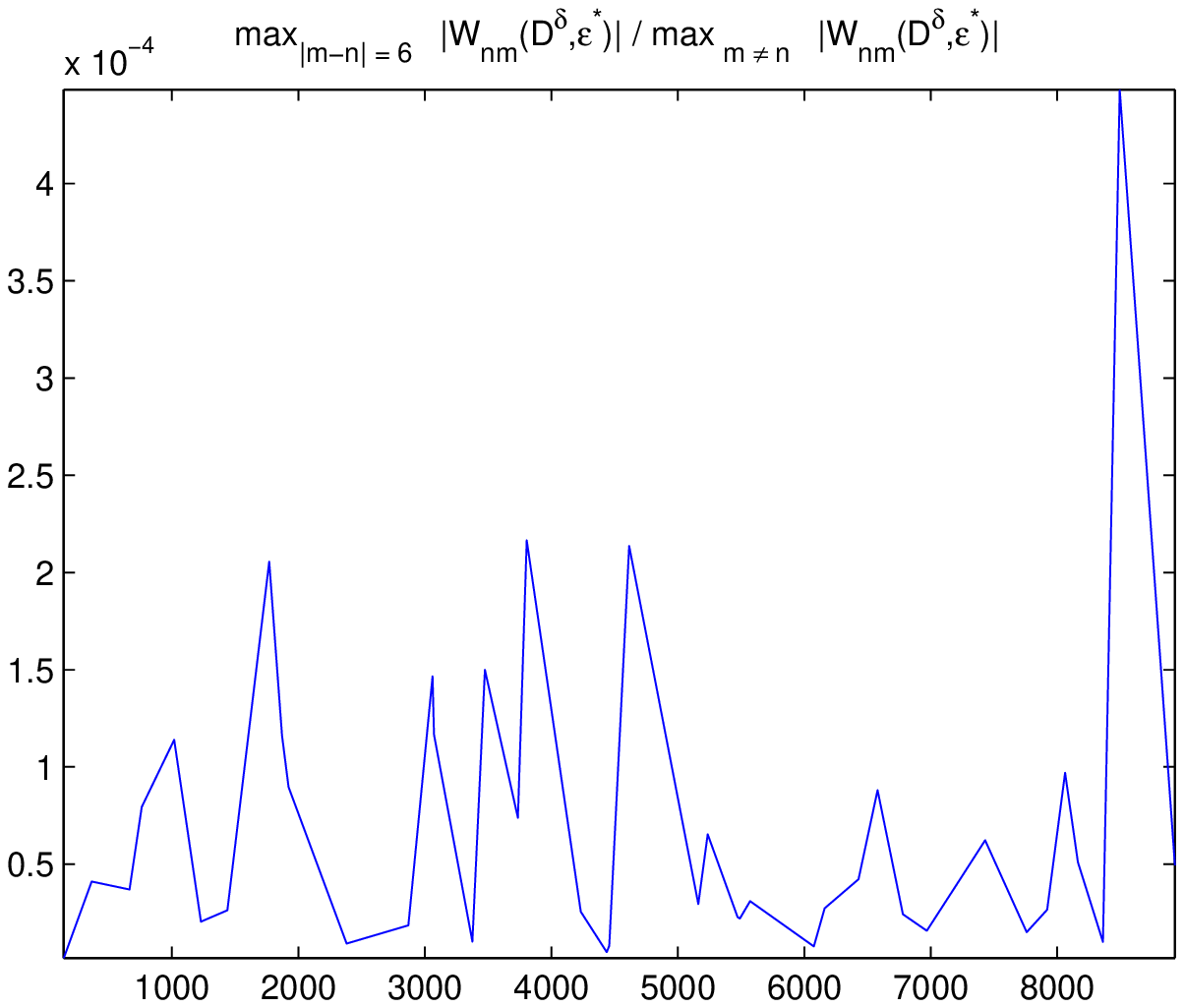}\hfill{}
     %   \hfill{}(5 a)\hfill{} \hfill{}(5 b)\hfill{} \hfill{}(5 c)\hfill{}
\caption{
Relative magnitudes of the scattering coefficients in Example 2.
}\label{Example 2_coeff}
\end{figurehere}

We observe from the relative magnitudes shown in the above figures that the scattering coefficients are best-conditioned for inversion when $\varepsilon^* = 5237.1406$. The scattering coefficients of the respective contrast are then plotted in Figure \ref{Example 2_recovery} (left), together with $\varepsilon^* = 143.6006$ corresponding to the first zero of $J_0$ as a comparison. The aforementioned inversion process is then applied with regularization parameter chosen as $\alpha = 1 \times 10^{-6}$. Figure \ref{Example 2_recovery} (middle) and (right) respectively show the magnitude of the recovered Fourier modes and the reconstructed domains. We can see that the shape obtained from $\varepsilon^* = 5237.1406$ is of a more similar shape to the right-angled triangle. Moreover, the scattering coefficients of $\varepsilon^* = 5237.1406$ are large enough for accurate classification.

\begin{figurehere}
\center

        \hfill{}   \hfill{} \hfill{}   \hfill{}  \hfill{}   \hfill{}

         \hfill{}\includegraphics[clip,width=0.3\textwidth]{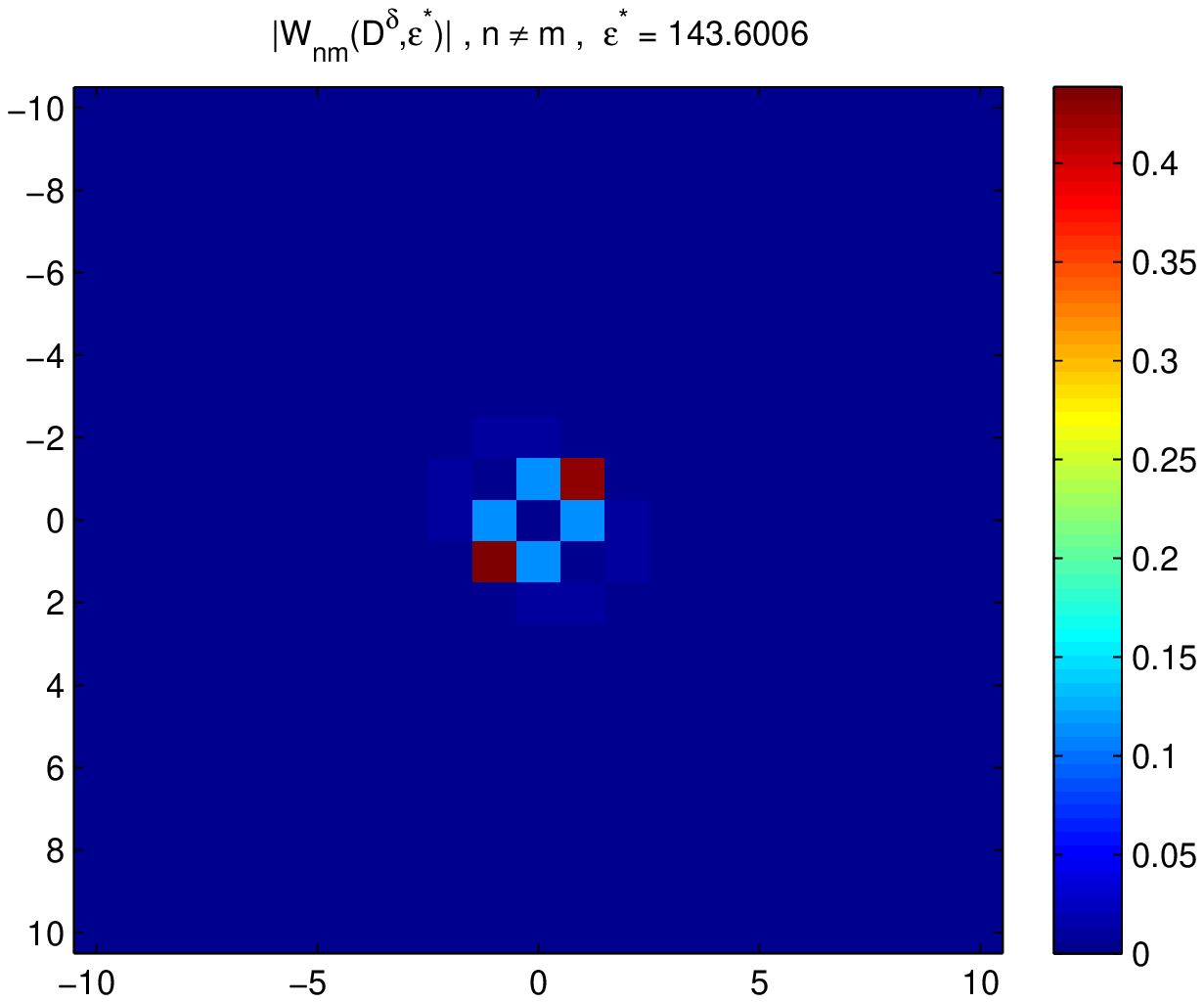}\hfill{}
        \hfill{}\includegraphics[clip,width=0.3\textwidth]{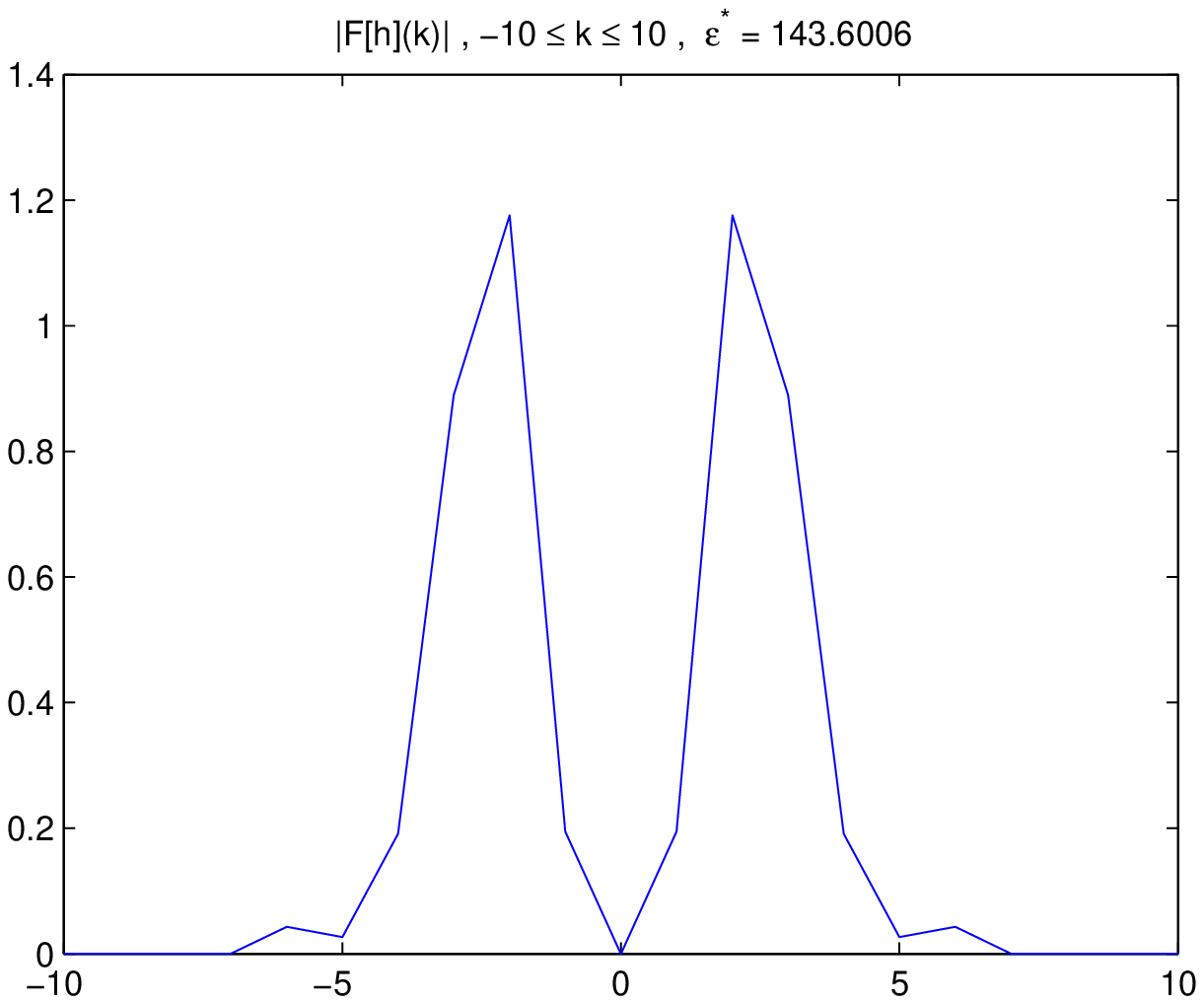}\hfill{}
        \hfill{}\includegraphics[clip,width=0.3\textwidth]{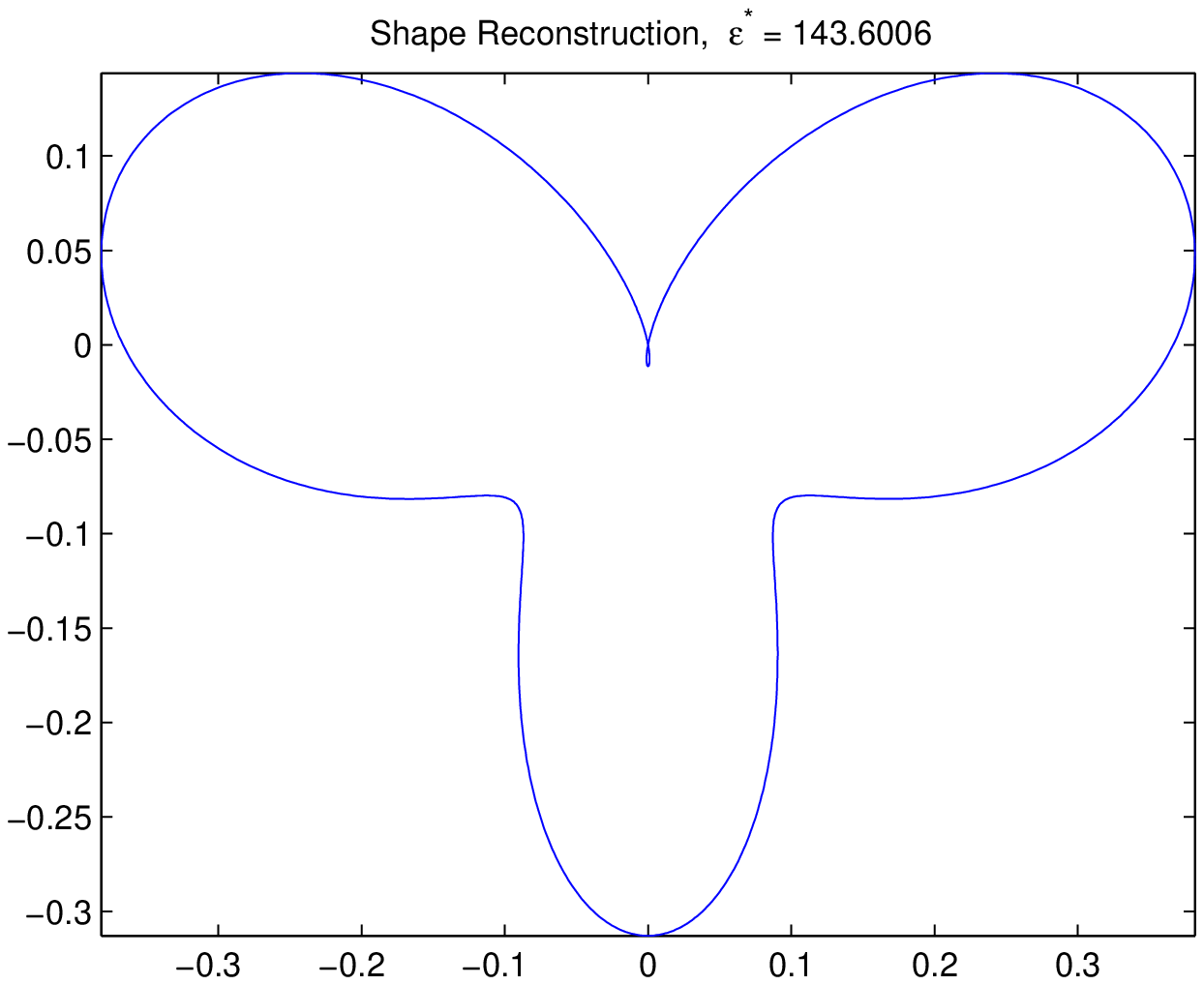}\hfill{}

        %\hfill{}(4 a)\hfill{} \hfill{}(4 b)\hfill{} \hfill{}(4 c)\hfill{}
%
         \hfill{}\includegraphics[clip,width=0.3\textwidth]{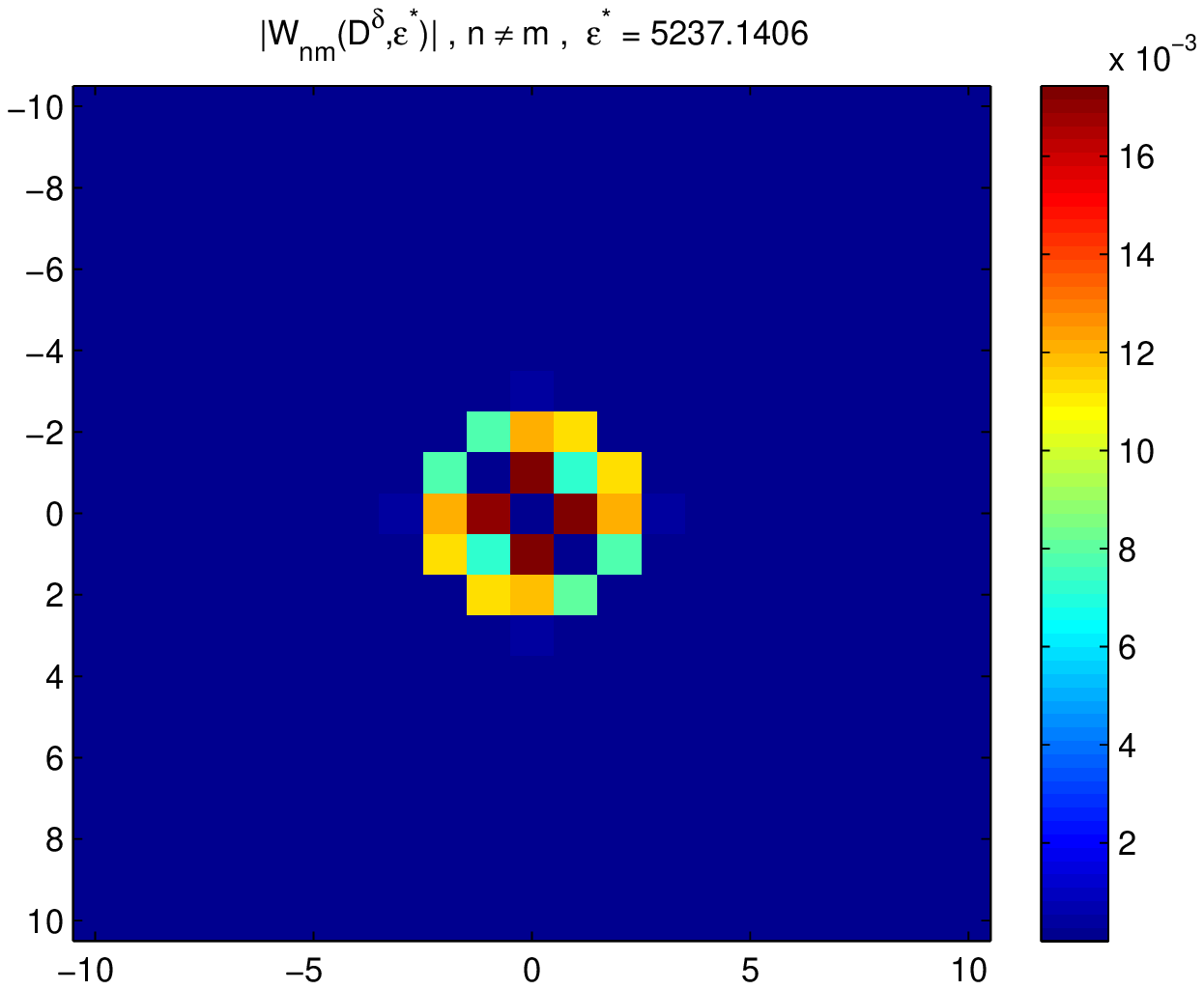}\hfill{}
        \hfill{}\includegraphics[clip,width=0.3\textwidth]{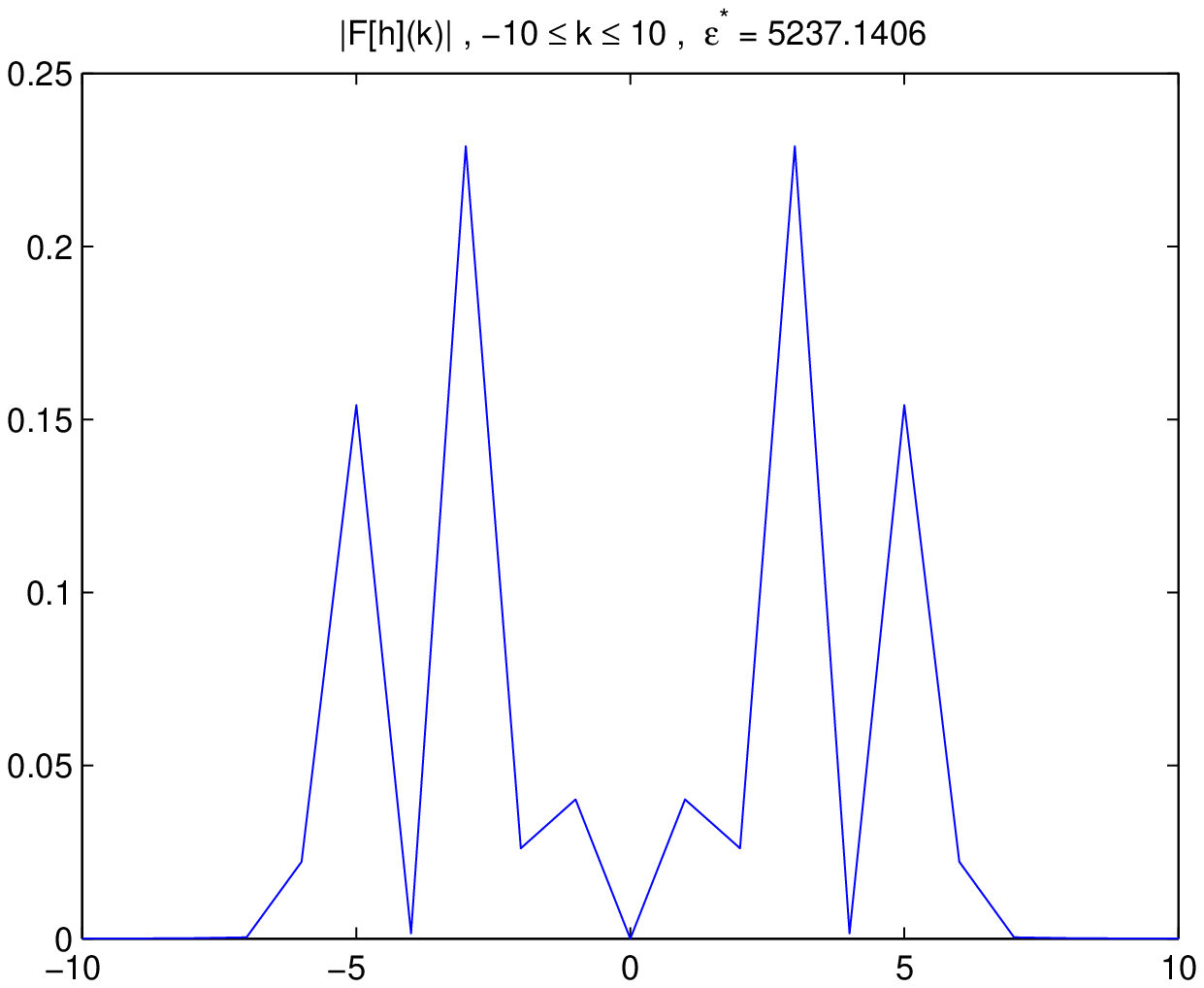}\hfill{}
        \hfill{}\includegraphics[clip,width=0.3\textwidth]{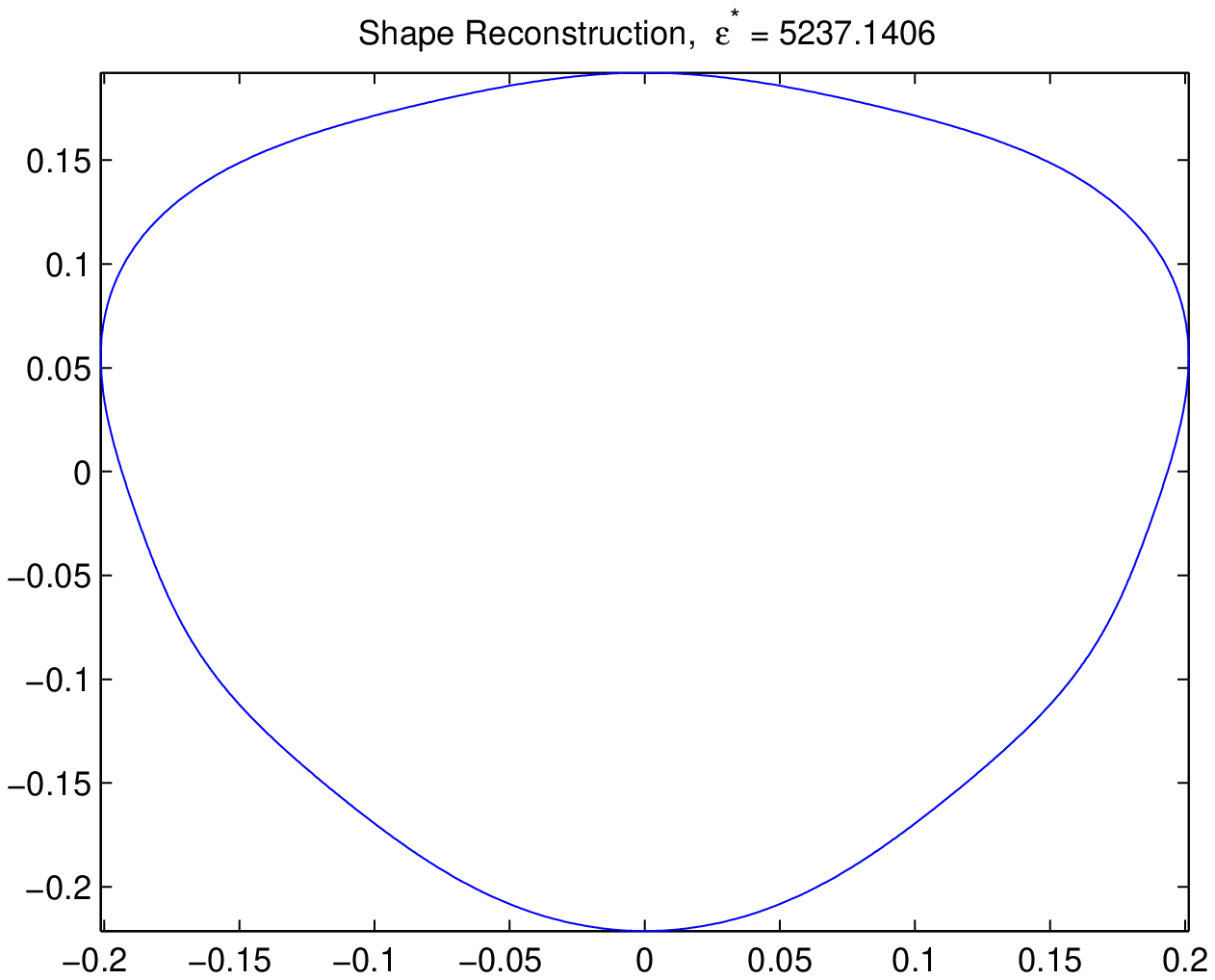}\hfill{}

      %  \hfill{}(5 a)\hfill{} \hfill{}(5 b)\hfill{} \hfill{}(5 c)\hfill{}

%        \hfill{}\includegraphics[clip,width=0.3\textwidth]{figures/NEW_SC_2/NEW_SC_2_cond_41}\hfill{}
%        \hfill{}\includegraphics[clip,width=0.3\textwidth]{figures/NEW_SC_2/NEW_SC_2_cond_41_fourier}\hfill{}
%        \hfill{}\includegraphics[clip,width=0.3\textwidth]{figures/NEW_SC_2/NEW_SC_2_cond_41_fourier_shape}\hfill{}

     %   \hfill{}(5 a)\hfill{} \hfill{}(5 b)\hfill{} \hfill{}(5 c)\hfill{}

\caption{
Illustration of super-resolution in Example 2. Left: magnitude of scattering coefficients; middle: magnitude of recovered Fourier coefficients; right: recovered domain.
}\label{Example 2_recovery}
\end{figurehere}

\section{Concluding remarks} \label{conclusion}

In this paper, we have for the first time established a mathematical theory of super-resolution
in the context of imaging high contrast inclusions.
We have found both analytically and numerically that
at some high resonant values of the contrast, super-resolution in reconstructing the shape of the inclusion can be achieved.

Our approach opens many new avenues for mathematical imaging and focusing in resonant media.
Many challenging problems are still to be solved.  It would be very interesting to generalize our approach in order to justify the fact that super-resolution can be achieved using structured light illuminations \cite{structured, structured2}. Another challenging problem is to generalize our approach to electromagnetic and elastic wave imaging problems of high contrast inclusions. This would be the subject of a forthcoming publication.


\begin{thebibliography}{11}

\bibitem{handbook}
M. Abramowitz and I.A. Stegun, Handbook of Mathematical Functions, 9th edition, Dover Publications, 1970.

%\bibitem{shapedev}
%H. Ammari, Y.T. Chow, Keji Liu and J. Zou \textit{Optimal shape design by partial spectral data}, to appear.

\bibitem{yves} H. Ammari, E. Bonnetier, and Y. Capdeboscq, \textit{Enhanced resolution in structured media}, SIAM J. Appl. Math. 70  (2009/10),  pp.~1428-1452.

\bibitem{heteroscattering}
H. Ammari, Y.T. Chow and J. Zou, \textit{The concept of heterogeneous scattering coefficients and its application in inverse medium scattering}, SIAM J. Math. Anal. 46 (2014), pp. 2905-2935.

\bibitem{lnm} H. Ammari, J. Garnier, W. Jing, H. Kang, M. Lim, K. S\o lna, and H. Wang,
{Mathematical and Statistical Methods for Multistatic Imaging}, Lecture Notes in Mathematics, Volume 2098, Springer, Cham, 2013.


\bibitem{homoscattering}
H. Ammari, H. Kang, H. Lee, and M. Lim, \textit{Enhancement of near-cloaking. Part II: the Helmholtz equation}, Comm. Math. Phys. 317 (2013), pp. 485-502.

\bibitem{han} H. Ammari, M.P. Tran, and  H. Wang,
\textit{Shape identification and classification in echolocation}, SIAM J. Imag. Sci., to appear (arXiv:1308.5625).

\bibitem{hai} H. Ammari and H. Zhang,
\textit{A mathematical theory of super-resolution by using a system of sub-wavelength Helmholtz resonators}, arXiv: 1405.2513.

\bibitem{physicssuperresolution}
S. Arhab, G. Soriano, Y. Ruan, G. Maire, A. Talneau, D. Sentenac, P.C. Chaumet, K. Belkebir, and H. Giovannini, \textit{Nanometric resolution with far-field optical profilometry},
Phys. Rev. Lett. 111 (2013), 053902.

\bibitem{normofresolvent}
O.F. Bandtlow, \textit{Estimates for norms of resolvents and an application to the perturbation of spectra},
Math. Nachr. 267 (2004), pp. 3–11.


\bibitem{gang1} G. Bao and J. Lin, \textit{Near-field imaging of the surface displacement on an infinite ground plan}e, Inverse Probl. Imag.  7 (2013),  pp.~377-396.

\bibitem{gang2} G. Bao and P. Li, \textit{Near-field imaging of infinite rough surfaces}, SIAM J. Appl. Math. 73  (2013),  pp.~2162-2187.


\bibitem{zerozero}
C.L. Beattie, \textit{Table of first 700 zeros of Bessel functions - $J_l(x)$ and $J_l' (x)$}, BSTJ 37 (1958), pp. 689-697.

\bibitem{carleman}
T. Carleman, \textit{Zur theorie der linearen intergralgleichungen}, Math. Z. 9 (1921), pp. 196-217.

\bibitem{spectral}
J.B. Conway, A Course in Functional Analysis, Graduate Texts in Mathematics 96, Springer, 1990.

\bibitem{singularnumber2}
D.E. Edmunds, W.D. Evans and D.J. Harris, \textit{Two-sided estimates of the approximation numbers of certain Volterra integral operators}, Studia Math. 124  (1997), pp. 59-80.

\bibitem{ideal}
I.C. Gohberg and M.G. Krein, Introduction to the Theory of Linear Non-Selfadjoint Operators, Amer. Math. Soc. Transl. 128, Amer. Math. Society, Providence, RI, 1969.

\bibitem{structured} M.G.L. Gustafsson, \textit{Surpassing the lateral resolution limit by a factor of two
using structured illumination microscopy}, J. Microscopy 198 (2000), pp.~82-87.

\bibitem{structured2} M.G.L. Gustafsson, \textit{Nonlinear structured-illumination microscopy: Wide-field fluorescence imaging with theoretically unlimited resolution}, Proc. Nat. Acad. Sci. 102 (2005), pp.~13081-13086.

\bibitem{science} S. W. Hell, Far-field Optical Nanoscopy, Science 316 (2007), pp.~1153-1158.

\bibitem{LFL2011}
F.\,Lemoult, M.\,Fink and G.\,Lerosey,
\textit{Acoustic resonators for far-field control of sound on a subwavelength scale}, Phys. Rev. Lett. 107 (2011), 064301.

\bibitem{lemoult11}
F. Lemoult, G. Lerosey, J. de Rosny, and M. Fink,
\textit{Time reversal in subwavelength-scaled resonant media: beating the diffraction limit},
Int. J. Microwave Sci. Tech. (2011),  425710.

\bibitem{lemoult10}
F. Lemoult, A. Ourir, J. de Rosny, A. Tourin, M. Fink, and G. Lerosey,
\textit{Resonant metalenses for breaking the diffraction barrier},
Phys. Rev. Lett. {104} (2010), 203901.

\bibitem{lerosey07}
G. Lerosey, J. de Rosny, A. Tourin, and M. Fink,
\textit{Focusing beyond the diffraction limit with far-field time reversal},
Science {315} (2007), pp.~1120-1122.


\bibitem{singularnumber}
A. Meskhi, \textit{On the singular numbers for some integral operators}, Revista Matematica Complutense,  vol. XIV, num. 2, (2001), pp. 379-393.

\bibitem{convex}
G. Popov, and G. Vodev, \textit{Resonances near the real axis for transparent obstacles}, Commun. Math. Phys. 207(2), (1999), pp. 411-438.

\bibitem{Watson}
G.N. Watson, {Theory of Bessel Functions}, 2nd edition, Cambridge University Press, 1944.

\bibitem{hansite} \url{http://www.math.ens.fr/~hanwang/software/examples/Helmholtz_R2/demo.html}

\end{thebibliography}
\end{document}